\documentclass[10pt, a4paper]{article}

\usepackage[T1]{fontenc}
\usepackage{amsthm,amsmath,amsfonts,amssymb}
\usepackage{thm-restate} 
\usepackage{stmaryrd} 
\usepackage{cmll}
\usepackage{xspace}
\usepackage[dvipsnames]{xcolor}
\usepackage{multicol}
\usepackage{rotating}
\usepackage{prooftree} 
\usepackage{mathpartir}
\usepackage[misc]{ifsym} 
\usepackage{tikz}
\usetikzlibrary{automata, positioning, arrows}
\tikzset{
        ->,  
        >=stealth, 
        node distance=4.5em, 
        initial text=$ $, 
        }

\makeatletter
\RequirePackage[bookmarks,unicode,colorlinks=true]{hyperref}%
   \def\@citecolor{blue}%
   \def\@urlcolor{blue}%
   \def\@linkcolor{blue}%

\def\orcidID#1{\smash{\href{http://orcid.org/#1}{\protect\raisebox{-1.25pt}{\protect\includegraphics{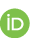}}}}}
\makeatother

\usepackage[capitalize]{cleveref}

\newtheorem{theorem}{Theorem}

\newtheorem{lemma}{Lemma}
\newtheorem{proposition}{Proposition}
\newtheorem{claim}{Claim}
\theoremstyle{definition}
\newtheorem{definition}{Definition}
\newtheorem{example}{Example}

\newcommand{\typec}[1]{{\color{RoyalBlue}{#1}}} 

\newcommand{\parc}[1]{\typec{#1}} 
\newcommand{\blk}[1]{{\color{black}{#1}}}

\newcommand{\grmeq}{\; ::= \;}
\newcommand{\grmor}{\; \mid \;}

\newcommand{\keyword}[1]{\ensuremath{\mathsf{#1}}\xspace}
\newcommand{\tkeyword}[1]{\keyword{\typec{#1}}}

\newcommand{\pakeyword}[1]{\keyword{\typec{#1}}}
\newcommand{\identifier}[1]{\ensuremath{\mathrm{#1}}\xspace}
\newcommand{\Label}[1]{\identifier{#1}}	

\newcommand{\quitl}{\Label{quit}}
\newcommand{\gol}{\Label{go}}

\newcommand{\incl}{\Label{inc}}
\newcommand{\dumpl}{\Label{dump}}
\newcommand{\pushal}{\Label{addOut}}
\newcommand{\pushbl}{\Label{addIn}}

\newcommand{\popl}{\Label{pop}}

\newcommand{\leafl}{\Label{leaf}}
\newcommand{\nodel}{\Label{node}}
\newcommand{\leftl}{\Label{l}}
\newcommand{\rightl}{\Label{r}}
\newcommand{\al}{\Label{a}}
\newcommand{\bl}{\Label{b}}

\newcommand{\typename}[1]{\typec{T_{\mathrm{#1}}}}
\newcommand{\tloop}{\typename{loop}}
\newcommand{\tcounter}{\typename{counter}}
\newcommand{\tstack}{\typename{meta}}
\newcommand{\titer}{\typename{iter}}
\newcommand{\ttree}{\typename{tree}}
\newcommand{\tnested}{\typename{nest}}
\newcommand{\tkorhop}{\typename{KH}}
\newcommand{\Xin}{\typec{X_{\mathrm{in}}}}
\newcommand{\Xout}{\typec{X_{\mathrm{out}}}}

\newcommand{\operator}[1]{\operatorname{#1}}

\newcommand{\dom}{\operator{dom}} 
\newcommand{\treeof}{\operator{treeof}}

\newcommand{\tree}{\treeof}
\newcommand{\Path}{\operator{path}}

\newcommand{\Eq}{\doteq} 
\newcommand{\Empty}{\varepsilon} 

\newcommand{\teq}{\simeq} 

\newcommand{\Oplus}{\hspace{-.1ex}\oplus\hspace{-.1ex}} 
\newcommand{\pair}[2]{\langle{#1},{#2}\rangle}

\newcommand{\zero}{\pakeyword z}
\newcommand{\zerol}{\keyword{z}}

\newcommand{\succc}{\pakeyword s} 
\newcommand{\succl}{\keyword{s}}

\renewcommand{\succ}[1]{\succc\,#1}
\newcommand{\dl}{\keyword d}
\newcommand{\cl}{\keyword c}

\newcommand{\sysvariant}[1]{\mathsf{#1}}
\newcommand{\sysinfty}{\sysvariant \infty}
\newcommand{\sysf}{\sysvariant f}
\newcommand{\sysr}{\sysvariant r}
\newcommand{\sysone}{\sysvariant 1}
\newcommand{\sysp}{\sysvariant p}
\newcommand{\systwo}{\sysvariant 2}
\newcommand{\syscf}{\sysvariant c}
\newcommand{\sysn}{\sysvariant n}

\newcommand{\setname}[1]{\mathbb{#1}}
\newcommand{\types}{\setname{T}} 
\newcommand{\typesi}{\setname{T}_\sysinfty}
\newcommand{\typesf}{\setname{T}_\sysf}
\newcommand{\typesr}{\setname{T}_\sysr}
\newcommand{\typeso}{\setname{T}_\sysone}
\newcommand{\typesp}{\setname{T}_\sysp}
\newcommand{\typest}{\setname{T}_\systwo}
\newcommand{\typescf}{\setname{T}_\syscf}
\newcommand{\typesn}{\setname{T}_\sysn}

\newcommand{\labels}{\setname{L}}

\newcommand{\nbb}{\mathbb{N}}

\newcommand{\E}{\mathcal{E}}

\newcommand{\lcal}{\mathcal{L}}

\newcommand{\ocal}{\mathcal{O}}

\newcommand{\T}{\mathcal{T}}

\newcommand{\X}{\mathcal X} 
\newcommand{\xcal}{\mathcal{X}} 

\newcommand{\dual}{\overline}

\newcommand{\foralllabel}[2]{(\forall{#1}\in{#2})}
\newcommand{\foralllinL}{\foralllabel \ell L}

\newcommand{\declrel}[2]{#1\hfill\fbox{{#2}}}

\newcommand{\inductive}{(\emph{ind.})\xspace}
\newcommand{\coinductive}{(\emph{coind.})\xspace}

\newcommand{\End}{\tkeyword{end}}
\newcommand{\Skip}{\tkeyword{skip}}
\newcommand{\INn}[1]{\typec{?{#1}}}		
\newcommand{\OUTn}[1]{\typec{!{#1}}}	
\newcommand{\MSGn}[1]{\typec{\sharp\,{#1}}}	
\newcommand{\IN}[2]{\typec{\INn{#1}.{#2}}}
\newcommand{\OUT}[2]{\typec{\OUTn{#1}.{#2}}}
\newcommand{\MSG}[2]{\typec{\MSGn{#1}.{#2}}}
\newcommand{\reck}{\tkeyword{rec}}
\newcommand{\REC}[2]{\typec{\reck\,{#1}.{#2}}}
\newcommand{\recordf}[3]{\{{#1}\colon {#2}\}_{{#1}\in{#3}}} 
\newcommand{\record}[3]{\recordf{#1}{#2_#1}{#3}}
\newcommand{\intchoice}{\typec\Oplus} 
\newcommand{\extchoice}{\typec\&}	
\newcommand{\choice}{\typec\star}
\newcommand{\semit}[2]{\typec{{#1};{#2}}}
\newcommand{\bool}{\tkeyword{bool}}
\newcommand{\Int}{\tkeyword{int}}

\newcommand{\inty}{\IN TU} 
\newcommand{\outt}{\OUT TU}
\newcommand{\recordt}[3]{\typec{\record{#1}{#2}{#3}}}
\newcommand{\extct}{\extchoice\recordt \ell T L}
\newcommand{\intct}{\intchoice\recordt \ell T L}
\newcommand{\choicet}{\choice\recordt \ell T L}

\newcommand{\TT}{\typec{T}}
\newcommand{\UT}{\typec{U}}
\newcommand{\VT}{\typec{V}}
\newcommand{\WT}{\typec{W}}
\newcommand{\XT}{\typec X}
\newcommand{\YT}{\typec Y}

\newcommand{\Endl}{\keyword{end}}
\newcommand{\errorl}{\keyword{error}}
\newcommand{\badl}{\keyword{bad}}

\newcommand{\CALL}[2]{{#1}\langle{#2}\rangle}
\newcommand{\CALLT}[2]{\typec{#1}\typec\langle\parc{#2}\typec\rangle}
\newcommand{\CALLNT}[2]{\typec{#1}\typec\langle\typec{#2}\typec\rangle} 

\newcommand{\CALLTT}[3]{\typec{#1}\typec\langle\parc{#2}\typec,\parc{#3}\typec\rangle} 

\newcommand{\subs}[2]{\blk[{\typec{#1}}\blk/{\typec{#2}}\blk]} 

\newcommand{\axiom}[1]{{#1}}

\newcommand{\judgementlabel}[1]{\mathrm{#1}} 

\newcommand{\finlabel}[1]{\judgementlabel{{#1}_\sysf}} 
\newcommand{\reclabel}[1]{\judgementlabel{{#1}_\sysr}} 
\newcommand{\onelabel}[1]{\judgementlabel{{#1}_\sysone}} 
\newcommand{\pushtlabel}[1]{\judgementlabel{{#1}_\sysp}} 
\newcommand{\twolabel}[1]{\judgementlabel{{#1}_\systwo}} 
\newcommand{\inflabel}[1]{\judgementlabel{{#1}_\sysinfty}} 
\newcommand{\conlabel}[1]{\judgementlabel{{#1}_\syscf}} 
\newcommand{\neslabel}[1]{\judgementlabel{{#1}_\sysn}} 

\newcommand{\judgement}[2]{{#1} \: \judgementlabel{#2}}

\newcommand{\judgementrel}[3]{{#1} \; {#2} \; {#3}}

\newcommand{\isType}[2]{\judgement{\typec{#1}}{#2}} 
\newcommand{\istype}[1]{\isType{#1}{type}} 
\newcommand{\istypef}[1]{\isType{#1}{\finlabel{type}}} 
\newcommand{\istyper}[1]{\isType{#1}{\reclabel{type}}} 
\newcommand{\istypeo}[1]{\isType{#1}{\onelabel{type}}} 
\newcommand{\istyped}[1]{\isType{#1}{\pushtlabel{type}}} 
\newcommand{\istypet}[1]{\isType{#1}{\twolabel{type}}} 
\newcommand{\istypei}[1]{\isType{#1}{\inflabel{type}}} 
\newcommand{\istypec}[1]{\isType{#1}{\conlabel{type}}} 
\newcommand{\istypen}[1]{\isType{#1}{\neslabel{type}}} 

\newcommand{\isnat}[1]{\judgement{\parc{#1}}{nat}}

\newcommand{\isdone}[1]{\judgement{\typec{#1}}{\checkmark}}
\newcommand{\isnotdone}[1]{\judgement{\typec{#1}}{\not\!\checkmark}}

\newcommand{\isDual}[3]{\judgementrel{\typec{#1}}{#2}{\typec{#3}}}
\newcommand{\isdual}[2]{\isDual{#1}{\bot}{#2}}
\newcommand{\isdualr}[2]{\isDual{#1}{\reclabel{\bot}}{#2}}
\newcommand{\isdualo}[2]{\isDual{#1}{\onelabel{\bot}}{#2}}
\newcommand{\isduald}[2]{\isDual{#1}{\pushtlabel{\bot}}{#2}}
\newcommand{\isdualt}[2]{\isDual{#1}{\twolabel{\bot}}{#2}}

\newcommand{\isEquiv}[3]{\judgementrel{\typec{#1}}{#2}{\typec{#3}}}
\newcommand{\isequiv}[2]{\isEquiv{#1}{\teq}{#2}}

\newcommand{\isequivr}[2]{\isEquiv{#1}{\reclabel{\teq}}{#2}}
\newcommand{\isequivo}[2]{\isEquiv{#1}{\onelabel{\teq}}{#2}}
\newcommand{\isequivd}[2]{\isEquiv{#1}{\pushtlabel{\teq}}{#2}}
\newcommand{\isequivt}[2]{\isEquiv{#1}{\twolabel{\teq}}{#2}}

\newcommand{\issequiv}[2]{\isequiv{#1}{#2}}

\newcommand{\embeds}[2]{\judgementrel{\typec{#1}}{\hookrightarrow}{\typec{#2}}}

\newcommand{\iseqt}[2]{\judgementrel{\typec{#1}}{\Eq}{\typec{#2}}} 
\newcommand{\iseq}[2]{\judgementrel{#1}{\Eq}{#2}} 
\newcommand{\lhs}[2]{#1\langle{\parc{#2}}\rangle}
\newcommand{\lhst}[2]{\typec{#1}\typec\langle{\parc{#2}}\typec\rangle}

\newcommand{\lhstt}[3]{\lhst{\typec{#1}}{{#2}\typec{,}{#3}}} 
\newcommand{\ispeqt}[3]{\iseqt{\lhs{#1}{#2}}{#3}} 

\newcommand{\iscontrt}[1]{\judgement{\typec{#1}}{contr}} 
\newcommand{\iscontrd}[1]{\judgement{#1}{\pushtlabel{contr}}} 

\ifx\vv\undefined
\newcommand{\vv}[1]{\marginpar{\textcolor{blue}{#1}}}
\else
\renewcommand{\vv}[1]{\marginpar{\textcolor{blue}{#1}}}
\fi

\newcommand{\ie}{i.e.,\xspace} 
\newcommand{\eg}{e.g.\xspace}  
\newcommand{\etal}{et al.\xspace} 
\newcommand{\cf}{cf.~} 

\newcommand{\Endrulename}{End} 

\newcommand{\msgrulename}{Msg}

\newcommand{\choicerulename}{Choice}

\newcommand{\skiprulename}{Skip}
\newcommand{\semirulename}{Semi}

\newcommand{\tidrulename}{Id}

\newcommand{\zrulename}{z}
\newcommand{\srulename}{s}

\newcommand{\rulename}[1]{\text{\small\sc #1}\xspace}

\newcommand{\dualrulename}[1]{\rulename{D-{#1}}}
\newcommand{\dend}{\dualrulename{\Endrulename}}
\newcommand{\dmsg}{\dualrulename{\msgrulename}}

\newcommand{\dchoice}{\dualrulename{\choicerulename}}

\newcommand{\typeequivrulename}[1]{\rulename{E-{#1}}}
\newcommand{\eqend}{\typeequivrulename{\Endrulename}}

\newcommand{\eqmsg}{\typeequivrulename{\msgrulename}}

\newcommand{\eqchoice}{\typeequivrulename{\choicerulename}}

\newcommand{\eqidl}{\typeequivrulename{ConsL}}
\newcommand{\eqidr}{\typeequivrulename{ConsR}}

\newcommand{\eqzl}{\typeequivrulename{\zrulename L}}
\newcommand{\eqsl}{\typeequivrulename{\srulename L}}

\newcommand{\eqskip}{\typeequivrulename{\skiprulename}}
\newcommand{\eqsemi}{\typeequivrulename{\semirulename}}
\newcommand{\eqmsgskipl}{\typeequivrulename{MsgSkipL}}
\newcommand{\eqmsgsemil}{\typeequivrulename{MsgSemiL}}
\newcommand{\eqneutl}{\typeequivrulename{NeutL}}
\newcommand{\eqneutronel}{\typeequivrulename{Neut1L}}
\newcommand{\eqneutrtwol}{\typeequivrulename{Neut2L}}
\newcommand{\eqneutrtwor}{\typeequivrulename{Neut2R}}
\newcommand{\eqidsemil}{\typeequivrulename{IdSemiL}}
\newcommand{\eqassocl}{\typeequivrulename{AssocL}}
\newcommand{\eqassocr}{\typeequivrulename{AssocR}}
\newcommand{\eqdistl}{\typeequivrulename{DistL}}

\newcommand{\typeformrulename}[1]{\rulename{T-{#1}}}
\newcommand{\wend}{\typeformrulename{\Endrulename}}

\newcommand{\wmsg}{\typeformrulename{\msgrulename}}

\newcommand{\wchoice}{\typeformrulename{\choicerulename}}
\newcommand{\wid}{\typeformrulename{\tidrulename}}

\newcommand{\wz}{\typeformrulename{\zrulename}}
\newcommand{\ws}{\typeformrulename{\srulename}}
\newcommand{\wskip}{\typeformrulename{Skip}}
\newcommand{\wsemi}{\typeformrulename{Semi}}

\newcommand{\termrulename}[1]{\rulename{$\checkmark$-{#1}}}
\newcommand{\teskip}{\termrulename{Skip}}
\newcommand{\tesemi}{\termrulename{Semi}}
\newcommand{\teid}{\termrulename{Id}}

\newcommand{\embedrulename}[1]{\rulename{Emb-{#1}}}
\newcommand{\embskip}{\embedrulename{\skiprulename}}
\newcommand{\embmsg}{\embedrulename{\msgrulename}}
\newcommand{\embchoice}{\embedrulename{\choicerulename}}
\newcommand{\embid}{\embedrulename{\tidrulename}}
\newcommand{\embsemiskip}{\embedrulename{\semirulename\skiprulename}}
\newcommand{\embsemimsg}{\embedrulename{\semirulename\msgrulename}}
\newcommand{\embsemichoice}{\embedrulename{\semirulename\choicerulename}}
\newcommand{\embsemiid}{\embedrulename{\semirulename\tidrulename}}
\newcommand{\embsemisemi}{\embedrulename{\semirulename\semirulename}}

\newcommand{\contrcrulename}[1]{\rulename{C-{#1}}}
\newcommand{\cend}{\contrcrulename{\Endrulename}}
\newcommand{\cskip}{\contrcrulename{\skiprulename}}

\newcommand{\cchoice}{\contrcrulename{\choicerulename}}

\newcommand{\ctid}{\contrcrulename{\tidrulename}}
\newcommand{\ctz}{\contrcrulename{\zrulename}}
\newcommand{\cts}{\contrcrulename{\srulename}}
\newcommand{\csemid}{\contrcrulename{Semi1}}
\newcommand{\csemind}{\contrcrulename{Semi2}}
\newcommand{\cmsg}{\contrcrulename{\msgrulename}}

\newcommand{\premspace}{\quad\;}

\newcommand{\rulewend}{
  \axiom{\istype{\End}}
}

\newcommand{\rulewmsg}{
  \frac{\istype T\premspace\istype U}{\istype{\MSG TU}}
}

\newcommand{\rulewchoice}{
  \frac{\istype{T_\ell}\premspace\foralllinL}{\istype{\choicet}}
}

\newcommand{\ruleeqend}{\axiom{\issequiv \End \End}}

\newcommand{\ruleeqmsg}{
  \frac{
    \isequiv T U
    \premspace
    \issequiv V W
  }{
    \issequiv{\MSG TV}{\MSG UW}
  }
}

\newcommand{\ruleeqchoice}{
  \frac{
    \issequiv{T_l}{U_l}
    \premspace
    \foralllinL
  }{
    \issequiv{\choicet}{\choice{\record \ell UL}}
  }
}

\newcommand{\ruledualend}{\axiom{\isdual\End\End}}
\newcommand{\ruledualmsg}{
  \frac{
    \isequiv T U 
    \premspace
    \isdual V W
  }{
    \isdual{\MSG TV}{\blk{\dual{\typec\sharp}}\,U.W}
  }
}

\newcommand{\ruledualchoice}{
  \frac{
    \isdual{T_\ell}{U_\ell}
    \premspace
    \foralllinL
  }{
    \isdual{\choice{\record \ell TL}}{\blk{\dual{\typec\star}}{\record \ell UL}}
  }
}

\begin{document}

\title{The Different Shades of Infinite Session Types\thanks{Supported by EPSRC
    EP/T014628/1 ``Session Types for Reliable Distributed Systems'', by FCT
    PTDC/CCI-CIF/6453/2020 ``Safe Concurrent Programming with Session Types''
    and by the LASIGE Research Unit UIDB/00408/2020 and UIDP/00408/2020.}}


\author{Simon J. Gay
\thanks{School of Computing Science, University of Glasgow, UK. 
{\tt\href{mailto:simon.gay@glasgow.ac.uk}{\nolinkurl{simon.gay@glasgow.ac.uk}}}}
\orcidID{0000-0003-3033-9091}
\and
Diogo Poças
\thanks{LASIGE, Faculdade de Ciências, Universidade de Lisboa, Portugal.
{\tt\{\href{mailto:dmpocas@ciencias.ulisboa.pt}{\nolinkurl{dmpocas}},\href{mailto:vmvasconcelos@ciencias.ulisboa.pt}{\nolinkurl{vmvasconcelos}}\}@ciencias.ulisboa.pt}}
\orcidID{0000-0002-5474-3614}
\and
Vasco T. Vasconcelos
\footnotemark[3]
\orcidID{0000-0002-9539-8861}
}
\maketitle



\begin{abstract}
  Many type systems include infinite types. In session type systems, which are
  the focus of this paper, infinite types are important because they allow the
  specification of communication protocols that are unbounded in time. Usually
  infinite session types are introduced as simple finite-state expressions
  $\REC X T$ or by non-parametric equational definitions $\iseqt{X}{T}$.
  Alternatively, some systems of label- or value-dependent session types go
  beyond simple recursive types. However, leaving dependent types aside, there
  is a much richer world of infinite session types, ranging through various
  forms of parametric equational definitions, all the way to arbitrary infinite
  types in a coinductively defined space. We study infinite session types across
  a spectrum of shades of grey on the way to the bright light of general
  infinite types. We identify four points on the spectrum, characterised by
  different styles of equational definitions, and show that they form a strict
  hierarchy by establishing bidirectional correspondences with classes of
  automata: finite-state, 1-counter, pushdown and 2-counter. This allows us to
  establish decidability and undecidability results for the problems of type
  formation, type equivalence and duality in each class of types. We also
  consider previous work on context-free session types (and extend it to
  higher-order) and nested session types, and locate them on our spectrum of
  infinite types.
\end{abstract}



\section{Introduction}
\label{sec:introduction}

Session types
\cite{DBLP:conf/concur/Honda93,DBLP:conf/esop/HondaVK98,DBLP:journals/csur/HuttelLVCCDMPRT16,DBLP:conf/parle/TakeuchiHK94}
are an established approach to specifying communication protocols, so that
protocol implementations can be verified by static typechecking or dynamic
monitoring. The simplest protocols are finite: for example,
$\IN{\Int}{\OUT{\bool}\End}$ describes a protocol in which an integer is
received, then a boolean is sent, and that's all. Most systems of session types,
however, include equi-recursive types for greater expressivity. A type that endlessly
repeats the simple send-receive protocol is $\XT$ such that
$\iseqt{X}{\IN{\Int}{\OUT{\bool}X}}$, which can also be specified by
$\REC{X}{\IN{\Int}{\OUT{\bool}X}}$. More realistic examples usually combine
recursion and choice, as in $\YT$ such that
$\iseqt{Y}{\typec{ \extchoice\{\gol\colon\IN{\Int}{\OUT{\bool}\YT},
  \quitl\colon\End\}}}$ which offers a choice between $\typec\gol$ and
$\typec\quitl$ operations, each with its own protocol. A natural observation is
that session types look like finite-state automata, but some systems from the
literature go beyond the finite-state format: for example, context-free
session types \cite{DBLP:conf/icfp/ThiemannV16} and nested session types
\cite{DBLP:conf/esop/DasDMP21,DBLP:journals/corr/abs-2103-15193}, as well as label-dependent session types
\cite{DBLP:journals/pacmpl/ThiemannV20} and value-dependent session types
\cite{DBLP:conf/ppdp/ToninhoCP11}.

Even without introducing dependent types, a range of definitional formats can be
considered for session types, presumably with varying degrees of expressivity,
but they have never been systematically studied. That is the aim of the present
paper. We consider various forms of parameterised equational definitions,
illustrated by six running examples. Because our formal system only has one base
type, the terminated channel type $\End$, the running examples simply use $\End$
(or $\Skip$ for context-free session types) as a representative basic message
type that could otherwise be $\bool$ or $\Int$.

Our study of classes of infinite types should be generally applicable; we make
it concrete by concentrating on session types where (potential) infinite types
occur naturally.
For the sake of uniformity, all our non-finite session types are introduced by
equations, rather than, say, $\reck$-types. Equations may be further
parameterized, thus accounting for types that go beyond recursive types. The
examples below illustrate the different kinds of parameterized equations we use.

\begin{example}[No parameters]
  \label{exa:rec}
  Type $\tloop$ is $\typec X$ with equation $\iseqt X {\OUT \End X}$.
  Intuitively $\tloop=\OUT\End{\OUTn\End\dots}$ continuously outputs values of type $\End$.
\end{example}

\begin{example}[One natural number parameter]
  \label{exa:onecounter}
  Assuming $\zero$ and $\succc$ as the natural number constructors and $\typec N$
  as a variable over natural numbers, type $\tcounter$ is $\CALLT X \zero$ with
  equations
  \begin{align*}
    \lhst X \zero \Eq&\; \typec{\& \{\incl\colon \CALLT X {\succ\zero}, \dumpl\colon \CALLT Y\zero\}}
    &&&
        \lhst Y \zero \Eq&\; \End
    \\
    \lhst X {\succ N} \Eq&\; \typec{\& \{\incl\colon \CALLT
                           X{\succ{\succ N}}, \dumpl\colon \CALLT Y{\succ N}\}}
    &&&
        \lhst Y {\succ N} \Eq&\; \OUT \End {\CALLT Y N}
  \end{align*}
A sequence of $n$ $\typec\incl$ operations followed by a $\typec\dumpl$
triggers a reply of $n$ $\End$ output messages.\footnote{The final $\End$ at
$\iseqt {\lhst Y \zero}\End$ closes the channel and does not count as a message.}
\end{example}

\begin{example}[Context-free types]
  \label{exa:context-free}
  With type $\Skip$ used either to finish a session or to move to the next
  operation, type $\ttree$ is $\XT$ with equation
  \begin{equation*}
    \iseqt{X}{\&\{\leafl\colon\Skip, \nodel\colon \semit X {\semit {\INn\Skip} X}\}}
  \end{equation*}
  The $\typec\leafl$ choice terminates the reception of a binary tree of $\Skip$
  values and the $\typec\nodel$ choice triggers the reception of a (left) tree,
  followed by $\INn\Skip$ (root), followed by a (right) tree. Even though
  the development in the rest of the paper considers higher-order types (where
  messages may convey arbitrary types rather than $\Skip$ alone), for simplicity
  our example is first-order.
\end{example}

\begin{example}[One list parameter]
  \label{exa:pushdown}
  Assuming $\typec\sigma$ and $\typec\tau$ as symbols and $\typec S$ as a
  variable over sequences of symbols (with $\typec\Empty$ the empty sequence), type
  $\tstack$ is $\CALLT X\Empty$ with equations
  \begin{align*}
    \lhst X \Empty \Eq&\; \typec{\&\{
                        \pushal\colon \CALLT X \sigma,
                        \pushbl\colon \CALLT X \tau\}}
    \\
    \lhst X {\sigma S} \Eq&\; \typec{\&\{
                            \pushal\colon \CALLT X {\sigma\sigma S},
                            \pushbl\colon \CALLT X {\tau\sigma S},
                            \popl\colon \OUT \End {\CALLT X {S}}\}}
    \\
    \lhst X {\tau S} \Eq&\; \typec{\&\{
                           \pushal\colon \CALLT X {\sigma\tau S},
                           \pushbl\colon \CALLT X {\tau\tau S},
                           \popl\colon \IN \End {\CALLT X {S}}\}}
  \end{align*}
  Type $\tstack$ records simple protocols composed of $\OUTn\End$ and $\INn\End$
  messages. Symbol $\typec\sigma$ in a parameter to a type constructor $\XT$
  denotes an output message and symbol $\typec\tau$ an input message. The
  protocol behaves as a stack with two distinct push operations ($\typec\pushal$
  and $\typec\pushbl$). The symbol ($\typec\sigma$ or $\typec\tau$) at top of
  the stack determines whether a $\typec\popl$ operation triggers 
$\OUTn\End$  or $\INn\End$,
respectively.
\end{example}

\begin{example}[Nested types]
  \label{exa:nested}
  Taking $\typec\alpha$ as a variable over types, type $\tnested$ is $\typec{X_\Empty}$ with
  equations
  \begin{align*}
    \typec{X_\Empty} \Eq&\; \typec{\&\{
      \pushal\colon \CALLT \Xout {X_\Empty},
      \pushbl\colon \CALLT \Xin {X_\Empty}
    \}}
    \\                                 
    \lhst \Xout \alpha \Eq&\; \typec{\&\{
      \pushal\colon \CALLT \Xout {\CALLT \Xout \alpha},
      \pushbl\colon \CALLT \Xin {\CALLT \Xout \alpha},
      \popl\colon \OUT \End \alpha
    \}}
    \\                                 
    \lhst \Xin \alpha \Eq&\; \typec{\&\{
      \pushal\colon \CALLT \Xout {\CALLT \Xin \alpha},
      \pushbl\colon \CALLT \Xin {\CALLT \Xin \alpha},
      \popl\colon \IN \End \alpha
    \}}
  \end{align*}
  Type constructors such as $\typec{X_\Empty}, \Xout, \Xin$ take an arbitrary but fixed
  number of arguments. Type $\tnested$ behaves as $\tstack$ in
  \cref{exa:pushdown}. The alignment should be clear if we take, \eg
  $\CALLT \Xout {\CALLT \Xin \alpha}$ for $\CALLT X {\sigma\tau S}$, with
  $\typec\sigma$ denoting output and $\typec\tau$ denoting input. Type
  constructors $\Xout$ and $\Xin$ play the roles of stack symbols (symbols at
  the top of the stack, $\typec\sigma$ or $\typec\tau$); type variable
  $\typec\alpha$ denotes the lower part of the stack ($\typec S$ in
  \cref{exa:pushdown}).
\end{example}

\begin{example}[Two natural number parameters]
  \label{exa:turing}
  Type $\titer$ is $\CALLTT X\zero\zero$ with 
  \begin{align*}
    \lhstt{X}{\zero}{N'} \Eq&\; {\IN\End {\CALLTT Y \zero{\succ{N'}}}}
    &
        \lhstt{Y}{N}{\zero} \Eq&\; {\CALLTT X N \zero}
    \\
    \lhstt{X}{\succ N}{N'} \Eq&\; {\OUT\End {\CALLTT X N {\succ{N'}}}}
    &
        \lhstt{Y}{N}{\succ{N'}} \Eq&\; {\CALLTT Y {\succ N}{N'}}
  \end{align*}
  Informally, writing $\OUTn{\End^n}$ for a sequence of $n$ output $\End$
  messages, these definitions give
  $\titer={}
  \IN\End{\OUT{\End^1}{\IN\End{\OUT{\End^2}{\IN\End{\OUTn{\End^3}}}}}} \typec\dots$
\end{example}

It is intuitively clear that \cref{exa:onecounter,exa:nested,exa:pushdown,exa:turing}
cannot be expressed without parameters. It is perhaps less clear that each
definitional style in \cref{exa:rec,exa:onecounter,exa:pushdown,exa:turing} is strictly more expressive than the
previous one. This is the main result of the paper. We establish a hierarchy
from
finite session types all the way up to non-computable types that have no representation at all. The latter certainly exist, because for every infinite binary expansion of a real number between zero and one there is a session type derived by mapping 0 to send and 1 to receive --- and we know for cardinality reasons that almost all of these types are non-computable.

Our methodology is to develop the connection between session types and automata,
in particular between progressively more expressive definitional styles and
progressively more powerful classes of automata. We also consider the formal
language class corresponding to each class of automata, and the decidability of
important properties such as contractiveness, type formation, type equivalence and type duality. Our
results are summarised in the table below, establishing a hierarchy of session
types in parallel to the Chomsky hierarchy of languages, where by a 1-counter
language, we mean a language accepted by a (deterministic) 1-counter automaton
and where DCFL abbreviates deterministic context-free languages.
In the final row of the table we make it clear that it is impossible to give an
explicit example of a non-computable type or to even state the decision
problems.
%

Context-free and 1-counter types are incomparable. Essentially, both models lie between levels 2 and 3 of the Chomsky hierarchy and correspond to different restrictions of deterministic pushdown automata. Context-free types correspond to constraining automata with a single state, whereas 1-counter types correspond to constraining the stack to have a single symbol.
\begin{center}
  \begin{tabular}{c|cccc}
Type class & Example & Contractiveness & Type duality / & Language model\\
           &         &                 & equivalence    &               \\
    \hline
    Finite & {\small $\OUT\End\End$} & Polytime & Polytime & Finite languages\\
    Recursive & $\tloop$ & Polytime & Polytime & Regular languages\\
    1-counter & $\tcounter$ & Polytime & Polytime & 1-counter languages\\
    HO context-free & $\ttree$ & Polytime & Decidable & Open\footnotemark\\
    Pushdown & $\tstack$ & Polytime & Decidable & DCFL\\
    Nested & $\tnested$ & Polytime & Decidable & DCFL\\
    2-counter & $\titer$ & Undecidable & Undecidable & Decidable languages\\
    Non-computable & --- & --- & --- & General languages
  \end{tabular}
\end{center}
\footnotetext{Possibly languages accepted by a single-state pushdown automata with empty stack acceptance.}

Our main contributions can be summarized as follows.

\begin{itemize}
  \item We propose three novel formal systems for representing session types (1-counter, pushdown, 2-counter), show that they are strictly more expressive than recursive session types, and that each system is strictly more expressive than the previous one (\cref{thm:inclusions}).
  \item We show that nested session types \cite{DBLP:conf/esop/DasDMP21} are
    equivalent to pushdown session types (\cref{thm:inclusions}).
  \item We introduce higher-order context-free session types and show that they
    stand between recursive and pushdown types, strictly (\cref{thm:inclusions}).
  \item We characterize each of the novel session types in our paper by a
    corresponding class in the Chomsky hierarchy of languages. Notably, we show
    that each model captures precisely the power of the corresponding class of
    automata (\cref{thm:equivalencetypesautomata}).
    This is in contrast with the results of Das \etal \cite{DBLP:conf/esop/DasDMP21}, who only show (in one direction) that nested session types can be simulated by deterministic pushdown automata.
  \item We prove that the problems of type formation, type equivalence and type
    duality are decidable up to pushdown session types (\cref{thm:typeformation,thm:typeequivalence,thm:typeduality}), but undecidable
    for 2-counter session types (\cref{thm:undecidability}). This implies, in
    particular, that equivalence for higher-order context-free session types is
    decidable. The decidability results are not entirely unexpected, given that
    type equivalence for nested session types was recently shown to be decidable
    \cite{DBLP:conf/esop/DasDMP21}, and that these are equivalent to pushdown
    types. However, 
our proofs 
are independent of Das \etal 
    \cite{DBLP:conf/esop/DasDMP21}.
  \item Finally, we show a technical result in formal language theory of independent interest: every (deterministic 1-counter, deterministic pushdown, deterministic 2-counter) automaton that accepts a prefix-closed language can be converted into an automaton with a single non-accepting state, which acts as a sink (\cref{thm:obviouslyprefixclosed}). 
\end{itemize}


\paragraph{Organization of the paper}
In \cref{sec:types-procs} we introduce the various classes of session types. In
\cref{sec:treeslanguages} we explain how to associate to each given type a
labelled infinite tree, as well as a set which we call the language of traces of
that type. We also state our main results on the strict hierarchy of types and
on how previously studied classes of types fit into this hierarchy
(\cref{thm:inclusions}). In \cref{sec:systemstoautomata} we describe
how to convert a type into an automaton
accepting its traces. In \cref{sub:automatatosystems} we travel in the converse
direction, \ie from an automata into the corresponding type, and present a
characterisation theorem of the different types in our hierarchy
(\cref{thm:equivalencetypesautomata}).
In \cref{sec:hierarchy,sec:hierarchy2} we provide the details in the proof of \cref{thm:inclusions}; \cref{sec:hierarchy} proves the main hierarchy and \cref{sec:hierarchy2} proves the results for context-free and nested session types.
In \cref{sec:decidability} we present
our main algorithmic results: type formation, type equivalence
and type duality are all decidable up to pushdown types (\cref{thm:typeformation,thm:typeequivalence,thm:typeduality}),
and undecidable for 2-counter types (\cref{thm:undecidability}).
In \cref{sec:related} we give an overview of related work and \cref{sec:conclusion} concludes the paper.



\section{Shades of types}
\label{sec:types-procs}

This section introduces the various 
session types in a uniform framework.


\begin{figure}[t!]
  \begin{multicols}{2}
    Polarity and view\hfill{}
    \begin{align*}
      \typec\sharp \grmeq{} \typec? \grmor{} \typec!
      &&
      \choice \grmeq \extchoice \grmor \intchoice
    \end{align*}
    \declrel{Type formation}{$\istype{T}$}
    \begin{gather*}
      \tag\wend\rulewend
      \\
      \tag\wmsg\rulewmsg
      \\
      \tag\wchoice\rulewchoice
    \end{gather*}  
    \declrel{Type equivalence}{$\isequiv TT$}
    \begin{gather*}
      \tag\eqend\ruleeqend 
    \end{gather*}
    \begin{gather*}
      \\
      \tag\eqmsg\ruleeqmsg
      \\
      \tag\eqchoice\ruleeqchoice
    \end{gather*}
    \declrel{Duality}{$\isdual SS$}
    \begin{align*}
      \dual{\typec?} ={} \typec!
      &&
      \dual{\typec!} ={} \typec?
      &&
      \dual{\typec\&} ={} \typec\oplus
      &&
      \dual{\typec\oplus} ={} \typec\&
    \end{align*}
    \begin{gather*}
      \tag\dend\ruledualend
      \\
      \tag\dmsg\ruledualmsg
      \\
      \tag\dchoice\ruledualchoice
    \end{gather*}
    %
    
  \end{multicols}
  \caption{Finite and infinite types.}
  \label{fig:finite-abbr}
\end{figure}


\paragraph{The finite world}

Finite types 
are in \cref{fig:finite-abbr}. The syntax of types is
introduced via formation rules, paving the way for infinite types. Session types
comprise the terminated type $\End$, the input type $\inty$ (input a value of
type $\TT$ and continue as $\UT$), the output type $\outt$ (output a value of
type $\TT$ and continue as $\UT$), external choice $\extct$ (receive a label
$k\in L$ and continue as $\typec{T_k}$) and internal choice $\intct$ (select a
label $k\in L$ and continue as $\typec{T_k}$). To avoid repeating similar rules,
we use the symbol $\typec\sharp$ to denote either $\typec?$ or $\typec!$, and
the symbol $\choice$ to denote either $\extchoice$ or $\intchoice$.
At this point type equivalence is essentially syntactic equality, but the rule
format allows for seamless extensions to infinite settings.
Types, type equivalence and duality are all
standard~\cite{DBLP:journals/acta/GayH05,DBLP:conf/esop/HondaVK98,DBLP:journals/iandc/Vasconcelos12}.
Note that rule \dmsg defines duality with respect to type equivalence: $\OUT TV$ and $\IN UW$ are dual types iff the type being exchanged is the same ($\isequiv TU$) and the continuations are dual ($\isdual VW$).

For finite types 
all judgements in \cref{fig:finite-abbr} are
interpreted \emph{inductively}.
For example, we can show that $\OUT{(\IN\End\End)}{\OUT\End\End}$ is a type by
exhibiting a finite derivation ending with this judgement.


\begin{figure}[t!]
  \begin{multicols}{2}
    \declrel{Type contractivity \inductive}{$\iscontrt T$}
    \begin{gather*}
      \tag\cend
      \axiom{\iscontrt \End}
      \\
      \tag\cmsg
      \axiom{\iscontrt {\MSG TU}}
      \\
      \tag\cchoice
      \axiom{\iscontrt \choicet}
      \\
      \tag\ctid
      \frac{
        \iseqt XT
        \premspace
        \iscontrt T
      }{
        \iscontrt X
      }
    \end{gather*}
    \declrel{New type formation rules \coinductive}{$\istype{T}$}
    \begin{gather*}
      \tag\wid
      \frac{
        \iseqt XT
        \premspace
        \iscontrt T
        \premspace
        \istype T
      }{
        \istype X
      }
    \end{gather*}
    %
    %
    \declrel{New type equivalence rules \coinductive}{$\isequiv TT$}
    \begin{gather*}
      \tag\eqidl
      \frac{
        \iseqt XU
        \premspace 
        \iscontrt U
        \premspace
        \isequiv UT
      }{
        \isequiv XT
      }
      \\
      \tag\eqidr
      \frac{
        \iseqt XU
        \premspace 
        \iscontrt U
        \premspace
        \isequiv TU
      }{
        \isequiv TX
      }
    \end{gather*}
  %
    %
    %
  \end{multicols}
\vspace{-3ex}
  \caption{Recursive types. Extends \cref{fig:finite-abbr}.
  }
  \label{fig:recursive-abbr}
\end{figure}


\paragraph{The recursive world}

Recursive types suggest the first glimpse of infinity. The details are in
\cref{fig:recursive-abbr}. Recursion is given via equations, rather than
$\mu$-types for example, for easier extension. Towards this end, we introduce
type constructors $\XT$ and 
equations of the form $\iseqt XT$. 
The
set of type 
constructors is finite. We further assume at most one
equation for each type, 
so that there are finitely many
type 
equations.
Every valid type $\TT$ 
is required to be contractive, that is
$\iscontrt T$. 
Contractiveness ensures that types reveal a type constructor after finitely many unfolds, and excludes undesirable cycles that don't describe any behaviour, \eg cycles of the form
$\{\iseqt XY, \iseqt YZ, \iseqt ZX\}$. Contractiveness is inductive: we look for
finite derivations for $\iscontrt T$ 
judgements.
A coinductive interpretation of the rules would allow to conclude $\iscontrt X$
given an equation $\iseqt XX$.
In contrast, type formation, type equivalence and duality are now
interpreted \emph{coinductively}. 

For example, no finite derivation would allow showing that $\istype\tloop$.
Instead we proceed by showing that set $\{\End, \OUT \End X, \typec X\}$ is
\emph{backward closed}~\cite{sangiorgi:bisimulation-coinduction} for the rules for $\istype T$ in
\cref{fig:recursive-abbr}, given that
$\OUT\End X$, the right-hand side of the equation for $\XT$, is contractive.



\begin{figure}[t!]
  \begin{multicols}{2}
    %
    %
    Natural numbers\hfill{}
    \begin{equation*}
      \parc n \grmeq \zero \grmor \succ{\typec n}
    \end{equation*}
    %
    %
    %
    \declrel{New type contractivity rules \inductive}{$\iscontrt T$}
    \begin{gather*}
      \tag\ctz
      \frac{
        \ispeqt X{\zero}T
        \premspace
        \iscontrt T
      }{
        \iscontrt {\CALLT X{\zero}}
      }
      \\
      \tag\cts
      \frac{
        \ispeqt X{\succ N}T
        \premspace
        \iscontrt {T\subs nN}
      }{
        \iscontrt {\CALLT X{\succ n}}
      }
    \end{gather*}
    \declrel{New type formation rule \coinductive}{$\istype{T}$}
    \begin{gather*}
      \tag\wz
      \frac{
        \ispeqt X \zero T
        \premspace
        \iscontrt T
        \premspace
        \istype T
      }{
        \istype{\CALLT X\zero}
      }
    \end{gather*}
    \begin{gather*}
      \tag\ws
      \frac{
        \ispeqt X {\succ N} T
        \;\;\;
        \iscontrt {T\subs nN}
        \;\;\;
        \istype {T\subs nN}
      }{
        \istype{\CALLT X{\succ n}}
      }
    \end{gather*}
    %
    %
    \declrel{New type equivalence rules \coinductive}{$\isequiv TT$}
    \begin{gather*}
      \tag\eqzl
      \frac{
        \ispeqt X{\zero}U
        \premspace
        \iscontrt U
        \premspace 
        \isequiv UT
      }{
        \isequiv {\CALLT X{\zero}} T
      }
      \\
      \tag\eqsl
      \frac{
        \ispeqt X{\succ N}U
        \premspace
        \iscontrt {U\subs nN}
        \premspace
        \isequiv {U\subs nN} T
      }{
        \isequiv {\CALLT X{\succ n}} T
      }
    \end{gather*}
    %
    %
  \end{multicols}
  \caption{1-counter types. Extends \cref{fig:recursive-abbr};
    removes $\typec X$; adds $\CALLT Xn$. Right versions of rules \eqzl and \eqsl
    omitted.}
  \label{fig:onecounter-abbr}
\end{figure}


\paragraph{The 1-counter world}

The next step takes us to equations parameterised on natural numbers. The
details are in \cref{fig:onecounter-abbr}. Natural numbers are built from the
nullary constructor $\zero$ and the unary constructor $\succc$. We discuss the
changes from the recursive world in \cref{fig:recursive-abbr}.
Given a variable
$\parc N$ on natural numbers, to each type constructor $\typec X$ we associate
at most two equations, $\ispeqt X{\zero}T$ and $\ispeqt X{\succ N}U$. 
The rules for recursive types 
are naturally adapted to 1-counter types.
Here again,
type 
formation requires a suitable notion of contractiveness to
exclude cycles of equations that never reach a type constructor, \eg cycles of
the form
$\{\iseqt{\CALLT X {\succ N}}{\CALLT Y{\succ{\succ N}}}, \iseqt{\CALLT Y{\succ
    N}}{\CALLT X N}\}$. Notice that the right-hand-side of an equation
$\ispeqt X{\succ N}T$ is not necessarily a type for it may contain natural
number variables ($\parc N$ in particular).
However, if $\parc n$ is a natural number,
then $\typec T\subs nN$ (that is, $\typec T$ with occurrences of $\parc N$
replaced by $\parc n$) should be a type (\cf rule
\ws). 
Again, to prove that $\istype\tcounter$, we show that the set $
\{
\CALLT X n,
\CALLT Y n,
\End,
\OUT\End {\CALLT Y n},
\typec{\&\{\incl\colon \CALLT X{\succ n}, \dumpl\colon \CALLT Y n}
\}
\mid \isnat n
\}$ is
backward closed.
%


\paragraph{Higher-order context-free session types}
\begin{figure}[t!]
  \begin{multicols}{2}
  \declrel{Is terminated predicate \inductive}{$\isdone T$}
  \begin{gather*}
    \tag\teskip\axiom{\isdone\Skip}
    \\
    \tag\tesemi
    \frac{
      \isdone{T}
      \premspace
      \isdone{U}
    }{
      \isdone{\semit TU}
    }
    \\
    \tag\teid
    \frac{
      \iseqt XT
      \premspace
      \isdone T
    }{
      \isdone{X}
    }
  \end{gather*}
  \declrel{New type contractive rules \inductive}{$\iscontrt T$}
  \begin{gather*}
    \tag\cskip
    \axiom{\iscontrt \Skip}
    \\
    \tag\cmsg
    \axiom{\iscontrt {\MSGn T}}
    \\
    \tag\csemid
    \frac{
      \isdone T
      \premspace
      \iscontrt U
    }{
      \iscontrt{T;U}
    }
    \\
    \tag\csemind
    \frac{
      \isnotdone T
      \premspace
      \iscontrt T
    }{
      \iscontrt{T;U}
    }
  \end{gather*}
  \declrel{New type formation rules \coinductive}{$\istype{T}$}
  \begin{gather*}
    \tag\wskip\axiom{\istype\Skip}
    \\
    \tag\wmsg\frac{\istype T}{\istype{\MSGn T}}
    \\
    \tag\wsemi\frac{\istype T \premspace \istype U}{\istype{\semit TU}}
  \end{gather*}
  \declrel{Type equivalence \coinductive}{$\isequiv TT$}
  \begin{gather*}
    \tag\eqskip
    \axiom{\isequiv \Skip \Skip}
    \\
    \tag\eqmsg
    \frac{\isequiv TU}{\isequiv{\MSGn T}{\MSGn U}}
    \\
    \tag\eqneutl
    \frac{\isequiv TU}{\isequiv {\semit \Skip T} U}
  \end{gather*}
  \begin{gather*}
    \tag\eqmsgskipl
    \frac{
      \isequiv TV
      \premspace
      \isequiv U\Skip
    }{
      \isequiv{\semit{\MSGn T}U}{\MSGn V}
    }
    \\
    \tag\eqmsgsemil
    \frac{
      \isequiv TV
      \premspace
      \isequiv UW
    }{
      \isequiv{\semit{\MSGn T}U}{\semit{\MSGn V}W}
    }
    \\
    \tag\eqdistl
    \frac{
      \isequiv{\choice\recordf \ell {\semit{T_\ell}{U}} L}{V}
    }{
      \isequiv {\semit \choicet U} V
    }
    \\
    \tag\eqassocl
    \frac{
      \isequiv {\semit T {(\semit UV)}}{W}
    }{
      \isequiv {\semit {(\semit TU)} V}{W}
    }
    \\
    \tag\eqidsemil
    \frac{
      \iseqt XT
      \premspace 
      \iscontrt T
      \premspace
      \isequiv {\semit TU}{V}
    }{
      \isequiv {\semit XU}{V}
    }
  \end{gather*}
  \declrel{Embedding \coinductive}{$\embeds TU$}
  \begin{gather*}
    \tag\embskip
    \axiom{\embeds{\Skip}{\End}}
    \\
    \tag\embmsg
    \frac{\embeds TU}{\embeds{\MSGn T}{\MSG U \End}}
    \\
    \tag\embchoice
    \frac{\embeds{T_\ell}{U_\ell}
      \premspace
      \foralllinL
    }{
      \embeds{\choicet}{\choice\recordf \ell {U_\ell} L}
    }
    \\
    \tag\embid
    \frac{
      \iseqt XT
      \premspace
      \iscontrt T
      \premspace
      \embeds TU
    }{
      \embeds XU
    }
    \\
    \tag\embsemiskip
    \frac{
      \embeds TU
    }{
      \embeds {\semit \Skip T}{U}
    }
    \\
    \tag\embsemimsg
    \frac{
      \embeds TV
      \premspace
      \embeds UW
    }{
      \embeds {\semit {\MSGn T} U}{\MSG VW}
    }
    \\
    \tag\embsemichoice
    \frac{
      \embeds {\semit {T_{\ell}} T}{U_\ell}
      \premspace
      \foralllinL
    }{
      \embeds {\semit \choicet T}{\choice\recordf \ell {U_\ell} L}
    }
    \\
    \tag\embsemiid
    \frac{
      \iseqt XT
      \premspace
      \iscontrt T
      \premspace
      \embeds {\semit TU}{V}
    }{
      \embeds {\semit XU}{V}
    }
    \\
    \tag\embsemisemi
    \frac{
      \embeds {\semit T {(\semit{U}{V})}}{W}
    }{
      \embeds {\semit {(\semit{T}{U})}{V}}{W}
    }
  \end{gather*}
  \end{multicols}
  \caption{Higher-order context-free types. Extends \cref{fig:recursive-abbr}; removes
    $\End$, $\MSG TU$; adds $\Skip$, $\MSGn T$, $\semit TT$. Right versions
    of rules \eqneutl, \eqmsgskipl, \eqmsgsemil, \eqdistl, \eqassocl, \eqidsemil
    omitted.}
  \label{fig:cfst}
\end{figure}


  %

A little detour takes us to context-free session types, proposed by Thiemann and
Vasconcelos~\cite{DBLP:conf/icfp/ThiemannV16} (see also Almeida et
al.~\cite{DBLP:journals/corr/abs-2106-06658}). Here we follow the distilled
presentation of Almeida et al.~\cite{DBLP:conf/tacas/AlmeidaMV20}, extending to
the higher-order setting (that is, allowing $\INn T$ and $\OUTn T$ for an
arbitrary type $\TT$ instead of just basic type $\Skip$). 
The syntax for context-free session types (presented in \cref{fig:cfst}) slightly departs from the main classes analysed in this paper; the
distinguishing aspects are as follows.
\begin{itemize}
\item There is a new type constructor for sequential composition of session
  types: the sequential composition of $\TT$ and $\UT$ is denoted by $\semit TU$.
\item Type $\End$ is replaced by a new type $\Skip$ with a distinct behaviour.
  Intuitively, $\End$ is used to finish a session type, while $\Skip$ merely
  moves to the next operation.
\item 
  The constructors for sending and receiving are now simply $\INn T$ and
  $\OUTn T$, rather than $\IN TU$ and $\OUT TU$.\footnote{Traditional
    (first-order) context-free session types restrict messages to $\INn \Skip$ and
  $\OUTn \Skip$, with $\Skip$ representing a basic type.}
\end{itemize}

In order to align the presentation with the other classes of types, we use
equations rather than $\mu$-types as in the original work.

We discuss the main differences with respect to recursive types (\cref{fig:recursive-abbr}).
For $\XT$ to be a type under equation $\iseqt X T$, the right-hand side $\TT$ must
be contractive, meaning that successive unfoldings either reach $\Skip$ or one
of the type constructs $\typec?, \typec!, \typec\&, \typec\Oplus$ after finitely
many steps. This excludes non-types $\XT$ such as those defined under equations
$\iseqt X {\semit \Skip X}$, or $\iseqt X {\semit X\Skip}$.
Contractiveness for sequential composition makes use of a new `is terminated'
predicate. Judgement $\isdone T$ denotes a type that exibits no behaviour.
Terminated types are composed solely of constructors $\Skip$, $\XT$ and
sequential composition.

On what concerns type equivalence, the first three rules (\wskip to \wsemi) are
the congruence rules for the new type constructors. The last six rules
constitute the novelty of context-free types.
Sequential composition provides an associative monoidal structure, with $\Skip$
acting as the identity (rules \eqneutl and \eqassocl). Rule
\eqdistl introduces distributivity of choice over sequencing.
The definition is again coinductive: $\isequiv TU$ looks at the top constructors
of $\TT$ and $\UT$. If either $\TT$ or $\UT$ are sequential compositions, then
one of the six left or six right (not shown) rules apply.
%

A formulation of type equivalence that explicitly incorporates the rules of an
equivalence relation would allow reducing the number of rules while simplifying
the remaining ones. For example the four $\Skip$-is-neutral rules would be reduced to
two axioms: $\isequiv{\semit \Skip T}{T}$ and $\isequiv{\semit T \Skip}{T}$.
Unfortunately, scaling this approach to the coinductive setting would make every
element related to every
other~\cite{DBLP:conf/mpc/DanielssonA10,DBLP:journals/jfp/GapeyevLP02}~\cite[Section 21.4]{DBLP:books/daglib/0005958}.


\begin{figure}[t!]
  \begin{multicols}{2}
    %
    %
    Strings\hfill{}
    \begin{equation*}
      \parc s \grmeq \parc\Empty \grmor \parc{\sigma s}
    \end{equation*}
    %
    %
    \declrel{New type contractive rules \inductive}{$\iscontrt T$}
    \begin{gather*}
      \tag\ctz
      \frac{
        \ispeqt X{\Empty}T
        \premspace
        \iscontrt T
      }{
        \iscontrt {\CALLT X{\Empty}}
      }
      \\
      \tag\cts
      \frac{
        \ispeqt X{\sigma S}T
        \premspace
        \iscontrt {T\subs s S}
      }{
        \iscontrt {\CALLT X{\sigma s}}
      }
    \end{gather*}
    \declrel{New type formation rules \coinductive}{$\istype{T}$}
    \begin{gather*}
      \tag\wz
      \frac{
        \ispeqt X \Empty T
        \premspace
        \iscontrt T
        \premspace
        \istype T
      }{
        \istype{\CALLT X\Empty}
      }
    \end{gather*}
    \begin{gather*}
      \tag\ws
      \frac{
          \ispeqt X {\sigma S} T
          \premspace
          \iscontrt {T\subs s S}
          \premspace
          \istype {T\subs s S}
      }{
        \istype{\CALLT X{\sigma s}}
      }
    \end{gather*}
    \declrel{New type equivalence rules \coinductive}{$\isequiv TT$}
    \begin{gather*}
      \tag\eqzl
      \frac{
        \ispeqt X{\Empty}U
        \premspace
        \iscontrt U
        \premspace 
        \isequiv UT
      }{
        \isequiv {\CALLT X{\Empty}} T
      }
      \\
      \tag\eqsl
      \frac{
          \ispeqt X{\sigma S}U
          \premspace
          \iscontrt {U\subs s S}
          \premspace
          \isequiv {U\subs s S}T
      }{
        \isequiv {\CALLT X{\sigma s}} T
      }
    \end{gather*}
    %
    %
    %
    %
  \end{multicols}
  \caption{Pushdown 
    types. Extends \cref{fig:recursive-abbr}; removes $\typec X$; adds
    $\CALLT Xs$. Right versions of rules \eqzl and \eqsl omitted.}
  \label{fig:pushdown-abbr}
\end{figure}


\paragraph{The pushdown world}

The next extension replaces natural numbers by finite sequences $\typec s$ of
symbols $\typec\sigma$ taken from a given stack alphabet. The details are in
\cref{fig:pushdown-abbr}. We use $\typec{\Empty}$ to denote the empty sequence. The extension from 1-counter is straightforward.
Parameters to type 
constructors are now
sequences of symbols, rather than natural numbers; all the rest remains the
same.
Once again, to show that $\istype \tstack$, we proceed coinductively.


\paragraph{Nested session types}
\begin{figure}[t!]
  \begin{multicols}{2}
  \declrel{New type contractiveness rules \inductive}{$\iscontrt T$}
  \begin{gather*}
    \tag\cts
    \frac{
      \ispeqt X{\overline\alpha}U
      \premspace
      \iscontrt {U\subs{\overline T}{\overline{\alpha}}}
    }{
      \iscontrt {\CALLNT X{\overline T}}
    }
  \end{gather*}
  \declrel{New type formation rules \coinductive}{$\istype{T}$}
  \begin{gather*}
    \tag\ws
    \frac{
      \ispeqt X{\overline\alpha}U
      \premspace
      \iscontrt {U\subs{\overline T}{\overline{\alpha}}}
      \premspace
      \istype {U\subs{\overline T}{\overline{\alpha}}}
    }{
      \istype {\CALLNT X{\overline T}}
    }
  \end{gather*}
  \declrel{New type equivalence rules \coinductive}{$\isequiv TT$}
  \begin{gather*}
    \tag\eqidl
    \frac{
      \ispeqt X{\overline\alpha}V
      \premspace
      \iscontrt {V\subs{\overline T}{\overline{\alpha}}}
      \premspace
      \isequiv {V\subs{\overline T}{\overline{\alpha}}} U
    }{
      \isequiv {\CALLNT X{\overline T}} U
    }
  \end{gather*}
  %
  \end{multicols}
  \caption{Nested types. Extends \cref{fig:recursive-abbr};
    removes $\typec X$; adds $\CALLNT X{\overline T}$. Right
    version of rule \eqidl omitted.}
  \label{fig:nested}
\end{figure}


A class of types that turns out to be equivalent to pushdown types was recently
proposed by Das et al.~\cite{DBLP:conf/esop/DasDMP21}.
The main idea is to have type constructors that are applied not to natural numbers or to sequences of symbols
but to types themselves; and to let type constructors take a variable (but fixed) number of parameters. The syntax rules for nested session types is given in \cref{fig:nested}, where $\typec\alpha$ denotes a variable on types, and
$\typec{\overline\alpha}$ a possibly empty sequence of variables (once again,
$\typec\Empty$ denotes the empty sequence).
There are two differences with respect to 1-counter and pushdown types
(\cref{fig:onecounter-abbr,fig:pushdown-abbr}): on the one hand type
constructors are now applied not to natural numbers or to sequences of symbols
but to types themselves; on the other hand, type constructors take a variable
(but fixed) number of parameters, so that each type constructor $\XT$ has an associated arity $n\in \nbb$. 
Type constructors are unfolded according to an equational definition of the form
$\iseqt {\CALL X {\alpha_1,\ldots,\alpha_n}} T$, where
$\typec{\alpha_1},\ldots,\typec{\alpha_n}$ are distinct type variables that
parameterise the type definition.


\paragraph{The 2-counter world}

2-counter types extend the 1-counter types by introducing equations
parameterised on two natural numbers, rather than one. The new rules are a
straightforward adaptation of those in \cref{fig:onecounter-abbr} for 1-counter
types and are thus omitted.
To show that $\istype \titer$, we proceed coinductively.


\paragraph{The infinite world}

The final destination takes us to arbitrary, coinductive, infinite types. The details are in
\cref{fig:finite-abbr}, except that all judgements not explicitly marked are
taken coinductively. 
No equations (of any sort) are needed, just plain
infinite types.
We also allow choices with an
infinite number of branches.

Infinite types arise by interpreting the syntax rules coinductively, which gives rise to potentially infinite chains of interactions. The structure of these arbitrary, coinductively defined, infinite types does not need to follow any pattern (e.g. it does not need to repeat itself), and arguably, the best way to think about these objects are as labelled infinite trees (\cref{sec:treeslanguages}). Such objects do not have in general a finite representation (or finite encoding), which can be shown by a simple cardinality argument (\cref{lem:hierarchyinfinite}). Hence the need for finding suitable subclasses of infinite types that can be represented and can be used in practice.

We can think of a type in two possible ways: as (one of) its representation(s), which is great for practical purposes as we can reason about types by reasoning about their representations; or as the underlying, possibly infinite, coinductive object which is being represented, which is suitable for developing a theory of types, in particular for comparing different models with one another.

\paragraph{Embedding context-free types into infinite types}

In order to compare context-free session types with the other classes in our
hierarchy, we must convert context-free types into infinite types. We do this by
defining an embedding $\embeds TU$ (\cref{fig:cfst}), where $\TT$ is a context-free session type
and $\UT$ is a corresponding infinite session type. The rules for the embedding
essentially unfold equational definitions, sequential composition, and
non-terminal occurrences of $\Skip$, until a lone $\Skip$ or one of the type
constructs $\typec?, \typec!, \typec\&, \typec\Oplus$ is found. This takes
finitely many steps due to contractiveness. Type $\End$ appears either from a
lone $\Skip$, or from a message $\MSGn T$ without a continuation.

\begin{theorem}[Embedding]
\label{lem:embedding}~
\begin{enumerate}
  \item If $\embeds TU$, then $\TT$ is a context-free type and $\UT$ is an infinite type.
  \item For every context-free type $\TT$, there exists $\UT$ with $\embeds TU$.
  \item Suppose $\embeds TU$ and $\embeds VW$. Then $\isequiv TV$ iff $\isequiv UW$.
\end{enumerate}
\end{theorem}

\begin{proof}
Sketched in \cref{sub:proof-embedding}.
\end{proof}

To be absolutely precise, we could explicitly define an embedding from each of the shades of types into the class of infinite (coinductive) types, in order to compare them with each other. However, for most cases this embedding is obvious and follows from the type formation rules. Only for context-free types, whose syntax is significantly different, did we feel the need to provide the rules for $\embeds TU$.







\section{Types, trees and traces}
\label{sec:treeslanguages}

It should be clear that the constructions defined in \cref{sec:types-procs} form
some sort of type hierarchy; this section studies the hierarchy. In any case,
every type lives in the largest universe; that of arbitrary, coinductively
defined, infinite types.

To each type one can associate a labelled infinite
tree~\cite{DBLP:journals/jfp/GapeyevLP02,DBLP:books/daglib/0005958}. This tree
can in turn be expressed by the language of words encoding its paths.
Let $\labels$ be the set of labels used in choice types. 
Following Pierce \cite[Definition 21.2.1]{DBLP:books/daglib/0005958}, a tree is
a partial function
$t \in (\{\dl,\cl\}\cup \labels)^* \rightarrow \{\Endl,?,!,\&_L,\oplus_L \mid L
\subseteq \labels\}$ subject to the following constraints (below, $\pi$ ranges over strings of symbols whereas $\sigma$ ranges over symbols):
  \begin{itemize}
  \item $t(\Empty)$ is defined;
  \item if $t(\pi\sigma)$ is defined, then $t(\pi)$ is defined;
  \item if $t(\pi)={}?$ or $t(\pi)={}!$, then $t(\pi\sigma)$ is defined for $\sigma\in\{\dl,\cl\}$ and undefined for all other $\sigma$;
  \item if $t(\pi)=\&_L$ or $t(\pi)=\Oplus_L$, then $t(\pi\sigma)$ is defined for $\sigma\in L$ and undefined for all other $\sigma$;
  \item if $t(\pi)=\Endl$, then $t(\pi\sigma)$ is undefined for all $\sigma$.
  \end{itemize}
The labels $\dl$ and $\cl$ are abbreviations for \emph{data} and \emph{continuation}, corresponding to the components of a session type.

If all sets $L$ in a tree are finite, the tree is finitely branching.
The tree generated by a (finite or infinite) type is coinductively defined as follows.
\begin{align*}
  \treeof(\MSG {T_\dl}{T_\cl})(\Empty) =&\ \sharp
  &&&
  \treeof(\choicet)(\Empty) =&\ \star_L
  \\
  \treeof(\MSG {T_\dl}{T_\cl})(\dl\pi) =&\ \treeof(\typec{T_\dl})(\pi)
  &&&
  \treeof(\choicet)(\ell\pi) =&\ \treeof(\typec{T_\ell})(\pi)
  \\
  \treeof(\MSG {T_\dl}{T_\cl})(\cl\pi) =&\ \treeof(\typec{T_\cl})(\pi)
  &&&
  \treeof(\End)(\Empty) =&\ \Endl
\end{align*}

A \emph{path} in a tree $t$ is a word obtained by combining the symbols in the
domain and the range of $t$. Given a symbol $\sigma\in\{?,!,\&_L,\oplus_L \mid L\subseteq \labels\}$ in the codomain of $T$ (but different from $\Endl$), and a symbol $\tau\in\{\dl,\cl\}\cup\labels$, let $\pair\sigma\tau$ denote the combination of both symbols, viewed as a letter over the alphabet $\{?,!,\&_L,\oplus_L \mid L\subseteq \labels\}\times(\{\dl,\cl\}\cup\labels)$. For simplicity in exposition, we often drop the angular brackets and the subscript $L$ on the label set, and write, for example, $?\cl$ instead of $\pair?\cl$, $\oplus l$ instead of $\pair{\oplus_L}{l}$, etc.

Given a string $\pi$ in the domain of a tree $t$, we can define the word
$\Path_t(\pi)$ recursively as $\Path_t(\Empty)=\Empty$ and
$\Path_t(\pi\tau)=\Path_t(\pi)\cdot\pair{t(\pi)}{\tau}$. We say that a string
\emph{$\pi$ is terminal wrt to $t$} if $t(\pi)=\Endl$. For terminal strings, we
can further define $\overline{\Path}_t(\pi)=\Path_t(\pi)\cdot\Endl$.

Finally, we can define the language of (the paths in) a tree $t$ as the set
  $\{\Path_t(\pi) \mid
  \pi\in\dom(t)\}\cup\{\overline{\Path}_t(\pi) \mid \pi\in\dom(t), \pi\text{ is terminal wrt }t\}$.
  The language of (the traces of) a type $\TT$, denoted by $\lcal(\TT)$, is the
  language of $\treeof(\TT)$.
  %
%
Note that the traces of types are defined over the following alphabet. 
\begin{equation}\label{eq:alphabet}
\Sigma = \{?,!,\&_L,\oplus_L \mid L\subseteq \labels\}\times(\{\dl,\cl\}\cup\labels)\cup\{\Endl\}
\end{equation}


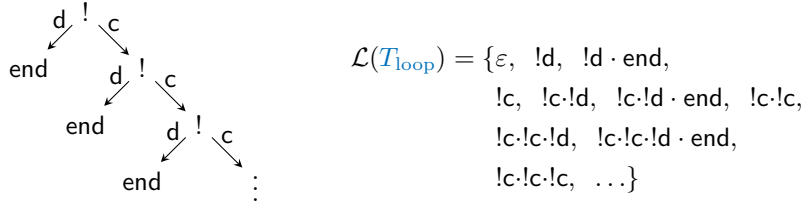
\begin{figure}[t]
  \begin{minipage}{0.45\textwidth}
    \begin{center}
      \begin{tikzpicture}[node distance = 3em]
        \node (bang1) {$!$};
        \node[below left of = bang1] (end1) {$\Endl$};
        \node[below right of = bang1] (bang2) {$!$};
        \node[below left of = bang2] (end2) {$\Endl$};
        \node[below right of = bang2] (bang3) {$!$};
        \node[below left of = bang3] (end3) {$\Endl$};
        \node[below right of = bang3] (dots) {$\vdots$};

        \draw (bang1) edge[above] node {$\dl$} (end1);
        \draw (bang1) edge[above] node {$\cl$} (bang2);
        \draw (bang2) edge[above] node {$\dl$} (end2);
        \draw (bang2) edge[above] node {$\cl$} (bang3);
        \draw (bang3) edge[above] node {$\dl$} (end3);
        \draw (bang3) edge[above] node {$\cl$} (dots);
      \end{tikzpicture}
    \end{center}
  \end{minipage}
  \begin{minipage}{0.50\textwidth}
    \begin{center}
      \begin{align*}
        \lcal(\tloop) = \{
        &\varepsilon,\;\;
          !\dl,\;\;
          !\dl\cdot\Endl,\;\;\\
        &!\cl,\;\;
          !\cl\cdot!\dl,\;\;
          !\cl\cdot!\dl\cdot\Endl,\;\;
          !\cl\cdot!\cl,\;\;\\
        &!\cl\cdot!\cl\cdot!\dl,\;\;
          !\cl\cdot!\cl\cdot!\dl\cdot\Endl,\;\;\\
        &!\cl\cdot!\cl\cdot!\cl,\;\;
          \ldots\}
      \end{align*}
    \end{center}
  \end{minipage}
  \caption{The tree and the language of type $\tloop$.}
  \label{fig:treeloop}
\end{figure}


\begin{figure}[t!]
  \begin{minipage}{0.25\textwidth}
    \begin{center}
      \vspace{2em}
      \begin{tikzpicture}[node distance = 3em]
        \node (choice1) {$\&$};
        \node[below left of = choice1, xshift=-1em] (choice2) {$\&$};
        \node[below right of = choice1, xshift=+1em] (end1) {$\Endl$};
        \node[below left of = choice2, xshift=-1em] (choice3) {$\&$};
        \node[below right of = choice2, xshift=+1.5em] (bang1) {$!$};
        \node[below left of = choice3] (dots) {$\vdots$};
        \node[below right of = choice3] (bang2) {$!$};
        \node[below left of = bang1] (end2) {$\Endl$};
        \node[below right of = bang1] (end3) {$\Endl$};
        \node[below left of = bang2] (end4) {$\Endl$};
        \node[below right of = bang2] (bang3) {$!$};
        \node[below left of = bang3] (end5) {$\Endl$};
        \node[below right of = bang3] (end6) {$\Endl$};

        \draw (choice1) edge[above] node[xshift=-0.2cm] {$\incl$} (choice2);
        \draw (choice1) edge[above] node[xshift=+3ex, yshift=-1ex] {$\dumpl$} (end1);
        \draw (choice2) edge[above] node[xshift=-0.2cm] {$\incl$} (choice3);
        \draw (choice2) edge[above] node[xshift=+3ex, yshift=-1ex] {$\dumpl$} (bang1);
        \draw (choice3) edge[above] node[xshift=-0.2cm] {$\incl$} (dots);
        \draw (choice3) edge[above] node[xshift=+3ex, yshift=-0.5ex] {$\dumpl$} (bang2);
        \draw (bang1) edge[above] node[xshift=-0.1cm, yshift=-0.5ex] {$\dl$} (end2);
        \draw (bang1) edge[above] node[xshift=+0.1cm, yshift=-0.5ex] {$\cl$} (end3);
        \draw (bang2) edge[above] node[xshift=-0.1cm, yshift=-0.5ex] {$\dl$} (end4);
        \draw (bang2) edge[above] node[xshift=+0.1cm, yshift=-0.5ex] {$\cl$} (bang3);
        \draw (bang3) edge[above] node[xshift=-0.1cm, yshift=-0.5ex] {$\dl$} (end5);
        \draw (bang3) edge[above] node[xshift=+0.1cm, yshift=-0.5ex] {$\cl$} (end6);
      \end{tikzpicture}
    \end{center}
  \end{minipage}
  \begin{minipage}{0.65\textwidth}
    \begin{center}
      \begin{align*}
        \lcal(\tcounter) = \{
        &\varepsilon,\;\;
          \&\incl,\;\;
          \&\dumpl,\;\;
          \&\dumpl\cdot\Endl,\\
        &\&\incl\cdot\&\incl,\;\;
          \&\incl\cdot\&\dumpl,\;\;
          \&\incl\cdot\&\dumpl\cdot!\dl,\\
        & \&\incl\cdot\&\dumpl\cdot!\dl\cdot\Endl,\;\;
          \&\incl\cdot\&\dumpl\cdot!\cl,\\
        & \&\incl\cdot\&\dumpl\cdot!\cl\cdot\Endl,\;\;
          \&\incl\cdot\&\incl\cdot\&\incl,\\
        &\&\incl\cdot\&\incl\cdot\&\dumpl,\;\;
          \&\incl\cdot\&\incl\cdot\&\dumpl\cdot!\dl,\\
        &\&\incl\cdot\&\incl\cdot\&\dumpl\cdot!\dl\cdot\Endl,\\
        &\&\incl\cdot\&\incl\cdot\&\dumpl\cdot!\cl,\\
        &\&\incl\cdot\&\incl\cdot\&\dumpl\cdot!\cl\cdot!\dl,\\
        &\&\incl\cdot\&\incl\cdot\&\dumpl\cdot!\cl\cdot!\dl\cdot\Endl,\;
        \ldots\}
      \end{align*}
    \end{center}
  \end{minipage}
  \caption{The tree and the language of type $\tcounter$.}
  \label{fig:treecounter}
\end{figure}


\begin{figure}[t!]
  \begin{minipage}{0.33\textwidth}
    \begin{center}
      \vspace{2em}
      \begin{tikzpicture}[node distance = 3em]
        \node (choice1) {$\&$};
        \node[below left of = choice1, xshift=-1em] (end1) {$\Endl$};
        \node[below right of = choice1, xshift=+1.5em] (choice2) {$\&$};
        \node[below left of = choice2, xshift=-3em] (in1) {$?$};
        \node[below right of = choice2, xshift=+3em] (choice3) {$\&$};
        \node[below left of = choice3, xshift=-0em] (in2) {$?$};
        \node[below right of = choice3, xshift=+0em] (dots1) {$\vdots$};
        \node[below left of = in1, xshift=-0em] (end2) {$\Endl$};
        \node[below right of = in1, xshift=+0em] (choice4) {$\&$};
        \node[below left of = choice4, xshift=-0em] (end3) {$\Endl$};
        \node[below right of = choice4, xshift=+0em] (dots2) {$\vdots$};
        \node[below left of = in2, xshift=-0em] (end4) {$\Endl$};
        \node[below right of = in2, xshift=+0em] (dots3) {$\vdots$};

        \draw (choice1) edge[above] node[xshift=-0.2cm] {$\leafl$} (end1);
        \draw (choice1) edge[above] node[xshift=+2ex, yshift=-0ex] {$\nodel$} (choice2);
        \draw (choice2) edge[above] node[xshift=-0.2cm] {$\leafl$} (in1);
        \draw (choice2) edge[above] node[xshift=+2ex, yshift=-0ex] {$\nodel$} (choice3);
        \draw (choice3) edge[above] node[xshift=-0.2cm] {$\leafl$} (in2);
        \draw (choice3) edge[above] node[xshift=+2ex, yshift=-0ex] {$\nodel$} (dots1);
        \draw (in1) edge[above] node[xshift=-0.1cm, yshift=-0.5ex] {$\dl$} (end2);
        \draw (in1) edge[above] node[xshift=+0.1cm, yshift=-0.5ex] {$\cl$} (choice4);
        \draw (choice4) edge[above] node[xshift=-0.2cm] {$\leafl$} (end3);
        \draw (choice4) edge[above] node[xshift=+2ex, yshift=-0ex] {$\nodel$} (dots2);
        \draw (in2) edge[above] node[xshift=-0.1cm, yshift=-0.5ex] {$\dl$} (end4);
        \draw (in2) edge[above] node[xshift=+0.1cm, yshift=-0.5ex] {$\cl$} (dots3);
      \end{tikzpicture}
    \end{center}
  \end{minipage}
  \begin{minipage}{0.65\textwidth}
      \begin{align*}
        \lcal(\ttree) = \{
        & \varepsilon,\;
          \&\leafl,\;
          \&\leafl\cdot\Endl,\;
          \&\nodel,\\
        & \&\nodel\cdot\&\leafl,\;
          \&\nodel\cdot\&\nodel,\\
        & \&\nodel\cdot\&\leafl\cdot?\dl,\;
          \&\nodel\cdot\&\leafl\cdot?\dl\cdot\Endl,\\
        & \&\nodel\cdot\&\leafl\cdot?\cl,\;
          \&\nodel\cdot\&\leafl\cdot?\cl\cdot\&\leafl,\\
        &\&\nodel\cdot\&\leafl\cdot?\cl\cdot\&\nodel,\\
        &\&\nodel\cdot\&\nodel\cdot\&\leafl,\;
        \&\nodel\cdot\&\nodel\cdot\&\nodel,\\
        &\&\nodel\cdot\&\nodel\cdot\&\leafl\cdot?\dl,\\
        &\&\nodel\cdot\&\nodel\cdot\&\leafl\cdot?\cl,\;
        \ldots\}
      \end{align*}
  \end{minipage}
  \caption{The tree and the language of type $\ttree$.}
  \label{fig:treetree}
\end{figure}


\Cref{fig:treeloop} depicts (a finite fragment of) the tree corresponding to
$\treeof(\tloop)$ (\cref{exa:rec}) and (some of the words in) its language $\lcal(\tloop)$.
Type $\tcounter$ (\cref{exa:onecounter}) describes an interaction that keeps
track of a counter.
Finite fragments of the corresponding tree and language are depicted in \cref{fig:treecounter}.
Type $\ttree$ (\cref{exa:context-free}) describes the reception of a binary tree
of $\End$ values.
Finite fragments of the corresponding tree and language are
depicted in \cref{fig:treetree}.


In the above examples, the language $\lcal(\TT)$ is closed under prefixes. This holds for a general type $\TT$, since elements of $\lcal(\TT)$ correspond to paths in $\treeof(\TT)$.

\begin{proposition}
  \label{prop:typeclosedprefix}
  $\lcal(\TT)$ is prefix closed, that is, if $w\in\lcal(\TT)$ and $u$ is a
  prefix of $w$, then $u\in\lcal(\TT)$.
\end{proposition}

Another immediate observation is that $\treeof$ (resp.\ $\lcal$) is an embedding from the class of all types to the class of all trees (resp.\ all languages).

\begin{proposition}\label{prop:typeequivalence}
  Let $\TT$ and $\UT$ be two types. 
  The following are equivalent:
\begin{enumerate}
  \item $\isequiv TU$;
  \item $\treeof(\TT)=\treeof(\UT)$;
  \item $\lcal(\TT)=\lcal(\UT)$.
\end{enumerate}
\end{proposition}

\Cref{prop:typeequivalence} tells us that two types are equivalent iff they have the same traces. Note that, in general, trace equivalence is a notion weaker than bisimulation \cite{sangiorgi:bisimulation-coinduction}. However, both notions coincide for deterministic transition systems. The syntax of (infinite) session types is in fact deterministic (\eg given a label $\ell$ for a choice, there can only be one type that continues from $\&\ell$), which explains our result.

\Cref{sec:types-procs} introduces eight classes of types. We now distinguish them
by means of subscripts:
finite types ($\istypef T$, \cref{fig:finite-abbr}),
recursive types ($\istyper T$, \cref{fig:recursive-abbr}),
1-counter types ($\istypeo T$, \cref{fig:onecounter-abbr}),
context-free types ($\istypec T$, \cref{fig:cfst}),
pushdown types ($\istyped T$, \cref{fig:pushdown-abbr}),
nested types ($\istypen T$, \cref{fig:nested}),
2-counter types ($\istypet T$) and 
coinductive, infinite types ($\istypei T$, \cref{fig:finite-abbr} with rules interpreted
coinductively).
To each class of types we introduce the corresponding class of languages. For example,
$\typesr$ is the set $\{\lcal(\TT)\mid \istyper T\}$. The strict hierarchy
result is as follows:
%
\begin{equation}\label{eq:mainchaininclusions}
\typesf \subsetneq \typesr \subsetneq \typeso \subsetneq \typesp \subsetneq
  \typest \subsetneq \typesi
\end{equation}
%
We remark that the last step in the chain of strict inclusions is obtained by a cardinality argument, since the set $\typesi$ is uncountable. This shows an even stronger statement: 
for any finite representation system (including the systems $\typesf$ to
$\typest$, as well as $\typescf$ and $\typesn$), there is an \emph{infinite, uncountable} set of types that cannot be represented by that system.

We now turn our attention to nested types ($\istypen{T}$) which turn out to be
%
equivalent to pushdown types, and further establish
equivalent sub-hierarchies inside both classes, parameterised by the
`complexity' of the corresponding representations. For pushdown session types, a
natural measure of complexity is the number of type constructors 
required to represent a given type. This number can be arbitrarily large, but
always finite. For a given $n\in\nbb$, we let $\typesp^n$ denote the subset 
corresponding to those types that can be represented with at most $n$ type constructors. When $n=0$, there are no constructors, and we can only represent finite types. As $n$ increases, so does the expressivity of our constructions, and we have the infinite chain of inclusions
\begin{equation*}
\typesf=\typesp^0\subsetneq\typesp^1\subseteq\typesp^2\subseteq\cdots\subseteq\typesp.
\end{equation*}

Similarly, for nested session types we can define a hierarchy by looking at the arities of the type constructors used. For a given $n\in\nbb$, we let $\typesn^n$ denote the subset 
corresponding to the nested session types 
whose type constructors have arity at most $n$. When $n=0$ all type constructors are constant, and we recover the class of recursive types. As $n$ increases, so does the expressivity, and we also have an infinite chain of inclusions
\begin{equation*}
\typesr=\typesn^0\subsetneq\typesn^1\subseteq\typesn^2\subseteq\cdots\subseteq\typesn.
\end{equation*}

It turns out that these hierarchies are one and the same (with the exception of
the bottom level), so that we have (\cref{sub:contextfreenested})
\begin{equation}\label{eq:chainpushdownnested}
  \typesf = \typesp^0 \subsetneq \typesr = \typesn^0 \subsetneq \typesp^1 = \typesn^1 \subseteq \typesp^2 = \typesn^2 \subseteq \cdots \subseteq \typesp = \typesn.
\end{equation}

Higher-order context-free types
(denoted by $\typescf$)
lie between levels $0$ and $1$ in the sub-hierarchies above, \ie they can represent recursive types, and can be
represented by pushdown session types using at most one type constructor, or
equivalently, by nested session types with either constant or unary type
constructors, so that we have (\cref{sub:contextfreenested})
\begin{equation}\label{eq:chaincontextfree}
  \typesr \subsetneq \typescf \subsetneq \typesp^1 = \typesn^1.
\end{equation}


Regarding the inclusion $\typescf \subsetneq \typesp^1$, we actually have a stronger observation. Context-free session types are included in pushdown session types
which have only one type constructor $\XT$, and where
the equation $\iseqt {\CALL X \varepsilon} \End$ accounts for the only occurrence of $\End$. The latter means that the type ends
iff the state $\CALLT X \varepsilon$ is reached, that is, iff the stack is
empty. Thus, we can intuitively think of context-free session types as
pushdown types with a single constructor and an empty stack
acceptance criterion.
This observation points to the fact that the qualifier `context-free' in the so
called context-free session types is a misnomer, a remark that is not unheard of
\cite{DBLP:conf/esop/DasDMP21}.

The hierarchy that puts in context all the classes of types studied in this
paper is summarized in the result below.

\newcommand\sepex{0.4ex}

\begin{theorem}[Inclusions]
  \label{thm:inclusions}
  \
  \begin{center}
\begin{tabular}{p{\sepex}p{\sepex}p{\sepex}p{\sepex}p{\sepex}p{\sepex}p{\sepex}p{\sepex}p{\sepex}p{\sepex}p{\sepex}p{\sepex}p{\sepex}p{\sepex}p{\sepex}p{\sepex}p{\sepex}p{\sepex}p{\sepex}p{\sepex}p{\sepex}p{\sepex}p{\sepex}p{\sepex}p{\sepex}}
      $\typesf$ & $=$ & $\typesp^0$ & $\subsetneq$ & $\typesr$ & $=$ & $\typesn^0$
      &&& $\subsetneq$ &&& $\typeso$ &&& $\subsetneq$ &
      & $\typesp$ & $=$ & $\typesn$ & $\subsetneq$ & $\typest$ & $\subsetneq$ & $\typesi$
      \\
                &&&&&& \rotatebox[origin=c]{270}{$\subsetneq$}
                                      &&&&&&&&&&& \rotatebox[origin=c]{90}{$\subseteq$}
      \\
                &&&&&& $\typescf$ & $\subsetneq$ & $\typesp^1$ & $\ =$ & $\typesn^1$
                                                                         &
                                                                           $\subseteq$ & $\typesp^2$ & $\ =$ & $\typesn^2$ & $\subseteq$ & \;$\cdots$
    \end{tabular}
  \end{center}
\end{theorem}


\section{From types to automata}
\label{sec:systemstoautomata}


This section describes procedures to convert types in different levels of the
hierarchy (recursive systems, 1-counter, pushdown and 2-counter) into automata
at the same level. All constructions follow the same guiding principles, so we
focus on the bottom level of the hierarchy (recursive systems) and then
highlight the main differences as we advance in the hierarchy.

All automata that we consider in this paper are \emph{deterministic} and
\emph{total}, \ie the transition functions are such that any input word has a
well-defined, unique computation path. We use the alphabet $\Sigma$ defined
in \eqref{eq:alphabet}. As standard references in automata theory we
mention the book by Hopcroft and Ullman \cite{hopcroftullman:1979} and Valiant's PhD thesis \cite{valiant:1973:phdthesis}.

\paragraph{Recursive types and finite-state automata}

Following the usual notation, a (deterministic) \emph{finite-state automaton} is
given by a set $Q$ of states, with a specified initial state $q_0\in Q$, a 
transition function $\delta: Q\times\Sigma\rightarrow Q$, and a set $A\subseteq Q$ of accepting states. Given a finite word $a_1a_2\cdots a_n$, its execution by the automaton yields the sequence of states $s_0,s_1,\ldots,s_n$ where $s_0=q_0$ and $s_{i+1}=\delta(s_i,a_{i+1})$.
We say that a word is accepted by the automaton if its execution ends in an accepting state.


Suppose we are given a system of recursive equations $\{\iseqt{X_i}{T_i} \}_{i\in I}$ over a
variable set $\xcal = \{\XT,\YT,\ldots\}$. Our first step is to
convert this system into a normal form in which every right-hand side is either a
variable $\XT$, or a single application of one of the type constructors, \ie one
of $\End$, $\IN XY$, $\OUT XY$, $\extchoice\recordt \ell X L$ or
$\intchoice\recordt \ell X L$. We can do this by introducing fresh, intermediate
variables as needed. Essentially, whenever we have an equation $\iseqt{X}{\IN{T_1}{T_2}}$ where $\typec{T_1}$, $\typec{T_2}$ are not variables, we add two new variables $\typec{X'}$, $\typec{X''}$, replace the above equation by $\iseqt{X}{\IN{X'}{X''}}$, and add two new equations $\iseqt{X'}{T_1}$ and $\iseqt{X''}{T_2}$. The process is the same for the other type constructors. By doing this repeatedly, we ``break down'' a long equation into many small equations. The number of new variables is linear in the encoding size of the original representation. 

Given such a system, we construct a finite-state automaton (over the alphabet $\Sigma$) as follows. The
automaton has a state $q_X$ for every type variable $\XT$, and two additional
states: an `end' state $q_{\Endl}$ and an `error' state $q_{\errorl}$. The
transitions from $q_{\errorl}$ are described by
$q_{\errorl}\overset{a}{\rightarrow}q_{\errorl}$ for every symbol $a$.
Similarly, the transitions at $q_{\Endl}$ are described by
$q_{\Endl}\overset{a}{\rightarrow}q_{\errorl}$ for every symbol $a$. The
transitions at state $q_X$ are given by the corresponding equation for variable
$\XT$, in the obvious way. Some examples:

\begin{itemize}
\item Suppose our system contains the equation $\iseqt XY$. Then we have an
  $\varepsilon$-transition given by $q_X\overset{\varepsilon}{\rightarrow}q_Y$.
\item Suppose our system contains the equation $\iseqt X {\OUT YZ}$. Then we
  have the reading transitions $q_X\overset{!\dl}{\rightarrow}q_Y$,
  $q_X\overset{!\cl}{\rightarrow}q_{Z}$, and
  $q_X\overset{a}{\rightarrow}q_{\errorl}$ for any $a\neq\ !\dl,!\cl$.
\item Suppose our system contains the equation
  $\iseqt X {\Oplus\{l\colon X, m\colon Y\}}$. Then we have the reading
  transitions $q_X\overset{\Oplus l}{\rightarrow}q_{X}$,
  $q_X\overset{\Oplus m}{\rightarrow}q_{Y}$ and
  $q_X\overset{a}{\rightarrow}q_{\errorl}$ for any
  $a\neq \Oplus l,\Oplus m$.
\item Suppose our system contains the equation $\iseqt X\End$. Then we have the reading moves $q_X\overset{\Endl}{\rightarrow}q_{\Endl}$ and $q_X\overset{a}{\rightarrow}q_{\errorl}$ for any $a\neq\Endl$.
\end{itemize}

We define all states other than $q_\errorl$ to be accepting states.\footnote{We need all states to be accepting, since we might need to look at finite traces to distinguish between two types. For example, $\iseqt X {\&\{\al\colon  X\}}$ and $\iseqt Y {\&\{\bl\colon  Y\}}$ define non-equivalent types that have no finite terminating paths.} Notice that
the finite-state automaton described above is an automaton with possible
$\varepsilon$-moves. Although, by definition, deterministic finite-state automata do not permit $\varepsilon$-moves, in our case paths of $\varepsilon$-moves are uniquely determined and either reach a state without outgoing $\varepsilon$-transitions, or become stuck in a loop\footnote{In this case, the system of equations is not contractive and does not define a type.}. We can convert the given automaton into an equivalent automaton without $\varepsilon$-moves by `shortcutting' such moves. Formally, suppose a state $\typec{X}$ has an outgoing $\varepsilon$-transition to $\typec{Y}$; by construction, it is $\typec{X}$'s only outgoing transition. Assuming $\typec{X}$ and $\typec{Y}$ are different states, we can change every transition entering $\typec{X}$ and make it enter $\typec{Y}$ instead; finally, we can remove state $\typec{X}$ (hence removing the $\varepsilon$-transition from $\typec{X}$). If $\typec{X}$'s outgoing $\varepsilon$-transition loops to itself, we can just remove this transition and treat $\typec{X}$ as a state from which no transitions are possible.

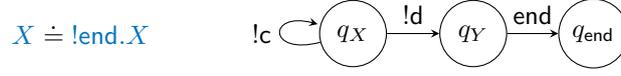
\begin{figure}[t!]
\centering
\begin{tikzpicture}
\node[state] (q0) {$q_X$};
\node[state, right of = q0] (q1) {$q_Y$};
\node[state, right of = q1] (q2) {$q_\Endl$};
\draw (q0) edge[loop left] node{$!\cl$} (q1);
\draw (q0) edge[above] node{$!\dl$} (q1);
\draw (q1) edge[above] node{$\Endl$} (q2);
\node[left of = q0, xshift = -2cm]{$\iseqt X {\OUT \End X}$};
\end{tikzpicture}
\caption{An automaton that defines $\tloop$ with initial state $q_X$. All depicted states are accepting.\label{fig:automuxendx}}
\end{figure}


We show in \cref{fig:automuxendx} the automaton that corresponds to type $\tloop$ (\cref{exa:rec}). Here we adopt the convention that every missing transition points to $q_{\errorl}$ which is not shown. In our examples, all depicted states are accepting, so we omit the usual double circle notation.

\paragraph{1-counter types and automata}
We augment the definition of finite-state automata into the definition of
1-counter automata as follows. We now have a partially defined transition
function $\delta :
Q\times\{\zerol,\succl\}\times(\{\varepsilon\}\cup\Sigma)\rightarrow
\{=,+,-\}\times Q$. The first argument of $\delta$ corresponds to the current
machine state. The second argument of $\delta$ indicates whether the counter
currently has value zero ($\zerol$) or some positive number ($\succl$). Note
that we cannot directly read the counter value, only whether it is non-zero. The
third argument can be either a symbol in $\Sigma$ (which is used for reading
moves), or $\varepsilon$ (which is used for $\varepsilon$-moves). The output of
$\delta$ is given by a new machine state, and additionally, a counter operation,
which can be either $=$ (no change), $+$ (increment by one) or $-$ (decrement by
one).\footnote{Of course, one has to be careful with the operation of
  decrementing when the counter value is zero. One can exclude such
  possibilities at the syntactic level, which is the case if automata are built
  from well formed 1-counter types.} We are solely interested in \emph{deterministic, total} transition functions, meaning that for each combination $(q,t)\in Q\times\{\zerol,\succl\}$, either
\begin{itemize}
\item $\delta(q,t,\varepsilon)$ is undefined, and $\delta(q,t,a)$ is defined for all $a\in\Sigma$ (so-called reading mode) or
\item $\delta(q,t,\varepsilon)$ is defined, and $\delta(q,t,a)$ is undefined for all $a\in\Sigma$ (so-called $\varepsilon$-mode).
\end{itemize}

Intuitively, at a reading mode we
must read the next input symbol, whereas at an $\varepsilon$-mode we cannot read the next input symbol (but we can change the value of the counter and the current state). A configuration is given by a pair $(q,n)\in
Q\times \mathbb{N}$, where $q$ denotes the current state and $n$ the current
value of the counter. Given $w$ in $\Sigma^*$, a derivation $(q,n) \overset{w}{\rightarrow} (q',n')$ is a sequence of moves specified by the transition rules, that leads from $(q,n)$ to $(q',n')$, and, in the process, reads the word $w$. Note that, for the same word $w$, there might be several configurations $(q',n')$ for which $(q,n) \overset{w}{\rightarrow} (q',n')$; all these lie in a unique path of $\varepsilon$-moves.

Similarly to finite-state automata, the semantics of 1-counter automata are given by a set $A$ of accepting states and an initial configuration $(q_0,n_0)$. A finite word $w$ is accepted if there is an accepting state $q_f$ and a natural number $n_f$ for which $(q_0,n_0)\overset{w}{\rightarrow} (q_f,n_f)$.

We should remark that our model is phrased in a slightly different manner from other formulations \cite{valiantpaterson:1975:onecounterautomata} that describe one-counter automata as pushdown automata with a single stack symbol, allow for increments of more than one unit in a single step, etc. Our formulation makes the parallel between types and automata somewhat more evident, and simplifies some of the proofs. It should be clear that our formulation is equivalent to the standard formulations, \ie one can easily convert between them.

We now explain how to convert 1-counter session types into 1-counter automata. Instead of non-parameterised variables our equations now involve terms of the
form $\CALLT X \zero$, $\CALLT X{\succ\zero}$, $\CALLT X N$,
$\CALLT X {\succ N}$, etc. We assume for simplicity that the variables appearing
in these equations are restricted in the following way: if the left-hand side of
an equation is of the form $\CALLT X\zero$, then the variables appearing in the
right-hand side must be of the form $\CALLT {X'} \zero$ or
$\CALLT {X'}{\succ\zero}$ (with $\typec{X'}$ possibly different from $\typec X$); and
if the left-hand side of an equation is of the form $\CALLT X {\succ N}$, then
the variables appearing in the right-hand side must be of the form
$\CALLT {X'} N$, $\CALLT {X'} {\succ N}$ or $\CALLT {X'} {\succ {\succ N}}$. Any
system can be converted into this form by adding finitely many new equations.
For example, $\CALLT X\zero \doteq \CALLT Y{\succ{\succ{\succ\zero}}}$ can be
rewritten as
%
\begin{align*}
  \iseqt{\CALLT X\zero}{\CALLT {X'}{\succ\zero}}
  &&
     \iseqt{\CALLT {X'}{\succ N}}{\CALLT {X''}{\succ {\succ N}}}
  &&
     \iseqt{\CALLT {X''}{\succ N}}{\CALLT Y {\succ{\succ N}}}
\end{align*}
and $\iseq{\CALLT X{\succ N}}{\CALLT Y\zero}$ can be rewritten as
\begin{align*}
  \iseqt{\CALLT X{\succ N}}{\CALLT {X'}N}
  &&
     \iseqt{\CALLT {X'}{\succ N}}{\CALLT {X'}{N}}
  &&
     \iseqt{\CALLT {X'}\zero}{\CALLT Y\zero}.
\end{align*}

We can convert a 1-counter type into a (deterministic) 1-counter automaton, so
that the transition function depends on whether the counter value is zero
(corresponding to a right-hand side of the form $\CALLT X \zero$) or positive
(corresponding to a right-hand side of the form $\CALLT X {\succ N}$).
Furthermore, the changes in the counter value along the variables are
incorporated by changes in the counter value along the automaton. For example,
take equation $\iseqt{\CALLT X {\succ N}}{\CALLT Y N}$. The corresponding
transition from $(q_X,\succl,\varepsilon)$ to $q_Y$ decrements the counter.

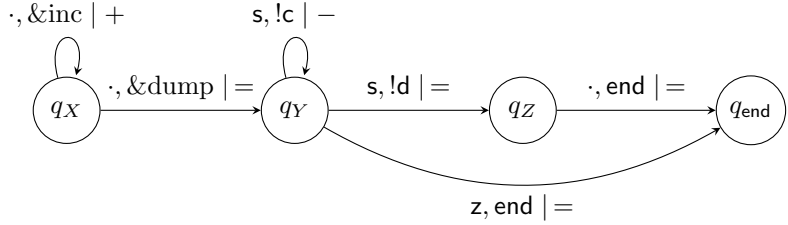
\begin{figure}[t!]
\centering
\begin{align*}
    \lhst X \zero \Eq&\; \typec{\& \{\incl\colon \CALLT X {\succ\zero}, \dumpl\colon \CALLT Y\zero\}}
    &&&
        \lhst Y \zero \Eq&\; \End
    \\
    \lhst X {\succ N} \Eq&\; \typec{\& \{\incl\colon \CALLT
                           X{\succ{\succ N}}, \dumpl\colon \CALLT Y{\succ N}\}}
    &&&
        \lhst Y {\succ N} \Eq&\; \OUT \End {\CALLT Y N}
 \end{align*}
\begin{tikzpicture}[node distance = 3cm]
\node[state] (q0) {$q_X$};
\node[state, right of = q0] (q1) {$q_Y$};
\node[state, right of = q1] (q2) {$q_Z$};
\node[state, right of = q2] (q3) {$q_\Endl$};

\draw (q0) edge[loop above] node {$\cdot,\&\incl \mid +$} (q0);
\draw (q0) edge[above] node {$\cdot,\&\dumpl \mid\,=$} (q1);
\draw (q1) edge[loop above] node {$\succl,!\cl \mid -$} (q1);
\draw (q1) edge[above] node {$\succl,!\dl \mid\,=$} (q2);
\draw (q2) edge[above] node {$\cdot,\Endl \mid\,=$} (q3);
\draw (q1) edge[bend right, below] node {$\zerol,\Endl \mid\,=$} (q3);


\end{tikzpicture}
\caption{A 1-counter automaton for type $\tcounter = \CALLT X \zero$. The initial configuration is $(q_X,0)$. Here a transition $\delta(q,t,a)=(o,q')$ is denoted by an arc from $q$ to $q'$ with label $t,a \mid o$, where $t\in\{\zerol,\succl\}$, $a\in\{\varepsilon\}\cup\Sigma$, and $o\in\{=,+,-\}$. If both $t=\zerol$ and $t=\succl$ lead to the same transition, then we use the symbol $\cdot$ to refer to both transitions. All depicted states are accepting, and any transition which is not depicted leads to a non-accepting sink state.\label{fig:autotcounter}}
\end{figure}

For illustration purposes, we show how to construct a 1-counter automaton
accepting $\lcal(\tcounter)$ from \cref{exa:onecounter}.
%
%
  First, we need to convert the equation for $\CALLT Y {\succ N}$ into normal form. We add an extra variable $\typec{Z}$ and write

  \begin{align*}
    \lhst X \zero \Eq&\; \typec{\& \{\incl\colon \CALLT X {\succ\zero}, \dumpl\colon \CALLT Y\zero\}}
    &&&\hspace{-1.7em}
    \lhst X {\succ N} \Eq&\; \typec{\& \{\incl\colon \CALLT
                           X{\succ{\succ N}}, \dumpl\colon \CALLT Y{\succ N}\}}
    \\
    \lhst Y \zero \Eq&\; \End
    &&&\hspace{-1.7em}
    \lhst Y {\succ N} \Eq&\; \OUT {\CALLT Z {\succ N}} {\CALLT Y N}
    \\
    \lhst Z \zero \Eq&\; \End
    &&&\hspace{-1.7em}
    \lhst Z {\succ N} \Eq&\; \End
  \end{align*}
  The corresponding automaton has states $q_X,q_Y,q_Z$, one for for each type constructor $\XT,\YT,\typec{Z}$, as well as an additional state
  $q_\Endl$. The outgoing transitions for state $q_X$ are the same regardless of
  the counter value: either read $\&\incl$, incrementing the counter and staying
  in $q_X$; or read $\&\dumpl$, keeping the counter value and moving to $q_Y$.
  For state $q_Y$, if the counter is zero, we can read $\Endl$ while moving to
  state $q_\Endl$. On the other hand, if the counter is non-zero, we can read
  $!\dl$, keeping the counter value and moving to $q_Z$; or read $!\cl$,
  decrementing the counter value and staying in $q_Y$. Finally, for state $q_Z$
  we can only read $\Endl$ and move to state $q_\Endl$. Note that whatever we
  choose to write on the equation for $\CALLT Z \zero$ is irrelevant, as this
  configuration is unreachable. Putting all these together, we arrive at the
  automaton in \cref{fig:autotcounter}.

\paragraph{Pushdown types and automata}

Just as the notion of 1-counter automata allows us to define a new class of types that extends the regular types, we can use (deterministic) pushdown automata to obtain the next class in our hierarchy of types. The main difference between pushdown automata and 1-counter automata is the ability of using a stack of symbols over a finite stack alphabet instead of a counter (which can be thought of as a stack over a singleton alphabet). 

We use $\Delta$ to denote a finite \emph{stack alphabet}. The contents of a
stack are denoted by a word $\omega\in\Delta^\ast$, with $\varepsilon$
representing an empty stack. We follow the convention that the first (leftmost)
symbol in $\omega$ corresponds to the top symbol of the stack. For ease of
notation, let $\mathrm{Op} = \{+\sigma:\sigma\in\Delta\} \cup \{=,-\}$ denote
the different stack operations (push a symbol $\sigma$ onto the stack, do
nothing, or pop the stack). In a (deterministic) pushdown automaton, we have a
partial-valued transition function
$\delta:Q\times(\{\varepsilon\}\cup\Delta)\times(\{\varepsilon\}\cup\Sigma)\rightarrow\mathrm{Op}\times
Q$. The transition function takes as input the current state, the current top
symbol of the stack (or an indication that the stack is empty), and the next
character of the word to be read (or an indication of an $\varepsilon$-move).
The output of the transition function is composed of a stack operation and the
next state.\footnote{Similarly to 1-counter automata, we can syntatically
  exclude the possibility that $\delta$ outputs a pop operation when the stack
  is empty by building automata from well formed pushdown types.} We are solely interested in \emph{deterministic, total} transitions, which mean that at each combination $(q,\sigma)\in Q\times (\{\varepsilon\}\cup\Delta)$ we can either only perform an $\varepsilon$-move ($\delta(q,\sigma,a)$ is undefined for all $a\in\Sigma$), or only perform reading moves ($\delta(q,\sigma,\varepsilon)$ is undefined). A configuration is given by a pair $(q,\omega)\in Q\times \Delta^\ast$, where $q$ denotes the current state and $\omega$ the current contents of the stack. In a similar way to 1-counter automata, we can define the notion of a derivation $(q,\omega)\overset{w}{\rightarrow}(q',\omega')$ as a sequence of moves going from $(q,\omega)$ to $(q',\omega')$ while reading the word $w$ over the input symbols. Again, we observe that several configurations may be derived from the same input word $w$, and that they belong in a unique path of $\varepsilon$-moves.

Finally, the semantics of a pushdown automata is given by a set $A\subseteq Q$ of accepting states, and an initial configuration $(q_0,\omega_0)$. A word $w$ is accepted if there is an accepting state $q_f$ and a stack word $\omega_f$ for which $(q_0,\omega_0)\overset{w}{\rightarrow}(q_f,\omega_f)$. A \emph{deterministic context-free language} (DCFL) is a language accepted by a deterministic pushdown automaton.

Again, we should remark the ways in which our formulation differs from the
standard \cite{hopcroftullman:1979}: we allow the transition function to
be defined on an empty stack, but we forbid pushing multiple stack symbols in a single transition. However, one can easily convert between formulations by adding extra symbols and states.

Pushdown systems act in a similar manner, but now the behaviour of a variable is specified by $|\Delta|+1$ equations, where $\Delta$ is the stack alphabet; one equation for each possible symbol at the top of the stack, and one equation for the case that the stack is empty. Accordingly, we use a (deterministic) pushdown automaton to simulate the stack contents by means of push and pop operations. The transitions from a state $q_X$ and a given stack indicator in $\{\varepsilon\}\cup\Delta$ are once more given by the corresponding equation with $\XT$ as the type constructor on the left-hand side. \cref{fig:autotstack} shows a pushdown automaton
accepting $\lcal(\tstack)$.
%

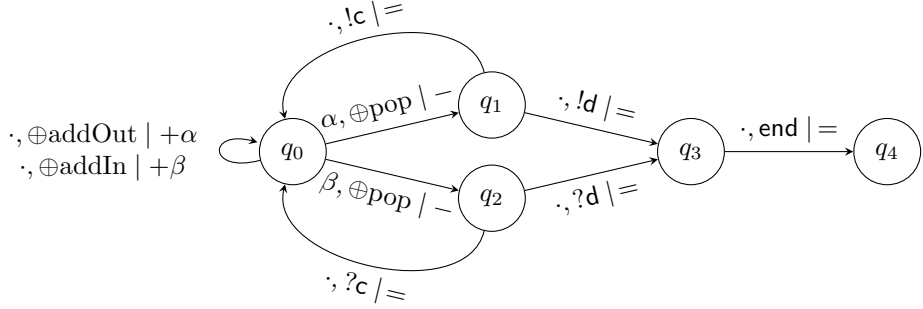
\begin{figure}[t!]
\centering
  \begin{align*}
    \lhst X \Empty \Eq&\; \typec{\&\{
                        \pushal\colon \CALLT X \sigma,
                        \pushbl\colon \CALLT X \tau\}}
    \\
    \lhst X {\sigma S} \Eq&\; \typec{\&\{
                            \pushal\colon \CALLT X {\sigma\sigma S},
                            \pushbl\colon \CALLT X {\tau\sigma S},
                            \popl\colon \OUT \End {\CALLT X {S}}\}}
    \\
    \lhst X {\tau S} \Eq&\; \typec{\&\{
                           \pushal\colon \CALLT X {\sigma\tau S},
                           \pushbl\colon \CALLT X {\tau\tau S},
                           \popl\colon \IN \End {\CALLT X {S}}\}}
  \end{align*}
\begin{tikzpicture}

\node[state] (q0) {$q_0$};
\node[state, above right of = q0, xshift = 1.5cm, yshift = -0.5cm] (q1) {$q_1$};
\node[state, below right of = q0, xshift = 1.5cm, yshift = +0.5cm] (q2) {$q_2$};
\node[state, below right of = q1, xshift = 1.5cm, yshift = +0.5cm] (q3) {$q_3$};
\node[state, right of = q3, xshift = 1cm] (q4) {$q_4$};

\draw (q0) edge[loop left, left] node {$\begin{array}{c}
\cdot,\oplus\pushal \mid +\alpha\\
\cdot,\oplus\pushbl \mid +\beta
\end{array}$} (q0);
\draw (q0) edge[above] node[sloped] {$\alpha,\oplus\popl \mid -$} (q1);
\draw (q0) edge[below] node[sloped] {$\beta,\oplus\popl \mid -$} (q2);
\draw (q1) edge[above] node[sloped] {$\cdot,!\dl \mid\,=$} (q3);
\draw (q2) edge[below] node[sloped] {$\cdot,?\dl \mid\,=$} (q3);
\draw (q3) edge[above] node {$\cdot,\Endl \mid\,=$} (q4);
\draw (q1) edge[above, bend right=90] node[sloped] {$\cdot,!\cl \mid\,=$} (q0);
\draw (q2) edge[bend left=90, below] node[sloped] {$\cdot,?\cl \mid\,=
  $} (q0);



\end{tikzpicture}
\caption{A pushdown automaton for type $\tstack = \CALLT X\Empty$. The initial configuration is $(q_0,\varepsilon)$. Here a transition $\delta(q,t,a)=(o,q')$ is denoted by an arc from $q$ to $q'$ with label $t,a \mid o$, where $t\in\{\varepsilon\}\cup\Delta$, $a\in\{\varepsilon\}\cup\Sigma$, and $o\in\mathrm{Op}$. If every possible choice of $t$ leads to the same transition, then we use the symbol $\cdot$ to refer to all possible transitions. All depicted states are accepting.
\label{fig:autotstack}}
\end{figure}


%

\paragraph{2-counter types and automata}

For the final step in our hierarchy we could think of extending the number of counters, or the number of stacks, of the representation models presented above. It should be clear by now that we would thus establish a correspondence from a type $\TT$ having a representation in terms of, say, $k$ counters, to a (deterministic) automaton with $k$ auxiliary counters accepting the language $\lcal(\TT)$. However, we know that this hierarchy collapses after $k=2$ in the Turing machine model \cite[Chapter 7]{hopcroftullman:1979}. That is, any language that is decidable (in the usual sense of the word) is accepted by a 2-counter automaton (and hence, also by a 2-stack pushdown automaton).

The translation to 2-counter automata is as for the 1-counter case, but now the
behaviour is specified by one of four different cases, depending on which of the
two counters is zero or non-zero. Accordingly, we use a (deterministic)
2-counter automaton with the appropriate transition function.

\section{From automata to types}
\label{sub:automatatosystems}

The constructions in \cref{sec:systemstoautomata} explain how, given a system of equations at some level in the hierarchy, we can construct a corresponding automaton. If $\istyped {\CALL X \sigma}$, then the language of the type given by $\CALLT X \sigma$ is the language accepted by the automaton with initial configuration $(q_X,\sigma)$ (and similarly for recursive, 1-counter, and 2-counter types).
Conversely, given an automaton which is promised to accept the language of traces of a type, we can construct the corresponding system of equations that specifies that type. This allow us to obtain a complete correspondence between classes of types and different models of computation based on automata theory.

Let us begin with the following observation. From \cref{prop:typeclosedprefix}, we know that the language of a type is prefix-closed. Furthermore, the construction in \cref{sec:systemstoautomata} gives rise to automata with the following interesting property: they have exactly one non-accepting state ($q_{\mathrm{error}}$), from which one cannot escape (all transitions from $q_{\mathrm{error}}$ lead to $q_{\mathrm{error}}$). It should be obvious that automata with such a property accept prefix-closed languages.

\begin{definition}\label{def:obviouslyprefixclosed}
An automaton is said to be obviously prefix-closed if it has exactly one non-accepting state, and this state is a sink.
\end{definition}

If our given automaton is obviously prefix-closed, and accepts the language of a type, it is straightforward to retrieve from its description the equivalent system of equations (as we shall see in this section, after \cref{thm:obviouslyprefixclosed}). However, what if the given automaton is not obviously prefix-closed, but it still promised to accept a prefix-closed language? We answer this question by showing how to convert such an automaton into an equivalent automaton which is obviously prefix-closed. For the case of finite-state automata the proof is straightforward (see for example Kao \etal~\cite{kaoetal:2009}); for the remaining three classes of automata, this is (to the best of our knowledge) a novel contribution of our paper.

As a first stage in our construction, we convert a given automaton into an equivalent automaton in the following form.

\begin{definition}\label{def:normalformautomaton}
An automaton (with initial configuration $c_0$) is said to be in normal form if it satisfies the following two properties.
\begin{itemize}
  \item guaranteed to read: for any input word $w$, there exists some configuration $c'$ for which $c_0\overset{w}\rightarrow c'$;
  \item immediate acceptance: for any input word $w$, let $c'$ be the first configuration for which $c_0\overset{w}\rightarrow c'$. Then $w$ is accepted by the automaton iff $c'$ is an accepting configuration.
\end{itemize}
\end{definition}

The first property (guaranteed to read) intuitively means that
the automaton cannot get stuck in an infinite sequence of $\varepsilon$-moves. Immediate acceptance means that we can decide whether an input word $w$ is accepted by the automaton immediately after reading its last symbol.

Of course, many different ``normal forms'' of automata have been adopted in the
literature. One which is particularly close to ours appears in
Valiant's PhD thesis~\cite{valiant:1973:phdthesis}; in his normal form (for 1-counter and pushdown
automata), the decision of acceptance is postponed until the last (as opposed to
the first) configuration $c'$ for which $c_0\overset{w}\rightarrow c'$
(implying, in other words, that all accepting states must correspond to reading
modes). A consequence of the following result is that these are all equivalent
automata formulations.

\begin{theorem}[Normal form automata]
\label{thm:normalformautomaton}~
\begin{itemize}
  \item Any finite-state automaton can be converted into an equivalent normal form automaton.
  \item Any 1-counter automaton can be converted into an equivalent normal form automaton.
  \item Any pushdown automaton can be converted into an equivalent normal form automaton.
  \item Any decidable language is accepted by a 2-counter normal form automaton.
\end{itemize}
\end{theorem}

\begin{proof}
In \cref{sub:proofs-automatatotypes}.
\end{proof}

With the above characterisation, we are now able to prove that prefix-closed languages can be assumed to be accepted by obviously prefix-closed automata.

\begin{theorem}\label{thm:obviouslyprefixclosed}~
\begin{itemize}
\item Every prefix-closed regular language is accepted by an obviously prefix-closed finite-state automaton.
\item Every prefix-closed language accepted by a one-counter automaton is accepted by an obviously prefix-closed one-counter automaton.
\item Every prefix-closed DCFL is accepted by an obviously prefix-closed pushdown automaton.
\item Every prefix-closed decidable language is accepted by an obviously prefix-closed two-counter automaton.
\end{itemize}
\end{theorem}

\begin{proof} The proof is identical in all four cases. Let $L$ be a language fitting into one of the above four cases, and without loss assume $L\neq\emptyset$. Let $A$ be an automaton accepting $L$. Due to \cref{thm:normalformautomaton}, $A$ can be assumed to be in normal form. We now construct an automaton $A'$ by modifying $A$ as follows.

\begin{itemize}
\item $A'$ has a fresh, non-accepting state $q_{\mathrm{error}}$; every configuration associated with $q_{\mathrm{error}}$ is a reading configuration for which reading $a$ moves again to $q_{\mathrm{error}}$, for every input symbol $a$;
\item Let $c\overset{a}{\rightarrow}c'$ be a reading move in $A$, and $q'$ the state corresponding to $c'$. If $q'$ was not an accepting state for $A$, then replace this transition by a reading move $c\overset{a}{\rightarrow}c'_{\mathrm{error}}$, where $c'_{\mathrm{error}}$ is like $c'$ except the corresponding state is $q_{\mathrm{error}}$ instead of $q'$.
\item Make every state in $A$ accepting in $A'$ (so that $q_{\mathrm{error}}$ becomes the unique non-accepting state).
\end{itemize}

By construction, $A'$ is obviously prefix-closed. It remains to show that it accepts the same language $L$. Let $w$ be an input word in $L$. Since $L$ is prefix-closed, each of the prefixes of $w$ is in $L$. Since $A$ is in normal form, each of the reading moves in the computation of $A$ on $w$ lead to an accepting state. Therefore, the computation of $A'$ on $w$ simulates the same transitions as those of $A$. In particular, it never transitions to state $q_{\mathrm{error}}$. Thus, $A'$ accepts $w$.

Now suppose that $w$ is an input word not in $L$. Decompose $w$ as $w'aw''$, where $w'$ is the largest prefix of $w$ such that $w'\in L$. This largest prefix exists since $L$ is prefix-closed and non-empty (in particular, the empty word must belong to $L$). By the previous argument, the computation of $A'$ in $w'$ simulates the same transitions as those of $A$. Let $c\overset{a}{\rightarrow}c'$ be the reading move that reads $a$ in the computation of $A$. Since $A$ is in normal form and $w'a\not\in L$, the state $q'$ corresponding to $c'$ is a non-accepting state of $A$. Therefore, the computation of $A'$ for $w$ transitions at this point to the state $q_{\mathrm{error}}$, and remains there for the rest of the computation. Thus, $A'$ rejects $w$. This concludes our proof.
\end{proof}

The final ingredient before proving \cref{thm:equivalencetypesautomata} is an explanation on how to construct a system of equations, given an obviously prefix-closed automaton accepting $\lcal(\TT)$, for some type $\TT$. Here we sketch only the construction for pushdown automata, as the ideas are essentially the same for the other models. For each accepting state $q$, we have a corresponding variable $\typec{X_q}$. For each mode $(q,\varepsilon)$ (resp.\ $(q,\sigma)$), we define the right-hand side corresponding to $\CALLT {X_q}{\varepsilon}$ (resp.\ $\CALLT{X_q}{\sigma S}$) according to the following case analysis (we sketch the case $(q,\varepsilon)$, as the analysis for $(q,\sigma)$ is identical):
\begin{itemize}
\item Suppose $(q,\varepsilon)$ is an $\varepsilon$-mode, with corresponding transition to, say, $(+\sigma,q')$. Then our system contains the equation $\iseqt{\CALL{X_q}{\varepsilon}}{\CALL{X_{q'}}{\sigma}}$.
\item Suppose $(q,\varepsilon)$ is a reading mode, and that reading $\Endl$ has a transition to an accepting state, say, $(=,q')$. Then, for any word $w$ such that there is a sequence of moves $(q_0,\omega_0)\overset{w}\rightarrow(q,\varepsilon)$, $w\cdot\Endl$ is a word in $\lcal(\TT)$. By the way $\lcal(\TT)$ is defined, it must be the only such word having $w$ as a proper prefix. Therefore, if such a word $w$ exists, reading any other symbol from configuration $(q,\varepsilon)$ must cause the automaton to transition to the non-accepting state. We include the equation $\iseqt{\CALL {X_q}\varepsilon}{\End}$ in our system.
\item Suppose $(q,\varepsilon)$ is a reading mode, and that reading $?\dl$ has a transition to an accepting state, say, $(=,q')$. Then, for any word $w$ such that $(q_0,\omega_0)\overset{w}\rightarrow(q,\varepsilon)$, $w\cdot?\dl$ is a word in $\lcal(\TT)$. This word and $w\cdot?\cl$ must be the only two words in $\lcal(\TT)$ that are immediate continuations of $w$. Therefore, if such a word $w$ exists, reading any other symbol must cause the automaton to transition to the non-accepting state. Suppose the state reached from $(q,\varepsilon)$ after reading $?\cl$ is, say, $(+\sigma,q'')$. We include the equation $\iseqt{\CALL{X_q}\varepsilon}{\IN{\CALL{X_{q'}}{\varepsilon}}{\CALL{X_{q''}}{\sigma}}}$ in our system.
\item A similar analysis takes care of the other cases in which reading a symbol leads to an accepting state. Notice that it is technically possible for multiple contradictory symbols to have reading moves to accepting states. For instance, it could be the case that $(q,\varepsilon)\overset{\Endl}{\rightarrow}(=,q')$ and $(q,\varepsilon)\overset{!\cl}{\rightarrow}(=,q'')$, with both $q'$ and $q''$ accepting. However, by the way $\lcal(\TT)$ is defined, this only occurs if there is no word $w$ such that $(q_0,\omega_0)\overset{w}{\rightarrow}(q,\varepsilon)$. Thus, we can put either option in the right-hand of $\CALLT X\varepsilon$, as this type constructor will also not be reachable.
\item The only case left is if $(q,\varepsilon)$ is a reading mode, but reading any symbol leads to the non-accepting state. This means that, for any $w$ with $(q_0,\omega_0)\overset{w}\rightarrow(q,\varepsilon)$, there is no other word in $\lcal(\TT)$ having $w$ as a prefix. By the way $\lcal(\TT)$ is defined, $w$ must end with the symbol $\Endl$. This again means that the right-hand side of equation $\CALLT X\varepsilon$ is irrelevant, as this variable will also not be reachable. We can define the corresponding equation to be $\iseqt{\CALL X\varepsilon}{\End}$ by default.
\end{itemize}

With the construction outlined above, we are able to prove the main result of this section. The following is a characterisation result that establishes a correspondence between classes of types and 
different models of computation based on automata theory. We remark that our result is stronger than previous similar results which only show a forward implication \cite{DBLP:conf/esop/DasDMP21}.
%
%
Recall that a language is said to be \emph{regular} if it is the set of words accepted by some finite-state automaton. We also say that a tree is \emph{regular} if it has a finite number of distinct subtrees.

\begin{theorem}[Types, traces and automata]
  \label{thm:equivalencetypesautomata}
  ~
\begin{enumerate}
\item $\istyper T$ iff $\lcal(\TT)$ is regular iff $\tree(\TT)$ is regular.
\item $\istypeo T$ iff $\lcal(\TT)$ is accepted by a 1-counter automaton.
\item $\istyped T$ iff $\lcal(\TT)$ is a deterministic context-free language.
\item $\istypet T$ iff $\lcal(\TT)$ is decidable.
\end{enumerate}
\end{theorem}

\begin{proof}
  The proof is identical for each of the four cases. In the forward direction,
  consider a system of equations that specify a type $\TT$ in one of the four
  classes, and use the construction in \cref{sec:systemstoautomata} to obtain
  the corresponding automata that accepts $\lcal(\TT)$. In the reverse
  direction, suppose that the language $\lcal(\TT)$ is in one of the four models
  of computation. Using \cref{thm:obviouslyprefixclosed}, we know that an
  obviously prefix-closed automaton exists that accepts $\lcal(\TT)$. Using the
  construction preceding this theorem, 
  we can obtain the
  corresponding system of equations that specifies $\TT$.
  The only case left is to prove that $\istyper T$ iff $\tree(\TT)$ is regular.
  However, for recursive types this has been observed before by Pierce \cite[Chapter
  21]{DBLP:books/daglib/0005958}.
\end{proof}


\section{The hierarchy of type classes}
\label{sec:hierarchy}

Using the above characterisation, we can show that the hierarchy of types is strict
(\eqref{eq:mainchaininclusions}; \cref{thm:inclusions}). 
The main idea in proving that our various formalisms for session types have different expressive power is to leverage known separation techniques from formal language theory, such as the pumping lemma. We illustrate the technique with the separation $\typesr\subsetneq\typeso$.

\begin{lemma}\label{lem:hierarchyonecounter}
If $\TT$ is a recursive type, then $\TT$ is a 1-counter type. On the other hand, $\tcounter$ is a 1-counter type but not a recursive type.
\end{lemma}

\begin{proof} Clearly, a system of recursive equations describing a type $\TT$ can be converted into a system of 1-counter equations whose transitions do not depend on the counter value, thus describing the same type $\TT$. On the other hand, suppose, for the sake of deriving a contradiction, that $\tcounter$ was a recursive type. By \cref{thm:equivalencetypesautomata}, we would conclude that its language $\lcal(\tcounter)$ is regular. Next, we apply the pumping lemma for regular languages~\cite[Section 3.1]{hopcroftullman:1979}: there must be a constant $n$ such that any word $z\in \lcal(\tcounter)$ with $|z|\geq n$ can be written as $z=uvw$ with $|uv|\leq n$, $|v|\geq 1$, and $uv^i w\in\lcal(\tcounter)$ for every $i\geq 0$. Take $$z=\left(\&\incl\right)^n\cdot\&\dumpl\cdot(!\cl)^n\cdot\Endl$$
It is clear that $z$ fits the condition in the pumping lemma, and that any $v$ in the desired decomposition must be a substring of $\left(\&\incl\right)^n$. However, $z$ is the only word in $\lcal(\tcounter)$ having $\&\dumpl\cdot(!\cl)^n\cdot\Endl$ as a suffix, which leads to a contradiction.
\end{proof}

To prove the separation between 1-counter types and pushdown types, we need to use a variant of the pumping lemma for 1-counter automata. The following result is due to Boasson.

\begin{lemma}[{Boasson~\cite[Theorem 3]{boasson:1973}}]\label{lem:pumpinglemmaboasson}
Let $L$ be a language accepted by a 1-counter automaton. Suppose that $f$ is a word in $L$ having a decomposition
$$f=g_1ug_2vg_3xg_4yg_5$$
with the following properties:
\begin{enumerate}
  \item $u,v,x,y$ are non-empty words;
  \item for all $n,m\geq 0$, the word $g_1u^ng_2v^mg_3x^mg_4y^ng_5$ is in $L$;
  \item for all $n\geq 0$, the set
  \begin{align*}\{m\,:\,&g_1u^ng_2vg_3xg_4y^mg_5\in L\text{ or }g_1u^mg_2vg_3xg_4y^ng_5\in L\\
  \text{ or }&g_1ug_2v^ng_3x^mg_4yg_5\in L\text{ or }g_1ug_2v^mg_3x^ng_4yg_5\in L\}\end{align*}
  is finite.
\end{enumerate}
Then, there exist $n,m,\lambda,\mu\geq 1$ such that for all $k\geq 0$, the word $$g_1u^{n+\lambda k-1}g_2vg_3x^{m+\mu k-1}g_4yg_5 \in L.$$
\end{lemma}

\begin{lemma}\label{lem:hierarchypushdown}
If\, $\TT$ is a 1-counter type, then $\TT$ is a pushdown type. On the other hand, $\tstack$ is a pushdown type but not a 1-counter type.
\end{lemma}

\begin{proof}
Clearly, a system of 1-counter equations describing a type $\TT$ can be converted into a system of 1-stack equations, whose stack has a unique symbol, and where the value of the counter corresponds to the size of the stack.

On the other hand, suppose, for the sake of deriving a contradiction, that $\tstack$ was a 1-counter type. By \cref{thm:equivalencetypesautomata}, we would conclude that its language $\lcal(\tstack)$ is accepted by a 1-counter automata. Now consider the following family of words parameterized by $n,n',m,m'\geq 0$:
$$f_{n,n',m,m'}=
\&\pushbl\cdot
(\&\pushal)^{n+1}\cdot
(\&\pushbl)^m\cdot
(\&\popl\cdot?\cl)^{m'}\cdot
(\&\popl\cdot!\cl)^{n'+1}\cdot
\&\popl\cdot?\dl\cdot
\Endl$$
Intuitively, $f_{n,n',m,m'}$ corresponds to the following sequence of interactions: pushing the symbol $\tau$; pushing $n+1$ copies of the symbol $\sigma$; pushing $m$ copies of the symbol $\tau$; popping the top symbol $\tau$ from the stack $m'$ times; popping the top symbol $\sigma$ from the stack $n'+1$ times; and popping the top symbol $\tau$ from the stack. By our construction of $\tstack$, it should be clear that $f_{n,n',m,m'}\in\lcal(\tstack)$ iff $n=n'$ and $m=m'$ (the reason for pushing each symbol $\sigma$, $\tau$ at least once is to exclude situations where the number of times a symbol is pushed would be higher than the number of times that symbol is popped).

In particular, the word $f_{1,1,1,1}$ satisfies the conditions in \cref{lem:pumpinglemmaboasson} with the decomposition
$$
g_1=\&\pushbl,\quad 
g_2=\&\pushal,\quad 
g_3=\varepsilon,\quad 
g_4=\&\popl\cdot!\cl,\quad
g_5=\&\popl\cdot?\dl\cdot\Endl,$$
$$
u=\&\pushal,\quad 
v=\&\pushbl,\quad 
x=\&\popl\cdot?\cl,\quad
y=\&\popl\cdot!\cl.
$$
Applying that lemma, we would conclude that there exist $n,m,\lambda,\mu\geq 1$ such that, for all $k\geq 0$, the word $f_{n+\lambda k-1,1,m+\mu k-1,1}$ is in $\lcal(\tstack)$. However, from our previous discussion, this means that $n+\lambda k-1=1$ and $m+\mu k-1=1$, which cannot be true for all $k\geq 0$. We have thus derived our contradicion.
\end{proof}

\begin{lemma}\label{lem:hierarchytwocounter}
If\, $\TT$ is a pushdown type, then $\TT$ is a 2-counter type. On the other hand, $\titer$ is a 2-counter type but not a pushdown type.
\end{lemma}

\begin{proof}
The inclusion follows from \cref{thm:equivalencetypesautomata} and the observation that all DCFLs are decidable. On the other hand, suppose, for the sake of deriving a contradiction, that $\titer$ was a pushdown type. By \cref{thm:equivalencetypesautomata}, we would conclude that its language $\lcal(\titer)$ is a DCFL, and in particular, a context-free language. Next, we apply the pumping lemma for context-free languages (Section 6.1 in Hopcroft and Ullman~\cite{hopcroftullman:1979}): there must be a constant $n$ such that any word $z\in \lcal(\titer)$ with $|z|\geq n$ can be written as $z=uvwxy$ with $|vwx|\leq n$, $|vx|\geq 1$, and $u v^i w x^i y\in\lcal(\titer)$ for every $i\geq 0$. Consider the following sequence of words in $\lcal(\titer)$, for $k\geq 0$:
$$z_k=?\cl\cdot!\cl\cdot?\cl\cdot(!\cl)^2\cdots?\cl\cdot(!\cl)^k\cdot?\dl\cdot\Endl$$
From inspection, we can conclude that $z_k$ are the only words in $\lcal(\titer)$ that end in $?\dl\cdot\Endl$, and that in $z_k$ the character $?\cl$ appears exactly $k$ times, the character $!\cl$ appears exactly $1+\ldots+k=\frac{k(k+1)}{2}$ times, and the characters $?\dl$ and $\Endl$ appear exactly once. Now apply the pumping lemma to get a decomposition of $z_n=uvwxy$, and consider the following two cases:

\begin{itemize}
  \item Suppose $vwx$ is contained in the prefix $?\cl\cdot!\cl\cdots?\cl\cdot(!\cl)^{n-1}$ of $z_n$. In this case, $y$ contains $?\cl\cdot(!\cl)^n\cdot?\dl\cdot\Endl$ as a suffix, and thus so do $uv^iwx^iy$ for any $i$. On the other hand, $z_n$ is the only word in $\lcal(\titer)$ having $?\cl\cdot(!\cl)^n\cdot?\dl\cdot\Endl$ as a suffix, which results in a contradiction.
  \item Suppose now that $vwx$ intersects the suffix $?\cl\cdot(!\cl)^n\cdot?\dl\cdot\Endl$ of $z_k$. Since every word in $\lcal(\titer)$ contains at most one $?\dl$ and one $\Endl$, $v$ and $x$ cannot contain those characters. Since $|vwx|\leq n$, the character $?\cl$ can appear at most once in $vwx$. If $?\cl$ never appears in $vx$, then $uwy$ is a word ending in $?\dl\cdot\Endl$ with $n$ occurences of the character $?\cl$ and strictly fewer than $1+\ldots+n$ occurences of the character $!\cl$. If $?\cl$ appears once in $vx$, then $!\cl$ appears at most $n-1$ times in $vx$. In this case, $uwy$ is a word ending in $?\dl\cdot\Endl$ with $n-1$ occurences of the character $?\cl$ and strictly more than $1+\ldots+(n-1)$ occurences of the character $!c$. In either case, $uwy$ cannot be one of the words $z_k$ and thus cannot be in $\lcal(\titer)$, from which we derive our contradiction.
\end{itemize}
\end{proof}

At the end of the hierarchy, we can prove the separation $\typest\subsetneq\typesi$ by a cardinality argument.

\begin{lemma}\label{lem:hierarchyinfinite}
  Let $\types'$ be the set of types that can be represented by some finite
  representation system. Then $\types'$ is a strict subset of $\typesi$. In particular, $\typest \subsetneq \typesi$.
\end{lemma}

\begin{proof}
Notice that the set of all possible infinite types is uncountable. In particular, for every infinite word $w = b_0b_1b_2\ldots$ over the alphabet $\{0,1\}$, we can define the type
$$\typec{T_w} = \typec{\sharp_0 \Endl . \sharp_1 \Endl . \sharp_2 \Endl . \ldots}$$
where $\sharp_n$ is either $?$ if $b_n=0$ or $!$ if $b_n=1$. As the set of such infinite words is uncountable, and $\lcal(\typec{T_w})\neq \lcal(\typec{T_{w'}})$ for $w\neq w'$, so is the set of all types. Moreover, any finite representation system can contain at most a countable set of types. Hence, we get the desired result.
\end{proof}

Notice that the cardinality argument presented above also shows that there is in fact an \emph{infinite, uncountable} set of types that cannot be represented by a given finite representation system.


\section{Results for context-free and nested session types}
\label{sec:hierarchy2}

Here we compare the context-free session types model \cite{DBLP:journals/corr/abs-2106-06658,DBLP:conf/tacas/AlmeidaMV20} and 
the nested session types model \cite{DBLP:conf/esop/DasDMP21} with the main hierarchy of our paper; \ie we prove the inclusions in \eqref{eq:chainpushdownnested} and \eqref{eq:chaincontextfree}, which complete the proof of \cref{thm:inclusions}.

Regarding context-free session types, it is quite clear that they extend recursive types: a recursive system of equations can be converted into the context-free syntax by replacing $\End$ with $\Skip$ and $\MSG TU$ with $\semit{\MSGn T}U$. The following result shows that $\typescf\subseteq\typesp^1$.

\begin{theorem}\label{thm:cstinpst1}
Let $\istypec T$, and let $\UT$ be such that $\embeds TU$. Then, there exists a representation of $\UT$ as a pushdown type, having the following properties:
\begin{itemize}
\item The representation uses a single type variable $\XT$.
\item The only occurrence of $\End$ 
is in the equation $\iseqt {\CALLT X \varepsilon} \End$.
\end{itemize}
\end{theorem}

\begin{proof}
Let $\TT$ be a context-free session type represented by some system of equations.
Without loss of generality, assume that this system is in the following normal form: in all equations $\iseqt{X_i}{T_i}$, the
right-hand side $T_i$ is given by only one type construct. We construct a
pushdown type $\TT$ using a single variable $\XT$. For each type variable $\typec{X_i}$ in
the definition of $\TT$, we have a corresponding stack symbol $\sigma_i$. Finally,
we translate the equations defining $\TT$ into equations defining $\UT$, as
follows.

\begin{itemize}
  \item For each equation $\iseqt{X_i}{\Skip}$, we have an equation $\iseqt{\CALL X {\sigma_i S}}{\CALL X S}$.
  \item For each equation $\iseqt{X_i}{\MSGn X_{i'}}$, we have an equation $\iseqt{\CALL X {\sigma_i S}}{\MSG {\CALL X {\sigma_{i'}}} {\CALL X S}}$.
  \item For each equation $\iseqt{X_i}{\choice\recordf {l_j}{X_{i_j}}{L}}$, we have an equation $\iseqt{\CALL X {\sigma_i S}}{\choice \recordf {l_j}{\CALL X {\sigma_{i_j} S}}{L}}$.
  \item For each equation $\iseqt{X_i}{X_{i_1};X_{i_2}}$, we have an equation $\iseqt{\CALL X {\sigma_i S}}{\CALL X {\sigma_{i_1} \sigma_{i_2} S}}$.
  \item For each equation $\iseqt{X_i}{X_{i'}}$, we have an equation $\iseqt{\CALL X {\sigma_i S}}{\CALL X {\sigma_{i'} S}}$.
  \item Additionally, we have the equation $\iseqt{\CALL X \varepsilon}{\End}$.
\end{itemize}

A simple coinductive proof then shows that, if the context-free session type is
given by $\TT = \typec{X_i}$, then $\UT = \CALLT X {\sigma_i}$ is a pushdown representation of the type corresponding to $\UT$.
\end{proof}

In fact, the previous result shows something stronger than $\typescf\subseteq\typesp$: it shows that $\typescf$ is a subset of the first level $\typesp^1$ of the hierarchy within pushdown session types (defined in \cref{sec:treeslanguages}). 

Next, we argue that context-free and 1-counter types are incomparable, which implies the strict inclusions $\typesr\subsetneq\typescf\subsetneq\typesp$. In our separation of $\typesp$ from $\typeso$ (\cref{lem:hierarchypushdown}), we show that $\tstack$ is not a 1-counter type. However, we can represent it as a context-free type $\typec{X_\varepsilon}$ with
\begin{align*}
  \typec{X_\varepsilon} \Eq&\; \typec{\&\{
        \pushal\colon \semit{X_\sigma}{X_\varepsilon},
        \pushbl\colon \semit{X_\tau}{X_\varepsilon}\}}\\
  \typec{X_\sigma} \Eq&\; \typec{\&\{
        \pushal\colon \semit{X_\sigma}{X_\sigma},
        \pushbl\colon \semit{X_\tau}{X_\sigma},
        \popl\colon \OUTn\Skip\}}\\
  \typec{X_\tau} \Eq&\; \typec{\&\{
        \pushal\colon \semit{X_\sigma}{X_\tau},
        \pushbl\colon \semit{X_\tau}{X_\tau},
        \popl\colon \INn\Skip\}}
\end{align*}

On the other-hand, from the work of Korenjak and Hopcroft~\cite{korenjakhopcroft:1966:simpledeterministic} we know that the language $L_3=\{\leftl^n\;\al\;\rightl^n\;\al\mid n\geq 0\}\cup\{\leftl^n\;\bl\;\rightl^n\;\bl\mid n\geq 0\}$ is deterministic context-free but cannot be accepted by a DPDA with a single state. This was used by Das \etal \cite{DBLP:conf/esop/DasDMP21} to argue that context-free session types cannot express language $L_3$. However, we can use 1-counter types to express this language, \ie define the type $\tkorhop$ as $\CALLT X \zero$ with equations
\begin{align*}
  \lhst X \zero \Eq&\; \typec{\&\{
      \leftl\colon \CALLT X {\succ \zero},
      \al\colon \CALLT Y \zero,
      \bl\colon \CALLT Z \zero\}}
  &
  \lhst X {\succ N} \Eq&\; \typec{\&\{
      \leftl\colon \CALLT X {\succ {\succ N}},
      \al\colon \CALLT Y {\succ N},
      \bl\colon \CALLT Z {\succ N}\}}
  \\
  \lhst Y \zero \Eq&\; \typec{\&\{
      \al\colon \End\}}
  &
  \lhst Y {\succ N} \Eq&\; \typec{\&\{
      \rightl\colon \CALLT Y N\}}
  \\
  \lhst Z \zero \Eq&\; \typec{\&\{
      \bl\colon \End\}}
  &
  \lhst Z {\succ N} \Eq&\; \typec{\&\{
      \rightl\colon \CALLT Z N\}}
\end{align*}
We conclude that $\tkorhop$ is a 1-counter type but not a context-free type.

Next we look at the equivalence between pushdown and nested session types. 
Recall that $\typesp^n$ corresponds to the types that have pushdown representations with at most $n$ type constructors, whereas $\typesn^n$ corresponds to the types that have nested representations with type constructors of arity at most $n$. The following result shows that $\typesp^n\subseteq\typesn^n$.

\begin{theorem}
Let $\TT$ be a pushdown session type with at most $n$ type constructors. Then, there is a nested session type representation for $\TT$, using type constructors of arity at most $n$.
\end{theorem}

\begin{proof}
Consider a pushdown representation of $\TT$ using the type constructors $\typec{X^{(1)}},\ldots, \typec{X^{(n)}}$ and a stack alphabet $\Delta$. We consider a nested session type representation using

\begin{itemize}
  \item for each variable $\typec{X^{(i)}}$, a type constructor $\typec{X^{(i)}_\varepsilon}$ of arity $0$;
  \item for each variable $\typec{X^{(i)}}$ and each stack symbol $\sigma$, a type constructor $\typec{X^{(i)}_\sigma}$ of arity $n$.
\end{itemize}

We translate each equation in the pushdown representation into an equation for the corresponding type constructor. We use the $n$ variables $\typec{\alpha_1}, \ldots, \typec{\alpha_n}$ in our equations. The translation only needs to convert type variables into type constructors:

\begin{itemize}
  \item $\CALLT {X^{(i)}} \varepsilon$ becomes $\typec{X^{(i)}_\varepsilon}$;
  \item $\CALLT {X^{(i)}} \sigma$ becomes $\CALLT {X^{(i)}_\sigma} {X^{(1)}_\varepsilon,\ldots,X^{(n)}_\varepsilon}$;
  \item $\CALLT {X^{(i)}} {\sigma S}$ becomes $\CALLT {X^{(i)}_\sigma} {\alpha_1,\ldots,\alpha_n}$;
  \item $\CALLT {X^{(i)}} S$ becomes $\typec{\alpha_i}$;
  \item $\CALLT {X^{(i)}} {\sigma \sigma' S}$ becomes $\CALLT {X^{(i)}_\sigma} {\CALL {X^{(1)}_{\sigma'}}{\alpha_1,\ldots,\alpha_n},\ldots,\CALL {X^{(n)}_{\sigma'}}{\alpha_1,\ldots,\alpha_n}}$.
\end{itemize}

Intuitively, each type constructor $\typec{X^{(i)}_\sigma}$ corresponds to the stage where variable $\typec{X^{(i)}}$ needs to be unfolded with $\sigma$ at the top of the stack. The arguments stored during the unfolding keep track of all possible ways one can empty the current stack. The argument at position $i$ is chosen if the pushdown unfolding would move to variable $\typec{X^{(i)}}$. A simple coinductive proof shows that both representations yield the same type.
\end{proof}

We illustrate the above proof with an example. Consider the pushdown session
type $\TT = \CALLT X \varepsilon$ with

\begin{align*}
\CALLT X \varepsilon &\Eq \typec{\&\{\keyword{pushOut}: \CALL X \sigma,\keyword{pushIn}: \CALL X \tau, \keyword{dump}: \CALL Y \varepsilon\}}\\
\CALLT X {\sigma S} &\Eq \typec{\&\{\keyword{pushOut}: \CALL X {\sigma \sigma S},\keyword{pushIn}: \CALL X {\tau \sigma S}, \keyword{pop}: \CALL X S, \keyword{dump}: \CALL Y {\sigma S}\}}\\
\CALLT X {\tau S} &\Eq \typec{\&\{\keyword{pushOut}: \CALL X {\sigma \tau S},\keyword{pushIn}: \CALL X {\tau \tau S}, \keyword{pop}: \CALL X S, \keyword{dump}: \CALL Y {\tau S}\}}\\
\CALLT Y \varepsilon &\Eq \End\\
\CALLT Y {\sigma S} &\Eq \typec{!\End. \CALL Y S}\\
\CALLT Y {\tau S} &\Eq \typec{?\End. \CALL Y S}
\end{align*}

The above session type is a variant of \cref{exa:pushdown} with two type
variables. It offers a choice of pushing either symbol $\sigma$ or $\tau$ into
the stack, popping the stack, or dumping the entire stack contents. When
dumping, the value at the top of the stack ($\sigma$ or $\tau$) determines
whether an $!\End$ or $?\End$ message is triggered. Applying the conversion in
the proof of the previous theorem, we get the following representation of the
session type $\TT = \typec{X_{\varepsilon}}$, using constructors of arity at most 2, which can be seen to yield the same type.

\begin{align*}
\typec{X_\varepsilon} &\Eq \typec{\&\{\keyword{pushOut}: \CALL {X_\sigma} {X_\varepsilon,Y_\varepsilon},\keyword{pushIn}: \CALL {X_\tau} {X_\varepsilon,Y_\varepsilon}, \keyword{dump}: Y_\varepsilon\}}\\
\CALLT {X_\sigma}{\alpha_1,\alpha_2} &\Eq \typec{\&\{\keyword{pushOut}: \CALL {X_\sigma} {\CALL {X_\sigma}{\alpha_1,\alpha_2}, \CALL {Y_\sigma}{\alpha_1,\alpha_2}},}\\
&\typec{\qquad\keyword{pushIn}: \CALL {X_\tau}{\CALL {X_\sigma}{\alpha_1,\alpha_2},\CALL {Y_\sigma}{\alpha_1,\alpha_2}},}\\
&\typec{\qquad\keyword{pop}: \alpha_1, \keyword{dump}: \alpha_2\}}\\
\CALLT {X_\tau}{\alpha_1,\alpha_2} &\Eq \typec{\&\{\keyword{pushOut}: \CALL {X_\sigma}{\CALL {X_\tau}{\alpha_1,\alpha_2},\CALL {Y_\tau}{\alpha_1,\alpha_2}},}\\
&\typec{\qquad\keyword{pushIn}: \CALL {X_\tau}{\CALL {X_\tau}{\alpha_1,\alpha_2},\CALL {Y_\tau}{\alpha_1,\alpha_2}},}\\
&\typec{\qquad\keyword{pop}: \alpha_1, \keyword{dump}: \alpha_2\}}\\
\typec{Y_\varepsilon} &\Eq \End\\
\CALLT {Y_\sigma}{\alpha_1,\alpha_2} &\Eq \typec{!\End.\alpha_2}\\
\CALLT {Y_\tau}{\alpha_1,\alpha_2} &\Eq \typec{?\End.\alpha_2}
\end{align*}

The remainder of this section is devoted to the reverse implication, \ie that nested session types can be simulated with pushdown session types. As a warmup, we begin by looking at nested session types using unary constructors. The following result shows that $\typesn^1$ is contained in $\typesp^1$.

\begin{lemma}\label{lem:nst1pst1}
Let $\TT$ be a nested session type using type constructors of arity at most one. Then, there is a pushdown session type representation for $\TT$ using only one type constructor.
\end{lemma}

\begin{proof}
Consider a nested session type representation of $\TT$ using type constructors $\typec{X^{(1)}},\ldots, \typec{X^{(n)}}$. By renaming the variables, we can assume that the single argument of every type constructor is denoted by $\typec{\alpha}$.

We convert the nested session type representation into a pushdown session type representation as follows. We have a single type variable $\XT$; for each type constructor $\typec{X^{(i)}}$, we have a corresponding stack symbol $\sigma_i$. The equation in the nested session type representation corresponding to type constructor $\typec{X^{(i)}}$ is converted into the equation for the case that $\sigma_i$ is at the top of the stack. Namely, the translation converts nested type constructors into type variables:

\begin{itemize}
\item if $\typec{X^{(i_k)}}$ has arity $0$, then $\CALLT{X^{(i_1)}}{\ldots\CALL{}{X^{(i_k)}}\ldots}$ becomes $\CALLT X {\sigma_{i_1}\ldots\sigma_{i_k} S}$;
\item if $\typec{X^{(i_k)}}$ has arity $1$, then $\CALLT{X^{(i_1)}}{\ldots\CALL{}{\CALL{X^{(i_k)}}{\alpha}}\ldots}$ becomes $\CALLT X {\sigma_{i_1}\ldots\sigma_{i_k} S}$;
\end{itemize}

Finally, if $\TT = \CALLT{X^{(i_1)}}{\ldots\CALL{}{X^{(i_k)}}\ldots}$ is the
initial type on the nested session type representation, then $\TT = \CALLT X {\sigma_{i_1}\ldots\sigma_{i_k}}$ is the corresponding initial type in the pushdown session type representation. A simple coinductive proof shows that both representations yield the same type.
\end{proof}

With significant more effort, we can extend the above simulation to $n$-ary constructors.

\begin{theorem}\label{thm:nstnpstn}
Let $\TT$ be a nested session type using type constructors of arity at most $n$. Then, there is a pushdown session type representation for $\TT$ using only $n$ type variables.
\end{theorem}

\begin{proof}
In \cref{sub:contextfreenested}.
\end{proof}

Given the close relationship between pushdown and nested session types, we make at this point some important remarks comparing both models.

\begin{itemize}
  \item The proofs in this section also provide algorithms for converting between (representations of) pushdown session types and nested session types. It can be seen that both algorithms run in polynomial time, and in particular they incur only a polynomial overhead. In other words, if $\TT$ has a pushdown representation of size $n$, then $\TT$ has a nested representation of size at most $\mathrm{poly}(n)$ and vice-versa.
  \item We arrived at our hierachy of session types by thinking about equational definitions and about possible ways by which  the type constructors can be parameterized. This makes pushdown session types a `natural' level of the hierarchy, with 1-counter and 2-counter types as other natural choices. Nested session types, however, arised by thinking of type constructors that are applied to other type constructors. It is not obvious what would be the counterparts of 1-counter or 2-counter session types in the nested session type framework.
  \item As we unfold a pushdown session type, its encoding size can only grow polynomially, whereas the unfolding of a nested session type can grow exponentially on the number of steps. More formally, suppose we sequentially unfold a pushdown expression $\TT$: at each stage, we choose a type constructor $\XT$ appearing in $\TT$ and replace it according to the appropriate equation. Then, the expression achieved at stage $n$ of this unfolding has size bounded by $\mathrm{poly}(n)$. On the other hand, consider a nested representation of the type $\tloop$ as $\CALLT X {Y,Y}$ with equations
  $$\iseqt Y \End\qquad \iseqt {\CALL X {\alpha,\beta}} {\IN\End {\CALL X {\CALL X {\alpha,\beta}, \CALL X {\alpha,\beta}}}}$$
  one can see that at each unfolding step the encoding size (\eg the number of characters) of the nested session type doubles and so after $n$ steps we reach an expression of size $\varTheta(2^n)$. Hence pushdown types permit a more efficient direct representation of their unfoldings.
\end{itemize}


\section{Decidability/Undecidability of Key Problems}
\label{sec:decidability}

We are now in a position 
to address the decidability of the key problems of type
formation, type equivalence and type duality for the various classes of type
languages studied in this paper.

  

Before looking at type formation, we need to study the problem of deciding type contractiveness, described in \cref{fig:recursive-abbr,fig:onecounter-abbr,fig:cfst,fig:pushdown-abbr,fig:nested}. Let us say that a system of recursive equations over $\xcal$ is contractive if $\iscontrt X$ for every $\XT\in\xcal$. Similarly, a system of 1-counter equations over $\xcal$ is said to be contractive if $\iscontrt {\CALL X n}$ for every $X\in\xcal$ and every $n\in\nbb$. We can extend this notion in the obvious way to pushdown systems and 2-counter systems.

From the construction described in \cref{sec:systemstoautomata}, we can inherit contractiveness conditions by looking at loop-freeness of the associated automata. The following definition captures the notion of loop-freeness (more precisely, $\varepsilon$-loop-freeness) for all automata models (see also Ginsburg and Greibach~\cite{ginsburggreibach:1965:dcfl} and Valiant~\cite{valiant:1973:phdthesis}). By a configuration we mean: in finite-state automata, a state $q\in Q$; in 1-counter automata, a pair $(q,n)\in Q\times \nbb$; in pushdown automata, a pair $(q,\omega)\in Q\times\Delta^\ast$; and in 2-counter automata, a triple $(q,n,m)\in Q\times\nbb\times\nbb$.

\begin{definition}\label{def:loopfreeness}
An automaton is said to be loop-free if, for every configuration $c$, the sequence of $\varepsilon$-moves started from $c$ eventually reaches a reading configuration.
\end{definition}

\begin{lemma}\label{lem:systemtoautomaton}
Let $\mathrm{Sys}$ be a system of recursive equations (resp.\ 1-counter equations, pushdown equations, 2-counter equations), and $A$ the corresponding automaton as constructed in \cref{sec:systemstoautomata}. Then $\mathrm{Sys}$ is contractive iff $A$ is loop-free.
\end{lemma}

\begin{proof}
We sketch the proof for pushdown systems, since the other cases follow the same analysis. Observe that any configuration of the form $(q_\mathrm{end},n)$ or $(q_\mathrm{error},n)$ is already a reading configuration, so it cannot be the start of an infinite sequence of $\varepsilon$-moves. For the remaining configurations $(q_X,\omega)$, it is clear by our construction that the sequence of $\varepsilon$-moves obtained by following the transition function is equivalent to a derivation attempt for $\iscontrt {\CALL X \omega}$ following rules \ctz and \cts, and that this sequence eventually reaches a reading configuration iff the derivation is successful. Thus we have an equivalence between systems for which all (variable, stack) pairs are contractive and automata for which all configurations eventually reach a reading configuration.
\end{proof}

\begin{theorem}\label{thm:contractivity}
The following problems are decidable in polynomial time:
\begin{itemize}
\item Given a system $\mathrm{Sys}$ of recursive equations, is $\mathrm{Sys}$ contractive?
\item Given a system $\mathrm{Sys}$ of 1-counter equations, is $\mathrm{Sys}$ contractive?
\item Given a system $\mathrm{Sys}$ of pushdown equations, is $\mathrm{Sys}$ contractive?
\end{itemize}
\end{theorem}

\begin{proof}
In \cref{sub:contractiveness}.
\end{proof}

\begin{theorem}
  \label{thm:typeformation}
  Problems $\istyper T$, $\istypeo T$ and $\istyped T$ are all decidable in
  polynomial time.
\end{theorem}

\begin{proof}
We only sketch the proof for the case of pushdown types, of which the other two can be seen as subcases. The algorithm described in the proof of \cref{thm:contractivity} not only determines whether a system of equations is contractive, but it can also be used to produce the set of ``bad'' type identifiers
\begin{equation*}
  \T_{\badl}=\{\CALLT X \sigma\mid\typec\sigma\in\Delta\cup\{\varepsilon\}\;,\;\neg \iscontrd {\CALLT X \sigma }\}.
\end{equation*}
Intuitively, a type $\TT=\CALLT X \omega$ is well-formed ($\istyped \TT$) iff
the expansion of $\TT$ (which might be infinite) never visits a type identifier
in $\T_{\badl}$. In particular, if the system is contractive, then
$\istyped{\CALL X\omega}$ for any variable $\XT$ and stack contents $\omega$.
Otherwise, we can apply the construction in \cref{sec:systemstoautomata} to
convert the system of pushdown equations into a deterministic pushdown automata;
we get that $\neg\istyped T$ iff there is a derivation
$(q,\omega)\overset{w}{\rightarrow}(q_X,\sigma)$ for some
$\CALLT X\sigma\in \T_{\badl}$. Now deciding whether such derivations exist can
be reduced to solving the reachability problem on deterministic pushdown
automata, which can be done in polynomial time (in fact, the problem even
remains polynomial-time solvable for nondeterministic pushdown automata). Here
is a short argument: one can change the automata in such a way that the only
accepting states are the states $(q_X,\sigma)$ corresponding to ``bad''
configurations, reducing the problem to deciding if the pushdown automaton
accepts a non-empty language. Then, we can transform the automaton into a
context-free grammar \cite[Theorem 5.4]{hopcroftullman:1979}. Finally, we can
use a polynomial-time algorithm~\cite[Lemma 4.1]{hopcroftullman:1979} to
decide if the language generated by a context-free grammar is non-empty.
\end{proof}

By making use of the known procedures for deciding equivalence of deterministic
automata, and since the construction in \cref{sec:systemstoautomata} can be
implemented by a computable procedure, we can immediately derive decidability for
the corresponding problems for types.

\begin{theorem}
  \label{thm:typeequivalence}
  Problems $\isequivr TU$, $\isequivo TU$ and $\isequivd TU$ are all decidable.
\end{theorem}

\begin{proof}
  An algorithm for deciding type equivalence works as follows. First convert
  each type $\TT,\UT$ into an equivalent automaton, following the steps in
  \cref{sec:systemstoautomata}. If $\TT$ (resp.\ $\UT$) is given by the initial
  type identifier $\CALLT X\omega$, then its corresponding automaton has
  $(q_X,\omega)$ as the initial configuration. By our construction, we get that
  $\lcal(\TT)$ (resp.\ $\lcal(\UT)$) is the language accepted by the
  corresponding automaton, and we can infer that $\isequiv TU$ iff (by
  \cref{prop:typeequivalence}) $\lcal(\TT)=\lcal(\UT)$ iff the corresponding
  automata are equivalent. We know that the equivalence of automata is decidable
  for finite-state automata 
  \cite{hopcroftkarp:1971,rabinscott:1959},
  1-counter automata
  \cite{DBLP:conf/stoc/BohmGJ13,valiant:1973:phdthesis,valiantpaterson:1975:onecounterautomata} and
  (deterministic) pushdown automata
  \cite{senizergues:1997:equivalencedpda,senizergues:2001:decidabilitycompleteformalsystems}.
  Applying the corresponding algorithm gives us the desired answer.
\end{proof}


Building on type equivalence, we can establish similar results to decide whether
two types are the dual of each other. We start by building a dual to any type $\TT$.

\begin{lemma}\label{lem:typedualityexistence}
  For each class of types, 
  if $\istype T$, then there exists
  $\typec{\dual{T}}$ such that $\isdual T{\typec{\dual{T}}}$.
\end{lemma}

\begin{proof}
  We sketch the proof for pushdown types, as the other classes use essentially
  the same idea. Consider a pushdown type $\TT$, say
  $\TT = \CALLT{X_0}{\omega_0}$, with respect to a system $\mathrm{Sys}$ of
  pushdown equations, defined on a set of variables $\X$. We construct a dual
  type $\typec{\dual{T}} = \CALLT{\overline{X}_0}{\omega_0}$ by extending the
  system $\mathrm{Sys}$ to a system $\mathrm{Sys}'$, defined on the set of
  variables $\{\XT\mid \XT\in\X\} \cup \{\typec{\dual{X}}\mid \XT\in\X\}$. The
  equations for the variables in $\X$ are the same as in $\mathrm{Sys}$. The
  equations for the duals of the variables in $\X$ are given according to the
  usual rules:
\begin{itemize}
  \item the dual of $\End$ is $\End$;
  \item the dual of $\MSG TU$ is $\typec{\dual\sharp T.\dual{U}}$ where
      $\typec{\dual{U}}$ is the dual of $\UT$;
  \item the dual of $\choicet$ is $\typec{\dual\star{\record \ell {\dual{T}}
          L}}$ where $\typec{\dual{T_\ell}}$ is the dual of $\typec{T_\ell}$;
  \item additionally, the dual of $\CALLT X\omega$ is $\CALLT{\dual{X}}{\omega}$.
\end{itemize}

A straightforward proof by coinduction then shows that
$\isduald {\CALL X\omega}{\CALL {\dual{X}}\omega}$ for every $\XT\in\X$ and
$\omega\in\Delta^\ast$, and thus $\isduald T{\typec{\dual{T}}}$.
\end{proof}

Notice that the proof above is constructive, \ie given a system $\mathrm{Sys}$ specifying $\TT$, we can effectively produce a system $\mathrm{Sys}'$ specifying $\typec{\dual{T}}$.

\begin{theorem}
  \label{thm:typeduality}
  Problems $\isdualr TU$, $\isdualo TU$ and $\isduald TU$ are all decidable.
\end{theorem}

\begin{proof}
Given the types $\TT$, $\UT$, construct $\typec{\dual{T}}$ according to the proof of \cref{lem:typedualityexistence}. Then decide whether $\typec{\dual{T}}$ and $\UT$ are equivalent by applying the procedure in the proof of \cref{thm:typeequivalence}.
\end{proof}


In \cref{sec:systemstoautomata} we show how to convert a system of equations into an automaton, which enables us to prove that certain problems on types are decidable by observing that their counterparts for automata are decidable. Similarly, our reverse construction from automata to systems of equations in \cref{sub:automatatosystems} allows to prove that certain problems on types are undecidable since their counterparts for automata are undecidable.

\begin{theorem}[Undecidability results]
\label{thm:undecidability}
Problems $\istypet T$, $\isequivt TU$ and $\isdualt TU$ are all undecidable.
\end{theorem}

\begin{proof}
  For $\istypet T$ we start from the following undecidable
  problem (essentially, the halting problem): given a description of a one-tape
  Turing machine $M$, and starting from an empty tape, determine whether the
  machine reaches a given state $q$. We reduce from this problem into the
  problem of deciding whether a given type identifier is contractive. Apply the
  constructions in Hopcroft and Ullman~\cite[Chapter 7]{hopcroftullman:1979} (already mentioned in
  our \cref{thm:normalformautomaton}) to convert $M$ into a 2-counter automaton
  $A$. The construction yields an initial configuration $c_0$ and a final state
  $q'$ in $A$ such that $q'$ is reachable from $c_0$ by $\varepsilon$-moves iff
  the machine $M$ reaches $q$. Now apply the construction at the end of
  \cref{sub:automatatosystems} to construct the associated system of 2-counter
  equations, but making the replacement $\iseqt{\CALL {X_{q'}}{N,M}}{\End}$ for all
  equations corresponding to state $q'$. Let $\TT=\CALLT{X_0}{n_0,m_0}$ be the type
  identifier corresponding to configuration $c_0$. We get that $\istypet T$
  iff $\iscontrt T$ iff $q'$ is reachable from $c$ by $\varepsilon$-moves iff
  the machine $M$ reaches $q$, concluding the reduction.

  For $\isequivt TU$ we can simply observe that type equivalence builds on contractivity (\cf the 2-counter
  type formation rule correspondent to rule \wid, \cref{fig:recursive-abbr}). 
  Thus, to decide type equivalence one must decide contractiveness, which we have just shown to be undecidable. Alternatively, we can show that $\isequivt TU$ is undecidable without resorting to the undecidability of $\iscontrt T$ (which suggests that the problem is `harder' than contractiveness). To do that, we start from the following undecidable
  problem: given two decidable languages $E$, $F$, determine whether $E=F$ (this
  is undecidable even for context-free languages, as shown by
  Hopcroft and Ullman~\cite{hopcroftullman:1979}). Without loss fix a computable encoding between
  the words in the language and natural numbers, so that we can assume that
  $E,F\subseteq \nbb$. Next, consider the types $\typec{T_E}, \typec{T_F}$ given
  by $\typec{T_E} = \typec{\sharp_0\End.\sharp_1\End.\sharp_2\End\ldots}$, where $\sharp_n$ is
  either $?$ if $n\in E$ or $!$ if $n\not\in E$, and similarly for
  $\typec{T_F}$. Since $E,F$ are decidable, so are
  $\lcal(\typec{T_E}), \lcal(\typec{T_F})$; thus, by
  \cref{thm:equivalencetypesautomata}, $\istypet {T_E}$ and $\istypet {T_F}$.
  Observing that $E=F$ iff $\isequivt {T_E}{T_F}$ concludes the reduction.

  For $\isdualt TU$ we reduce from type equivalence, noting
  that $\isdualt TU$ iff $\isequivt {\overline{T}}U$ where $\typec{\dual{T}}$ is
  the type constructed from $\TT$ according to the proof of
  \cref{lem:typedualityexistence}.
\end{proof}


\section{Related work}
\label{sec:related}

The first papers on session types by Honda \cite{DBLP:conf/concur/Honda93} and Takeuchi \etal
\cite{DBLP:conf/parle/TakeuchiHK94} feature finite types only. Recursive types
were introduced later \cite{DBLP:conf/esop/HondaVK98} using $\mu$-notation. Gay and Hole
\cite{DBLP:journals/acta/GayH05} introduce algorithms for deciding duality and
subtyping of finite-state session types, based on bisimulation. Much of the
literature on session types, surveyed by Hüttel \etal
\cite{DBLP:journals/csur/HuttelLVCCDMPRT16}, uses the same approach. The
natural decision algorithms for duality and subtyping presented by Gay and
Hole were shown to be exponential in the
size of the types by Lange and Yoshida~\cite{DBLP:conf/tacas/LangeY16}, due to reliance on syntactic unfolding. Our polytime
complexity for recursive type equivalence follows from the equivalence algorithm for finite-state automata by Hopcroft and Karp~\cite{hopcroftkarp:1971}, and thus has quadratic complexity in the description size, which is an improvement on that of Gay and Hole. Lange and Yoshida use an
automata-based algorithm to also achieve quadratic complexity for checking
subtyping.

We use a coinductive formulation of infinite session types. This approach has
some connections with the work of Keizer \etal~\cite{DBLP:conf/esop/KeizerB021} who
present session types as states of coalgebras. Their types are restricted to
finite-state recursive types, but they do address subtyping and non-linear
types, two notions that we do not take into consideration. Our coinductive
presentation avoids explicitly building coalgebras, and follows 
Gay \etal~\cite{DBLP:journals/corr/abs-2004-01322}, solving problems with duality in the
presence of recursive
types~\cite{DBLP:journals/corr/BernardiH13,DBLP:journals/corr/abs-2004-01322,DBLP:conf/icfp/LindleyM16}.

This paper does not address the problem of deciding subtyping, but the panorama
is not promising. Subtyping is known to be decidable for recursive types
$\typesr$ \cite{DBLP:journals/acta/GayH05} 
and undecidable for context-free types $\typescf$
\cite{DBLP:journals/toplas/Padovani19} or nested types with arity at
most one $\typesn^1$ \cite{DBLP:journals/corr/abs-2103-15193}, hence for
pushdown types with one type constructor $\typesp^1$ (\cref{thm:inclusions}).
The undecidability proof of the subtyping problem for context-free session
types reduces from the inclusion problem for simple deterministic languages, which was shown to be undecidable by Friedman~\cite{DBLP:journals/tcs/Friedman76}.
That for nested session types reduces from the inclusion problem for Basic Process
Algebra \cite{DBLP:conf/rex/BergstraK88}, which was shown to be undecidable by
Groote and Hüttel~\cite{DBLP:journals/iandc/GrooteH94}.
Given that 1-counter types $\typeso$ and pushdown types with one type
constructor $\typesp^ 1$ are incomparable (\cref{thm:inclusions}), the problem
of subtyping for 1-counter types remains open.

Dependent session types have been studied in several forms, for binary session types \cite{DBLP:journals/pacmpl/ThiemannV20,DBLP:conf/ppdp/ToninhoCP11}, for multi-party session types \cite{DBLP:journals/corr/abs-1208-6483,DBLP:journals/corr/abs-1904-01288,DBLP:conf/fossacs/YoshidaDBH10} and for polymorphic, nested session types \cite{DBLP:conf/esop/DasDMP21}.
Although our parameterised type definitions have some similarities with definitions in some dependently typed systems, we do not support the connection between values in messages and parameters in types, and we have not yet studied how the types that can be expressed in dependent systems fit into our hierarchy.

Connections between multiparty session types and communicating finite-state
automata have been explored by Deniélou and Yoshida~\cite{DBLP:conf/esop/DenielouY12} but the
investigation has not been extended to other classes of automata.

Solomon~\cite{DBLP:conf/popl/Solomon78} studies the connection between inductive type
equality for nested types and language equality for DPDAs and shows that the
equivalence problem for nested types is as hard as the equivalence problem for
DPDAs, an open problem at the time. We follow a similar approach but take type
equivalence coinductively, as a bisimulation, rather than as a problem of
language equivalence.

Many of the main results in this paper borrow from the theory of automata,
developed in the mid-20th century. Here our standard reference is the book by
Hopcroft and Ullman~\cite{hopcroftullman:1979}, where the notions of finite-state automata,
pushdown automata, and counter automata can be found. 1-counter automata were
studied in detail in Valiant's PhD thesis
\cite{valiant:1973:phdthesis}. To prove the equivalence between types and
automata, we need to convert automata into equivalent ones satisfying certain
properties; similar techniques have appeared in Kao \etal~\cite{kaoetal:2009} and Valiant and Paterson~\cite{valiantpaterson:1975:onecounterautomata}. Our
proofs of decidability of type equivalence make use of the corresponding results
for automata
\cite{DBLP:conf/stoc/BohmGJ13,hopcroftkarp:1971,rabinscott:1959,senizergues:1997:equivalencedpda,senizergues:2001:decidabilitycompleteformalsystems,valiantpaterson:1975:onecounterautomata};
we specifically mention Sénizergues' impressive result on the decidability of
equivalence for deterministic pushdown automata
\cite{senizergues:2001:decidabilitycompleteformalsystems}, a work which granted
him the Gödel Prize in 2002. Finally, the strict hierarchy results use textbook
pumping lemmas for regular languages (due to Rabin and Scott~\cite{rabinscott:1959}) and
context-free languages (due to
Bar-Hillel \etal \cite{barhilleletal:1961:formalpropertiesgrammars} and
Kreowski~\cite{kreowski:1978:pumpinglemma}), as well as a somewhat less known result for
1-counter automata (due to Boasson~\cite{boasson:1973}).


\section{Conclusion}
\label{sec:conclusion}

We introduce different classes of session types, some new, others from the
literature, under a uniform framework and place them in an hierarchy. We further
study different type-related problems---formation, equivalence and duality---and
show that these relations are all decidable up to and including pushdown types.

Much remains to be done. From the point of view of programming languages, one
should investigate whether decidability results translate into algorithms that
may be incorporated in compilers.
Even if subtyping is known to be undecidable for most systems ``above'' that of
recursive types, the problem remains open for 1-counter types, an interesting
avenue for further investigation.
Our study of classes of infinite types may have applications beyond session
types. One promising direction is that of non regular datatypes for functional
programming (or polymorphic recursion schemes
\cite{DBLP:conf/programm/Mycroft84}), such as nested datatypes
\cite{DBLP:conf/mpc/BirdM98}.

We have not
addressed the decidability of the type checking problem. 
Type checking is known to be decidable for finite types, recursive, context-free
and nested session types.
%
%
Given that type checking for nested session types is incorporated in the RAST language
\cite{DBLP:conf/esop/DasDMP21}, a natural first step would be to investigate how to translate 1-counter and pushdown processes into that language.




\bibliographystyle{splncs04}
\bibliography{references}


\appendix
\section*{Appendix}

\section{Proof of \texorpdfstring{$\text{\cref{lem:embedding}}$}{Lemma 2.1} (Embedding from context-free to infinite session types)}
\label{sub:proof-embedding}

To prove the embedding theorem, we use the classical coinduction principle for set-based
coalgebras 
\cite{DBLP:journals/mscs/KozenS17}.

\begingroup
\def\thetheorem{\ref{lem:embedding}}
\begin{theorem}[Embedding]~
\begin{enumerate}
  \item\label{lem:embed-1} If $\embeds TU$, then $\istypec T$ and $\istypei U$.
  \item\label{lem:embed-2} If $\istypec T$, then there exists $\UT$ with $\embeds TU$.
  \item\label{lem:embed-3} Suppose $\embeds TU$ and $\embeds VW$. Then $\isequiv TV$ iff $\isequiv UW$.
\end{enumerate}
\end{theorem}
\addtocounter{theorem}{-1}
\endgroup

\begin{proof}
For \cref{lem:embed-1}, we proceed coinductively on the structure of the proof of $\embeds TU$. We illustrate some relevant cases:
\begin{itemize}
  \item Suppose a proof for $\embeds TU$ ends with rule \embskip. Then $\TT$ is $\Skip$ and $\UT$ is $\End$. We know that $\istypec\Skip$ by \wskip and $\istypei\End$ by \wend.
  \item Suppose a proof for $\embeds TU$ ends with rule \embid. Then
    $\TT$ is $\XT$ with $\iseqt X{T'}$ and $\iscontrt {T'}$ and $\embeds{T'}U$. By coinduction, $\istypec{T'}$ and $\istypei U$. By rule \wid, $\istypec X$ as well.
  \item Suppose a proof for $\embeds TU$ ends with rule \embsemichoice. Then $\TT$ is $\semit{\choice\recordf \ell {T_\ell} L}{T'}$ and $\UT$ is $\choice\recordt \ell U L$. We also have that $\embeds{\semit{T_\ell}{T'}}{U_\ell}$ for all $\ell\in L$. By coinduction, $\istypec{\semit{T_\ell}{T'}}$ and $\istypei{U_\ell}$ for all $\ell\in L$. By rule \wchoice, $\istypei{\choice\recordt \ell U L}$ as well. 
  Moreover, the proof for $\istypec{\semit{T_\ell}{T'}}$ must end with rule $\wsemi$, so that $\istypec{T_\ell}$ for all $\ell\in L$ and $\istypec{T'}$. Therefore, by rules \wchoice and \wsemi, $\istypec{\semit{\choice\recordf \ell {T_\ell} L}{T'}}$ as well.
  \item Suppose a proof for $\embeds TU$ ends with rule \embsemisemi. Then $\TT$ is $\semit{(\semit{T_1}{T_2})}{T_3}$ and we have $\embeds {\semit{T_1}{(\semit{T_2}{T_3})}}U$. 
  By coinduction, $\istypec{\semit{T_1}{(\semit{T_2}{T_3})}}$ and $\istypei U$. 
  The proof for $\istypec{\semit{T_1}{(\semit{T_2}{T_3})}}$ must end with rule $\wsemi$, so that $\istypec{T_1}$ and $\istypec{\semit{T_2}{T_3}}$. 
  Similarly, we have that $\istypec{T_2}$ and $\istypec{T_3}$. By rule \wsemi, $\istypec{\semit{(\semit{T_1}{T_2})}{T_3}}$ as well.
\end{itemize}

For \cref{lem:embed-2}, we proceed coinductively on the structure of the proof of $\istypec T$. We illustrate some relevant cases.

\begin{itemize}
  \item Suppose $\TT$ is $\MSGn{T'}$. Then $\istypec{T'}$ and by coinduction, there exists $\typec{U'}$ with $\embeds{T'}{U'}$. Then, by rule \embmsg, we have $\embeds {\MSGn{T'}}{\MSG{U'}\End}$.
  \item Suppose $\TT$ is $\XT$, with $\iseqt X {T'}$ and $\iscontrt{T'}$ and $\istypec{T'}$. By coinduction, there exists $\UT$ with $\embeds {T'}U$. By rule \embid, we then have $\embeds XU$.
  \item Suppose $\TT$ is $\semit{T_1}{T_2}$. Here we have $\istypec{T_1}$, and again we proceed by coinduction on that proof (\ie $\typec{T_1}$ is either $\Skip$, message passing, choice, a variable, or sequential composition). Suppose $\typec{T_1}$ is $\choice\recordt \ell T L$. For each $\ell\in L$, we have $\istypec{T_\ell}$ and also (due to rule \embskip) $\istypec{\semit{T_\ell}{T_2}}$. By coinduction, there exists $\typec{U_\ell}$ such that $\embeds{\semit{T_\ell}{T_2}}{U_\ell}$ for each $\ell\in L$. We then have by rule \embsemichoice that $\embeds{\semit{\choice\recordt \ell T L}{T_2}}{\choice\recordt \ell U L}$.
  \item Suppose $\TT$ is $\semit{(\semit{T_1}{T_2})}{T_3}$. We have that $\istypec{T_1}$, $\istypec{T_2}$ and $\istypec{T_3}$, so that also $\istypec{\semit{T_1}{(\semit{T_2}{T_3})}}$. Taking $\UT$ such that $\embeds{\semit{T_1}{(\semit{T_2}{T_3})}}{U}$, we have by rule \embsemisemi that $\embeds{\semit{(\semit{T_1}{T_2})}{T_3}}U$.
\end{itemize}

For \cref{lem:embed-3}, we proceed coinductively on the structure of the proofs for $\embeds TU$ and $\embeds VW$. Note that there are nine possible rules for each proof, so there are eighty-one cases in total. We illustrate some relevant cases.

\begin{itemize}
  \item Suppose $\TT$ is $\semit{\Skip}{T'}$ and $\VT$ is $\Skip$. The proof for $\embeds TU$ must have ended with rule \embsemiskip, so that we have $\embeds {T'}U$. The proof for $\embeds VW$ must have ended with rule \embskip, so that $\WT$ is $\End$. In the forward direction, suppose $\isequiv TV$. Given the structure of $\TT$ and $\VT$, its proof must have ended with rule \eqneutronel, and so $\isequiv{T'}V$. Then, by coinduction, we get $\isequiv UW$. In the converse direction, suppose $\isequiv UW$. By coinduction, we have $\isequiv {T'}{V}$. Then, by rule \eqneutronel we get $\isequiv TV$.

  \item Suppose $\TT$ is $\semit{\choice\recordt \ell T L}{T'}$ and $\VT$ is $\semit{(\semit{V_1}{V_2})}{V_3}$. The proof for $\embeds TU$ must have ended with rule \embsemichoice, so that $\UT$ is $\choice\recordt \ell U L$ with $\embeds{\semit{T_\ell}{T'}}{U_\ell}$ for each $\ell\in L$. By rule \embchoice, this means that \linebreak$\embeds {\choice\recordf \ell {\semit{T_\ell}{T'}} L} {\choice\recordt \ell U L}$. The proof for $\embeds VW$ must have ended with rule \embsemisemi, so we must have $\embeds{\semit{V_1}{(\semit{V_2}{V_3})}}W$.

  In the forward direction, suppose $\isequiv TV$. Given the structure of $\TT$ and $\VT$, we must have ended that proof with either rule \eqsemi, or rules \eqdistl and \eqassocr. In the case that rule \eqsemi was used, we would get \linebreak$\isequiv {\choice\recordt \ell T L}{\semit{V_1}{V_2}}$ and $\isequiv {T'}{V_3}$. By a coinductive argument, we can derive that $\isequiv{\choice\recordf \ell {\semit{T_\ell}{T'}}L}{\semit{V_1}{(\semit{V_2}{V_3})}}$. On the other hand, if we used rules \eqdistl and \eqassocr, we would also arrive at $\isequiv{\choice\recordf \ell {\semit{T_\ell}{T'}}L}{\semit{V_1}{(\semit{V_2}{V_3})}}$. In either case, we would conclude by coinduction that $\isequiv{\choice\recordt \ell U L}W$, that is, $\isequiv UW$.

  In the converse direction, suppose $\isequiv UW$. Given the structure of $\UT$, we must have ended that proof with rule \eqchoice, which means that $\WT$ is $\choice\recordt \ell W L$ with $\isequiv{U_\ell}{W_\ell}$ for each $\ell\in L$. Thus $\embeds{\semit{V_1}{(\semit{V_2}{V_3})}}{\choice\recordt \ell W L}$. A coinductive argument shows that in this situation, there exist $\typec{V_\ell}$ for $\ell\in L$ with $\isequiv{\semit{V_1}{(\semit{V_2}{V_3})}}{\choice\recordt \ell V L}$ and $\embeds{V_\ell}{W_\ell}$. By coinduction, we then conclude that $\isequiv{\semit{T_\ell}{T'}}{V_\ell}$ for all $\ell\in L$. Then, by rule \eqchoice we have $\isequiv{\choice\recordf \ell{\semit{T_\ell}{T'}}L}{\choice\recordt \ell V L}$. We then have the chain of equivalences
  $$\semit{\choice\recordt \ell T L}{T'}\simeq \typec{\choice\recordf \ell{\semit{T_\ell}{T'}}L} \simeq \choice\recordt \ell V L \simeq \semit{V_1}{(\semit{V_2}{V_3})} \simeq \semit{(\semit{V_1}{V_2})}{V_3},$$
  so that $\isequiv TV$ as desired.

  \item Suppose $\TT$ is $\semit{\XT}{T''}$ and $\VT$ is $\YT$. The proof for $\embeds TU$ must have ended with rule \embsemiid, so we must have $\iseqt X {T'}$ and $\iscontrt{T'}$ and $\embeds{\semit{T'}{T''}}U$ for some $\typec{T'}$. The proof for $\embeds VW$ must have ended with rule \embid, so we must have $\iseqt Y {V'}$ and $\iscontrt{V'}$ and $\embeds{V'}W$ for some $\typec{V'}$.

  In the forward direction, suppose $\isequiv TV$. By examining the structure of $\TT$ and $\VT$, we have three possibilities for the last rule used. The first possibility is that rule \eqneutrtwol was used (\ie $\typec{T''}$ is $\Skip$). Then $\isequiv {T'}V$. We can prove (by coinduction) that if $\embeds{\semit{T'}{\Skip}}U$ then $\embeds{T'}U$. Therefore by coinduction, we have $\isequiv UW$. The second possibility is that rule \eqidsemil was used. In this case we would have $\isequiv{\semit{T'}{T''}}V$. Again by coinduction, we conclude $\isequiv UW$. The third possibility is that rule \eqidr was used. In this case we would have $\isequiv T{V'}$. Again by coinduction, we conclude that $\isequiv UW$.

  In the converse direction, suppose $\isequiv UW$. By coinduction, this implies that $\isequiv{\semit{T'}{T''}}{V'}$. Applying rules \eqidsemil and \eqidr would then enable us to conclude that $\isequiv TV$ as desired.

  \item Suppose $\TT$ is $\semit{(\semit{T_1}{T_2})}{T_3}$ and $\VT$ is $\semit{(\semit{V_1}{V_2})}{V_3}$. The proof for $\embeds TU$ must have ended with rule \embsemisemi, so we must have $\embeds{\semit{T_1}{(\semit{T_2}{T_3})}}U$. Similarly, we must have $\embeds{\semit{V_1}{(\semit{V_2}{V_3})}}W$.

  In the forward direction, suppose $\isequiv TV$. By examining the structure of $\TT$ and $\VT$, we have five possibilities for the last rule used. The first possibility is that rule \eqneutrtwol was used (\ie $\typec{T_3}$ is $\Skip$). Then $\isequiv{\semit{T_1}{T_2}}V$. We can prove (by coinduction) that $\embeds{\semit{T_1}{T_2}}U$. Therefore by coinduction, we have $\isequiv UW$. The second possibility is that rule \eqneutrtwor was used (\ie $\typec{V_3}$ is $\Skip$). The reasoning is analogous. The third possibility is that rule \eqsemi was used. In this case we would have $\isequiv{\semit{T_1}{T_2}}{\semit{V_1}{V_2}}$ and $\isequiv{T_3}{V_3}$. A coinductive argument shows that these two equivalences imply $\isequiv{\semit{T_1}{(\semit{T_2}{T_3})}}{\semit{V_1}{(\semit{V_2}{V_3})}}$. Therefore by coincudtion, we have $\isequiv UW$. The fourth possibility is that rule \eqassocl was used. In this case we get $\isequiv{\semit{T_1}{(\semit{T_2}{T_3})}}{\semit{(\semit{V_1}{V_2})}{V_3}}$. Again by coinduction, we conclude $\isequiv UW$. The fifth possibility is that rule \eqassocr was used. The reasoning is analogous.

  In the converse direction, suppose $\isequiv UW$. By coinduction, this implies $\isequiv{\semit{T_1}{(\semit{T_2}{T_3})}}{\semit{V_1}{(\semit{V_2}{V_3})}}$. We have the chain of equivalences
  $$\semit{(\semit{T_1}{T_2})}{T_3} \simeq \semit{T_1}{(\semit{T_2}{T_3})} \simeq \semit{V_1}{(\semit{V_2}{V_3})} \simeq \semit{(\semit{V_1}{V_2})}{V_3},$$
  so that $\isequiv TV$ as desired.
\end{itemize}
\end{proof}



\section{Proof of \texorpdfstring{$\text{\cref{thm:normalformautomaton}}$}{Theorem 3} (Normal form automata)}
\label{sub:proofs-automatatotypes}

\begingroup
\def\thetheorem{\ref{thm:normalformautomaton}}
\begin{theorem}[Normal form automata]~
\begin{itemize}
	\item Any finite-state automaton can be converted into an equivalent normal form automaton.
	\item Any 1-counter automaton can be converted into an equivalent normal form automaton.
	\item Any pushdown automaton can be converted into an equivalent normal form automaton.
	\item Any decidable language is accepted by a 2-counter normal form automaton.
\end{itemize}
\end{theorem}
\addtocounter{theorem}{-1}
\endgroup

\begin{proof}
  For finite-state automata, there are well-known techniques to convert any automaton into an
  automaton without $\varepsilon$-moves \cite[Section
  2.4]{hopcroftullman:1979}, which is trivially in normal form. For pushdown
  automata, this result is a consequence of a result in Hopcroft and Ullman~\cite[Section 10.3
  and Exercise 10.7]{hopcroftullman:1979}. The remaining two cases have not, to
  the best of our knowledge, been considered in the literature. Note in
  particular that the result for 1-counter automaton does not immediately follow
  from the result for pushdown automaton, since the construction of an
  equivalent pushdown automaton in normal form presented by Hopcroft and Ullman extends the stack alphabet with new symbols.

We begin with the case of decidable languages. Let $L$ be a decidable language. In other words, there is a Turing machine $M$ with two distinguished final states ($q_{\mathrm{accept}}$ and $q_{\mathrm{reject}}$) such that for any word $w$ written in the input tape of $M$,
\begin{center}
if $w\in L$, then $M$ terminates in state $q_{\mathrm{accept}}$; and if $w\not\in L$, then $M$ terminates in state $q_{\mathrm{reject}}$.
\end{center}
By using standard techniques in the theory of Turing machines, we can assume that:
\begin{itemize}
	\item the machine has a read-only input tape, whose head can only move in one direction;
	\item the machine has a single working tape.
\end{itemize}

From this we can construct a new machine $M'$ that `knows' (\ie by a suitable
encoding on its finite control) whether a word $w$ is in $L$ immediately after
reading the last symbol of $w$. The idea is that machine $M'$ stores in the
working tape the contents of the input word $w'$ read thus far; before reading
the next input symbol, $M'$ simulates $M$ for all possible immediate
continuations of $w'$, storing in the finite control which simulations resulted
in acceptance. Given this machine $M'$, we use the construction in Hopcroft and Ullman~\cite[Lemma
7.3]{hopcroftullman:1979} to obtain an equivalent two-stack machine. Because
$M'$ knows whether a word $w$ is in $L$ immediately after reading the last
symbol of $w$, the resulting two-stack machine can be ensured to be in normal
form, \ie it can immediately accept $w$ after reading its last symbol. Finally,
we apply the constructions described in Hopcroft and Ullman~\cite[Lemma 7.4 and Theorem
7.9]{hopcroftullman:1979} to convert this two-stack automaton into an equivalent
four-counter automaton and subsequently a two-counter automaton. All these
constructions essentially simulate a single stack or counter operation by a
sequence of counter operations, and thus they do not interfere with the
semantics of reading moves. In other words, the resulting automata are also
guaranteed to be in normal form.

The only case left is that of one-counter automata, and here the proof is more extensive. Let $A$ be a 1-counter automaton. Our proof will become simpler if we assume that, for every state $q$, $(q,\zerol)$ is a reading mode iff $(q,\succl)$ is a reading mode. $A$ can be converted in this form by creating additional states $q_\varepsilon$ for each state $q$ that can be a reading mode or an $\varepsilon$-mode depending on the value of the counter. Hence, from now on we assume this property of $A$; in particular, we can talk about reading states and $\varepsilon$-states instead of reading modes and $\varepsilon$-modes.

Now observe that each choice of reading state $q$, counter value $n$, and input symbol $a$, defines a unique path from the configuration $(q,n)$ that reads $a$ and either takes $\varepsilon$-moves forever or ends in some reading configuration $(q',n')$. In other words, we can define a function
$$F:Q\times\nbb\times\Sigma\rightarrow\{\mathrm{accept},\mathrm{reject}\}\times(\{\uparrow\}\cup Q\times\nbb)$$
such that $F(q,n,a)$ precisely captures the unique behaviour of the automaton from configuration $(q,n)$ after reading $a$. In particular, the first component of $F(q,n,a)$ is `accept' if the aforementioned path of $\varepsilon$-moves visits some accepting state, and `reject' otherwise. The second component of $F(q,n,a)$ is $\uparrow$ if the aforementioned path is infinite, or $(q',n')$ if it ends in that reading configuration. Notice also that $F$ essentially tells us all we need to know about the automaton, since it specifies how we move from a reading state into the next reading state.

For $(q,n)\in Q\times\nbb$ and $k\in\nbb$, let us use the notation $(q,n)+k$ to denote the configuration $(q,n+k)$. Similarly, when $z$ is either `accept' or `reject', let us use the notation $(z,(q,n))+k$ to denote $(z,(q,n+k))$. The first key idea of the proof is the following characterisation of the sections of $F$.

\begin{claim}
For each reading state $q$ and input symbol $a$, the function $F(q,\cdot,a)$ (with domain $\nbb$) must be one of the following two types:
\begin{itemize}
	\item there exist integers $K,K'$ such that $F(q,n,a)=F(q,n',a)$ for all $n,n'\geq K$ with $n= n'\mod K'$;
	\item there exists an integer $K$ such that $F(q,n,a)=F(q,n',a)+(n-n')$ for all $n\geq n'\geq K$.
\end{itemize}
\end{claim}

To prove this claim, consider the path that starts from state $q$, reading $a$, and continues with $\varepsilon$-moves while taking the branches associated with non-zero counter value.

\begin{itemize}
	\item Suppose this path reaches a reading state $q'$. Then there is a minimal value $K$ such that any configuration $(q,n)$ with $n\geq K$ would follow this path of moves. Hence, the decision between acceptance and rejection would be the same for all such $n$. Moreover, the difference in the counter values at $(q,n)$ and at the end of the path is the same for every such $n\geq K$. Hence, we fit into the second case of the claim.
	\item Suppose this path revisits an $\varepsilon$-state $q'$. Consider the first two occurrences of this state, and the change of the counter value between these two occurrences. If the change is non-negative, \ie the counter value on the second occurrence is not smaller, then all configurations $(q,n)$ with $n$ large enough will lead to infinite looping paths that visit the same set of states. Hence, the decision between acceptance and rejection would be the same, and we fit into the first case of the claim with $K'=1$ (since $n = n' \mod 1$ is trivially satisfied).
	\item Suppose again that this path revisits an $\varepsilon$-state $q'$, but now the change of the counter value between these two occurences is negative. Let $K'$ be the decrease associated with this loop of $\varepsilon$-moves, and let $K$ be a value such that any configuration $(q,n)$ with $n\geq K$ would follow this path of moves up to the second occurrence of $q'$. Then, for any $\ell\geq 0$, the path of moves starting with a configuration of the form $(q,n+\ell K')$ would be identical to the path associated to $(q,n)$, but with an additional intermediate sequence of $\ell$ loops, each of which decrements the counter value by $K'$. Hence, we fit into the first case of the claim.
\end{itemize}

Now that we have our claim proven, we extend the above characterisation in a way that the constants $K,K'$ do not depend on the reading state $q$ or the input symbol $a$. We can simply take $K$ to be the maximum of the corresponding values of $K$ obtained by the claim, and $K'$ to be the lowest common multiple of the corresponding values. Therefore, we have found constants $K,K'$ (that are fixed for a given automaton), such that for every reading state $q$ and input symbol $a$, the function $F(q,\cdot,a)$ fits into one of the above types.

Now that we have a global value of $K'$, we can convert our automaton into an equivalent one, but for which the value of $K'$ may be assumed to be $1$ (\ie the first case of the claim reduces to the statement that $F(q,n,a)=F(q,n',a)$ for all $n,n'\geq K$). To achieve this, we essentially create $K'$ copies of each state of the original automaton; each state is now of the form $(q,k)$, where $q$ encodes the original state and $k$ encodes the current equivalence class modulus $K'$ of the counter value. Any transition that increments or decrements the counter now also moves to the corresponding equivalence class, and for each combination of reading state $q$, input symbol $a$, and equivalence class $k\mod K'$, such that $F(q,\cdot,a)$ is of the first type (according to the original automaton), we change the transition function so that (for the new automaton) $F((q,k),n,a)=F((q,k),n',a)$ for every $n,n'\geq K$.

In summary, we can assume at this stage that our 1-counter automaton $A$ has the following property: there exists a constant $K$ such that, for every reading state $q$ and input symbol $a$, either: $F(q,n,a)=F(q,K,a)$ for every $n\geq K$; or $F(q,n,a)=F(q,K,a)+(n-K)$ for every $n\geq K$. From this assumption, we can now construct our automaton $A'$ in normal form as follows. Essentially, the new automaton $A'$ has a state $q$ for each of the reading states $q$ in $A$. We will create additional states and transitions such that $A'$ postpones the reading of the next symbol, while simulating the computation of $\varepsilon$-moves of $A$.

When reaching a state $q$ that was originally a reading state in $A$, the automaton $A'$ now proceeds by checking whether the counter value $n$ is one of $0,1,\ldots,K-1$, or greater than $K$ (this can be done with a sequence of $K$ states linked by $\varepsilon$-moves). For each case, we can immediately decide whether a given symbol $a$ should lead to an accepting or rejecting state, by looking at the first component of $F(q,\min(n,K),a)$. Moreover, the second component of $F(q,\min(n,K),a)$ is either $\uparrow$ or another configuration $(q',n')$. We can handle the case $\uparrow$ by including a transition to a fresh non-accepting state, for which all transitions are reading self-loop moves. This also makes the automaton guaranteed to read. If, on the other hand, the second component of $F(q,\min(n,K),a)$ is a configuration $(q',n')$, then we can handle this case by adding a sequence of states and $\varepsilon$-transitions that update the counter value accordingly. For each of the finitely many cases where $n<K$, we simply create a sequence of $\varepsilon$-transitions that increment the counter value. For the case that $n\geq K$ and $F(q,\cdot,a)$ is of the first type, we introduce an $\varepsilon$-self-loop that resets the counter to zero, and then a sequence of $\varepsilon$-transitions that end up at $(q',n')$. For the case that $n\geq K$ and $F(q,\cdot,a)$ is of the second type, we create a sequence of $\varepsilon$-transitions that increment the counter by $n'$ and then move to state $q'$. This ensures that the automaton transitions to the configuration $(q',n'+(n-K))$ as desired. Thus, the resulting automaton $A'$ is equivalent to $A$ and is in normal form.
\end{proof}

\section{Proof of \texorpdfstring{$\text{\cref{thm:nstnpstn}}$}{Theorem 8} (Equivalence of pushdown and nested session types)}
\label{sub:contextfreenested}

\begingroup
\def\thetheorem{\ref{thm:nstnpstn}}
\begin{theorem}
Let $\TT$ be a nested session type using type constructors of arity at most $n$. Then, there is a pushdown session type representation for $\TT$ using only $n$ type variables.
\end{theorem}
\addtocounter{theorem}{-1}
\endgroup

\begin{proof}
Given \cref{lem:nst1pst1}, we can assume $n\geq 2$. The proof is significantly more elaborate than the proof of \cref{lem:nst1pst1}. The reason is that general nesting of $n$-ary operators gives rise to an evaluation tree, which is more complex than a sequential composition of operators. However, due to the restricted way in which nesting can occur, it is still possible to represent these evaluation trees using a single stack.

Formally, consider a representation of $\TT$ using the type constructors $\typec{X^{(1)}}$, \ldots, $\typec{X^{(m)}}$, where each type constructor has arity at most $n$. Again, without loss of generality we can assume that the arguments in the equation defining a $k$-ary type constructor are $\typec{\alpha_1},\ldots,\typec{\alpha_k}$ in that order. These equations may have arbitrarily nested expressions on their right-hand side, for example we could have an expression like $\CALLT{X^{(2)}}{\CALL{X^{(2)}}{\alpha_3,\alpha_2},\CALL{X^{(2)}}{X^{(3)},\alpha_1}}$. Let $d$ denote the highest depth of any nested expression appearing in any equation of the representation (the previous example has a depth of $2$).
Let $\E_d$ denote all possible nested expressions of depth at most $d$. This can also be seen as the space of rooted trees of depth at most $d$ where internal nodes (having $k$ children, for $0<k\leq n$) are labelled by a type constructor $\typec{X^{(i)}}$ (of arity $k$), and leaves are labelled either by a variable $\alpha_j$ or a type constructor $\typec{X^{(j)}}$ of arity $0$. Since the depth $d$, the maximum arity $n$ and the number of type constructors $m$ are all finite, the number of such expressions is finite (albeit exponentially large; later we will argue that at most polynomially many trees need to be considered).

Our stack alphabet is then defined to be $\Delta=\E_d\cup\E_d^n$, that is, the union of $\E_d$ with the space of $n$-tuples of expressions in $\E_d$, which is again a finite set. As it shall be seen in the proof, the intuition is that a stack symbol $\sigma\in\E_d$ captures the current, top-level expression, and it can appear only at the top of the stack; whereas a stack symbol $\sigma\in\E_d^n$ captures the possible continuations at the leaves of the tree, and it usually appears below the top level of the stack. Moreover, we will use (pushdown) type constructors $\typec{X_1},\ldots,\typec{X_n}$ in our representation, where $\typec{X_i}$ intuitively means that our evaluation continues with the $i$-th element of the tuple $\sigma$. Additionally, $\typec{X_1}$ has the double duty of unfolding the top-level expression.

Formally, we need to define the right-hand side of equation $\CALLT {X_1} {\sigma S}$ for the case that $\sigma\in\E_d$, as well as the right-hand side of equation $\CALLT {X_j} {\sigma S}$ for $j=1,\ldots,n$ and $\sigma\in \E_d^n$. The remaining cases ($\CALLT {X_j} {\sigma S}$ for $j\neq 1$ and $\sigma\in\E_d$, or $\CALLT {X_j} \varepsilon$) are not of concern, as they will not be reached by our construction; for completeness, we could define the right-hand sides of those cases to be $\End$.

Let us start with the case $\CALLT {X_1} {\sigma S}$ with $\sigma\in\E_d$. $\sigma$ is either a variable $\typec{\alpha_j}$ or a type constructor $\typec{X^{(j)}}$ applied with zero or more subexpressions. If $\sigma$ is a variable $\typec{\alpha_j}$, we pop our stack and continue with type variable $\typec{X_j}$. In other words, we have
$$\CALLT {X_1} {\alpha_j S} \Eq \CALLT {X_j} S.$$
Suppose now that $\sigma$ is a $k$-ary type constructor $\typec{X^{(j)}}$ applied to subexpressions $\sigma_1,\ldots,\sigma_k$. The nested session type representation includes an equation $\CALLT {X^{(j)}}{\alpha_1,\ldots,\alpha_k}\Eq \TT$. We will define
$$\CALLT {X_1} {\sigma S} \Eq \typec{\tilde{T}},$$
where $\typec{\tilde{T}}$ is obtained from $\TT$ by performing an appropriate replacement on all nested expressions appearing in $\TT$. Let $\bar{\sigma}=(\sigma_1,\ldots,\sigma_k,\varepsilon,\ldots,\varepsilon)\in\E_d^n$ denote the $n$-tuple whose first $k$ components are the subexpressions $\sigma_1,\ldots,\sigma_k$, and the remaining $(n-k)$ components are empty subexpressions. Then, every (maximal) nested expression $\sigma'$ appearing in $\TT$ is replaced by the type variable $\CALLT {X_1} {\sigma'\bar{\sigma}S}$. In this way, $\typec{\tilde{T}}$ is a valid right-hand side for pushdown session types representations.

The case $\CALLT {X_j} {\sigma S}$ where $j=1,\ldots,n$ and $\sigma\in \E_d^n$ is straightforward: if $\sigma_j$ is the $j$-th compontent of the tuple $\sigma$, then we define
$$\CALLT {X_j} {\sigma S} \Eq \CALLT {X_1} {\sigma_j S}.$$

Finally, suppose the initial type is given by $\TT = \typec\sigma$. Then, in our pushdown representation,\footnote{Here we are assuming that the depth of $\sigma$ is also at most $d$, the maximum depth of any expression appearing in the equational specification. This is without loss, as otherwise we could just apply the same construction but for a depth $d'$ which is the maximum between $d$ and the depth of $\sigma$.} the corresponding initial type is $\CALLT {X_1} {\sigma}$. To finish the proof, we need to show that the representation described above gives rise to the same type $T$. This can be proven coinductively by observing that both unfoldings follow the same type constructs, and that the stack encoding explained above is enough to express the expressions appearing while unfolding the nested session type representation.
\end{proof}

Although the proof above increases the number of equations by an superexponential multiplicative factor (there are $O((m+n)^{n^d})$ expressions in $\E_d$ and $O((m+n)^{n^{d+1}})$ stack symbols), we remark here that the construction can be altered to have only a polynomial overhead. The reasoning is that there are polynomially many expressions appearing in the original nested session type representation (in fact, linearly in the encoding size of the representation); each of the expressions has polynomially many subexpressions (in fact, linearly in its encoding size). We only need to consider the case that the stack symbol is one of these subexpressions, or the tuple corresponding to its direct children. Thus, at most linearly many equations need to be considered for each equation in the original nested session type representation.

We illustrate the above proof with an example of this simulation. Since the simulation creates a large number of equations, we will only show an initial fragment of the unfolding steps, and a subset of relevant equations as we go along. In our example there are four type constructors $\typec{X^{(1)}}$ (arity 3), $\typec{X^{(2)}}$ (arity 2) and $\typec{X^{(3)}}$ (arity 0). Therefore, the maximum arity is $n=3$. Suppose the nested session type equations are

\begin{align*}
\TT&= \CALLT {X^{(1)}}{X^{(3)},X^{(3)},X^{(3)}}\\
\CALLT {X^{(1)}}{\alpha_1,\alpha_2,\alpha_3}&\Eq \typec{\Oplus\{\keyword{stop}:\End, \keyword{go}: \CALL{X^{(2)}}{\CALL {X^{(2)}}{\alpha_3,\alpha_2},\CALL {X^{(2)}}{X^{(3)},\alpha_1}}\}}\\
\CALLT {X^{(2)}}{\alpha_1,\alpha_2}&\Eq \typec{\Oplus\{\keyword{left}:\alpha_1,\keyword{right}:\alpha_2\}} \\
\typec{X^{(3)}}&\Eq \CALLT {X^{(1)}}{X^{(3)},X^{(3)},X^{(3)}}
\end{align*}
One can see that the maximum depth in subexpressions is $d=2$. By looking at the initial type, our pushdown representation begins with\footnote{Pay close attention to the notation: here $\typec{X_1}$ is the type constructor, and $\CALLT {X^{(1)}}{X^{(3)},X^{(3)},X^{(3)}}$ is the stack symbol.}
$$\TT= \CALLT {X_1} {\CALL {X^{(1)}}{X^{(3)},X^{(3)},X^{(3)}}}.$$ Next, applying the technique in the proof of \cref{thm:nstnpstn}, we see that the pushdown representation would have an equation
\begin{align*}
\CALLT {X_1} {\CALL {X^{(1)}}{X^{(3)},X^{(3)},X^{(3)}} \cdot S} &\Eq \typec{\Oplus\{\keyword{stop}:\End,\keyword{go}:\CALL {X_1} {\sigma' \bar{\sigma} S}\}}\\
\text{with }\sigma' &= \CALLT{X^{(2)}}{\CALL{X^{(2)}}{\alpha_3,\alpha_2},\CALL{X^{(2)}}{X^{(3)},\alpha_1}}\\
\text{and }\bar{\sigma} &= (\typec{X^{(3)}},\typec{X^{(3)}},\typec{X^{(3)}})
\end{align*}

Suppose we unfold the type, and take the branch $\keyword{go}$, that is, we take the transition labeled $\Oplus\keyword{go}$. Then, we arrive at type $$\CALLT {X_1} {\CALL{X^{(2)}}{\CALL{X^{(2)}}{\alpha_3,\alpha_2},\CALL{X^{(2)}}{X^{(3)},\alpha_1}} \cdot (X^{(3)},X^{(3)},X^{(3)})}.$$
From here, we continue by looking at one of the equations corresponding to the type constructor $\typec{X^{(2)}}$. Namely, we observe that the pushdown representation must include the equation
\begin{align*}\CALLT {X_1} {\CALL{X^{(2)}}{\CALL{X^{(2)}}{\alpha_3,\alpha_2},\CALL{X^{(2)}}{X^{(4)},\alpha_1}} \cdot S}&\Eq \typec{\Oplus\{\keyword{left}: \CALL {X_1} {\alpha_1 \bar{\sigma} S}, \keyword{right}: \CALL {X_1} {\alpha_2 \bar{\sigma} S}\}}\\
\text{with }\bar{\sigma}&=(\CALLT{X^{(2)}}{\alpha_3,\alpha_2},\CALLT{X^{(2)}}{X^{(3)},\alpha_1},\varepsilon)
\end{align*}

Suppose we unfold the type, taking the transition $\Oplus\keyword{left}$. We would arrive at the type
$$\CALLT {X_1} {\alpha_1\cdot(\CALL {X^{(2)}}{\alpha_3,\alpha_2},\CALL{X^{(2)}}{X^{(3)},\alpha_1},\varepsilon)\cdot(X^{(3)},X^{(3)},X^{(3)})}.$$
From here the next two steps are straightforward: first, we take the equation $\CALLT {X_1} {\alpha_1 S} \Eq \CALLT {X_1} S$, arriving at the type
$$\CALLT {X_1} {(\CALL{X^{(2)}}{\alpha_3,\alpha_2},\CALL{X^{(2)}}{X^{(3)},\alpha_1},\varepsilon)\cdot(X^{(3)},X^{(3)},X^{(3)})}.$$
Afterwards, we take the equation
$$\CALLT {X_1} {(\CALL{X^{(2)}}{\alpha_3,\alpha_2},\CALL{X^{(2)}}{X^{(3)},\alpha_1},\varepsilon)\cdot S} \Eq \CALLT {X_1} {\CALL{X^{(2)}}{\alpha_3,\alpha_2}\cdot S}$$
to arrive at the type
$$\CALLT {X_1} {\CALL{X^{(2)}}{\alpha_3,\alpha_2}\cdot(X^{(3)},X^{(3)},X^{(3)})}.$$
At this point, we again look at one of the equations corresponding to the type constructor $\typec{X^{(2)}}$, namely
\begin{align*}\CALLT {X_1} {\CALL{X^{(2)}}{\alpha_3, \alpha_2}\cdot S} &\Eq \typec{\Oplus\{\keyword{left}:\CALL {X_1} {\alpha_1 \bar{\sigma} S}, \keyword{right}:\CALL {X_1} {\alpha_2\bar{\sigma} S}\}}\\
\text{with }\bar{\sigma}&=(\typec{\alpha_3}, \typec{\alpha_2},\varepsilon)\end{align*}
Suppose this time we take the transition $\Oplus\keyword{right}$, arriving at the type
$$\CALLT {X_1} {\alpha_2\cdot(\alpha_3, \alpha_2,\varepsilon)\cdot(X^{(3)},X^{(3)},X^{(3)})}$$
The next few steps are again straightforward: we move to
$$\CALLT {X_2} {(\alpha_3, \alpha_2,\varepsilon)\cdot(X^{(3)},X^{(3)},X^{(3)})}$$
and then
$$\CALLT {X_1} {\alpha_2\cdot(X^{(3)},X^{(3)},X^{(3)})}$$
and then
$$\CALLT {X_2} {(X^{(3)},X^{(3)},X^{(3)})}$$
and finally
$$\CALLT {X_1} {X^{(3)}}.$$
Let us take one more transition to conclude the example. Here we would look at an equation corresponding to the constructor $\typec{X^{(3)}}$, namely
$$\CALLT {X_1} {X^{(3)}\cdot S}\Eq \CALLT{X_1}{\CALL{X^{(1)}}{X^{(3)},X^{(3)},X^{(3)}}S}$$
which would take us back to
$$\CALLT{X_1}{\CALL{X^{(1)}}{X^{(3)},X^{(3)},X^{(3)}}}.$$
One can observe that the transitions taken in the pushdown session type representation match exactly the transitions valid for the nested session type representation, so that they are equivalent.



\section{Proof of \texorpdfstring{$\text{\cref{thm:contractivity}}$}{Theorem 9}
  (Decidability of contractivity)}
\label{sub:contractiveness}

\begingroup
\def\thetheorem{\ref{thm:contractivity}}
\begin{theorem}
The following problems are decidable in polynomial time:
\begin{itemize}
\item Given a system $\mathrm{Sys}$ of recursive equations, is $\mathrm{Sys}$ contractive?
\item Given a system $\mathrm{Sys}$ of 1-counter equations, is $\mathrm{Sys}$ contractive?
\item Given a system $\mathrm{Sys}$ of pushdown equations, is $\mathrm{Sys}$ contractive?
\end{itemize}
\end{theorem}
\addtocounter{theorem}{-1}
\endgroup

\begin{proof}
It should be clear that the construction in \cref{sec:systemstoautomata} can be implemented in polynomial time, for each of the three equation schemes. \cref{lem:systemtoautomaton} establishes that $\mathrm{Sys}$ is contractive iff the corresponding automaton is loop-free. The decision problem for loop-freeness is well-known to be decidable. For example, Hopcroft and Ullman~\cite[Section 2.4]{hopcroftullman:1979} describe a procedure to convert a finite-state automaton with $\varepsilon$-moves into an equivalent automaton without $\varepsilon$-moves. Similarly, Valiant~\cite[Lemma 2.4]{valiant:1973:phdthesis} shows how to convert a (deterministic) pushdown automaton into an equivalent loop-free automaton (which therefore works also for 1-counter automata). In both cases, the procedures detect infinite sequences of $\varepsilon$-moves if they exist, and thus, they can be used to decide whether a given automaton is loop-free. It can also be established from those proofs that the running time is polynomial in the description of the automaton.

However, it will be useful for \cref{thm:typeformation} to have algorithms with certain desirable properties; hence, we devote the rest of the proof to presenting an algorithm that decides if a system of pushdown equations is contractive in time $\ocal(|Q||\Delta|)$; since recursive equations and 1-counter equations can be seen as special cases, this also implies a polynomial time algorithm for those systems.

Our algorithm works as follows. Let us say that an equation is \emph{trivial} if its right-hand side is a type constructor, e.g.\ $\CALLT X \varepsilon \doteq \CALLT Y \alpha$ or $\CALLT X {\alpha S} \doteq \CALLT Y S$. Suppose the system is not contractive. Then there exists an infinite sequence $\CALLT{X_1}{\omega_1}$, $\CALLT{X_2}{\omega_2}$, \ldots of type identifiers obtained by following trivial equations. By looking at the stack lengths along this sequence, one of two properties must hold: either some stack length is repeated infinitely often, or each stack length value occurs only finitely many times in the sequence.

Suppose some value of the stack length is repeated infinitely often. Let $n$ be the minimum such value (for simplicity, let us assume $n>0$; the argumentation for $n=0$ is essentially the same). There is some order $N$ such that $|\omega_i|\geq n$ for every $i\geq N$. In particular, this implies that all such $\omega_i$ for $i\geq N$ coincide on the bottom $n-1$ stack symbols. Let $\omega$ be the bottom substack of size $n-1$ of these $\omega_i$. Now, consider all the type identifiers $\CALLT {X_i} {\omega_i}$ for which $|\omega_i|=n$ and $i\geq n$, which by assumption occur for infinitely many $i$. Notice that all such $\omega_i$ coincide except possibly on their topmost symbols. Since there are finitely many possible type variables in the system, and finitely many stack symbols, an immediate application of the pigeonhole principle yields that there must be a type constructor $\XT$ and a stack symbol $\sigma$ such that $\typec{X_i} = \XT$ and $\omega_i=\sigma\omega$ for at least two different $i$ (in fact, for infinitely many $i$). This means that, if we were to start a sequence of type identifiers with $\CALLT X \sigma$, we would eventually return to $\CALLT X \sigma$ (as we would never have to observe the contents of the stack below this position), and thus be stuck in an infinite loop. Similarly, if $n=0$ was the minimum such value, then some type identifier $\CALLT X \varepsilon$ is repeated at least twice (in fact, infinitely often), and $\CALLT X \varepsilon$ would accordingly be the start of an infinite sequence.

Now suppose each stack length occurs only finitely many times in the sequence. For each $n$, let $i_n$ denote the last time that the stack length is $n$. Notice that $i_n$ is defined for all $n$ sufficiently large. Via the same pigeonhole principle as before, there must exist some type constructor $\XT$ and stack symbol $\sigma$ such that $\typec{X_{i_n}}=\XT$ and $\sigma$ is the top symbol of $\omega_{i_n}$ for at least two different $n,n'$ (in fact, for infinitely many $n$). Assuming without loss that $n<n'$, and since $i_n$ was the last time that the stack length was $n$, there must be stack words $\omega,\omega'$ of size $n-1$, $n'-n$ respectively such that $\omega_{i_n}=\sigma w$ and $\omega_{i_{n'}}=\sigma\omega'\omega$. Then, if we were to start with $\CALLT X \sigma$, we would eventually arrive at $\CALLT X {\sigma\omega'}$ (without ever having looked at an empty stack) and thus we would also get an infinite sequence of type identifiers.

In conclusion, a system is non-contractive iff one of the following properties hold:
\begin{itemize}
\item there is a variable $\XT$ such that, when starting from $\CALLT X \varepsilon$ and following trivial rewriting rules, one returns to $\CALLT X \varepsilon$;
\item there is a variable $\XT$ and a stack symbol $\sigma$ such that, when starting from $\CALLT X \sigma$ and following trivial rewriting rules, one reaches $\CALLT X {\sigma\omega}$ for some (possibly empty) stack contents $\omega$, without ever observing an empty stack.
\end{itemize}

Next, we show that these properties can be decidable in polynomial time. Here we only sketch the algorithm for the case $\CALLT X \sigma$ as the case $\CALLT X \varepsilon$ is similar. Starting from $\CALLT X \sigma$, we construct a sequence of type identifiers $\CALLT {X_i} {\omega_i}$ by following the rewriting rules, and additionally applying some shortcutting which we describe in a moment. At stage $i$, we have produced the type identifier $\CALLT {X_i} {\omega_i}$; let $\sigma_i$ be the top stack symbol of $\omega_i$. We produce the next term in the sequence as follows.

\begin{itemize}
	\item Suppose $\sigma_i=\varepsilon$, \ie we have reached the empty stack. Then we can immediately end the procedure and exclude $\CALLT X \sigma$ from being the start of an infinite sequence that never observes an empty stack.
	\item Suppose $\CALLT {X_i} {\sigma_i}$ is the left-hand side of a non-trivial equation. Then we can immediately end the procedure, and correctly assert that no infinite sequence of rewriting rules could start with $\CALLT X \sigma$.
	\item Suppose $\CALLT {X_i} {\sigma_i}$ is the left-hand side of a trivial equation (with $\sigma_i\neq\varepsilon$), and that this is the first time in the sequence that we have observed this combination of variable and top stack symbol. Then we continue with $\CALLT {X_{i+1}} {\omega_{i+1}}$ by following the appropriate rewriting rule.
	\item Suppose $\CALLT {X_i} {\sigma_i}$ is the left-hand side of a trivial equation (with $\sigma_i\neq\varepsilon$), but it repeats a previous combination. That is, we have previously seen $\CALLT {X_{i'}} {\omega_{i'}}$ with $\typec{X_{i'}}=\typec{X_i}$ and $\sigma_i$ being the top stack symbol of $\omega_{i'}$. Let $i'$ be the last time this combination was observed. Now suppose that, between $i'$ and $i$, the stack length has never dropped below $|\omega_{i'}|$; in particular, this implies that $|\omega_{i'}|\leq|\omega_i|$. Then, we can immediately end the sequence and correctly assert that $\CALLT {X_i}{\sigma_i}$ (and hence, also $\CALLT X \sigma$) leads to an infinite sequence of rewriting rules.
	\item Suppose again that $\CALLT {X_i} {\sigma_i}$ is the left-hand side of a trivial equation (with $\sigma_i\neq\varepsilon$), whose combination was seen before for the last time at $\CALLT{X_{i'}}{\omega_{i'}}$. But now suppose that, at some point between $i'$ and $i$, the stack length has dropped below $|\omega_{i'}|$. Let $j$ be the very first time it did so, which means that $i'<j\leq i$ and in particular that $\omega_j$ is the stack obtained by popping the top symbol $\sigma$ from $\omega_{i'}$. Now, if we were to follow the rewriting rules from $\CALLT{X_i}{\omega_i}$, we would essentially repeat the same rules from $\CALLT {X_{i'}}{\omega_{i'}}$ to $\CALLT{X_j}{\omega_j}$. Hence, we can shortcut this part and produce the next term $\CALLT{X_{i+1}}{\omega_{i+1}}$ as $\typec{X_{i+1}}=\typec{X_j}$ and $\omega_{i+1}$ as $\omega_i$ without the top symbol $\sigma_i$. Notice that each time we apply this shortcut, the stack length necessarily decreases.
\end{itemize}

Next, we need to show that our procedure eventually ends after polynomially many iterations. We can terminate our procedure either by observing a non-trivial equation, or by observing an empty stack, or by repeating a previous combination of variable and top stack symbol that provably yields an infinite sequence. Notice that, each time a combination of variable and top stack symbol is repeated, either we terminate the procedure (asserting the existence of an infinite sequence) or the next term in the sequence has a smaller stack length. Hence, we can only increase the stack length at most $|Q||\Delta|-1$ times (once for each different combination of variable and top stack symbol, excluding the initial term). Therefore, if after $2|Q||\Delta|$ iterations we have not yet ended the procedure, then we must reach an empty stack. Thus the number of iterations is $\ocal(|Q||\Delta|)$. We also need to argue that each iteration can be done in constant time.
For most cases, this is obvious, as we merely have to lookup the right-hand side of the equation corresponding to the currently known symbol. The only case which is not obvious is when we apply the shorcutting rule. But here, we can use a lookup table that saves, for each combination $\CALLT X \sigma$, whether the sequence starting from $\CALLT X \sigma$ eventually drops to an empty stack, and if so, to which variable does it lead. We can do this by, at each iteration and for each $n$, keeping track of the set of configurations $\CALLT X \omega$, such that $\omega$ has size $n$ and the stack length has not dropped below $n$ since $\CALLT X \omega$ was last visited.

To conclude, we can decide whether a system is contractive by applying the above procedure for every choice of $\CALLT X\varepsilon$ and $\CALLT X \sigma$. As there are only $|Q|(|\Delta|+1)$ (\ie polynomially many) choices, the total running time would be $\ocal(|Q|^2|\Delta|^2)$. This can be further reduced to $\ocal(|Q||\Delta|)$ since, after following the procedure above for a given choice of $\CALLT X \sigma$, we can infer whether $\CALLT {X'}{\sigma'}$ leads to an infinite sequence for all the $\CALLT {X'}{\sigma'}$ that were visited, and thus exclude them from our subsequent analysis. In other words, we can keep a lookup table that says, for every $\CALLT X \sigma$, whether it has been discovered before, and prune our search the next time we revisit $\CALLT X \sigma$.
\end{proof}

\end{document}